\documentclass[preprint,authoryear,12pt]{elsarticle}

\usepackage[ruled,vlined]{algorithm2e}

\usepackage{helvet}
\usepackage{courier}
\usepackage{mathrsfs}
\usepackage{graphicx}
\usepackage{amsmath}
\usepackage{amssymb}
\usepackage{amsthm}
\usepackage{multirow} 
\usepackage{hyperref}
\usepackage{subcaption}
\usepackage{tikz}
\usetikzlibrary{arrows, petri,topaths}
\usepackage{tkz-berge}
\newtheorem{theorem}{Theorem}
\newtheorem{corollary}[theorem]{Corollary}
\newtheorem{proposition}[theorem]{Proposition}

\newtheorem{lemma}[theorem]{Lemma}

\newtheorem{observation}[theorem]{Observation}

\theoremstyle{definition}
\newtheorem{example}{Example}
\newtheorem{definition}{Definition}

\DeclareMathOperator{\argmin}{argmin}
\DeclareMathOperator{\argmax}{argmax}

\newenvironment{rlemma}[1]{\medskip\sc{Lemma~\ref{#1}.}\begin{itshape}}{\end{itshape}}
\newenvironment{rproposition}[1]{\medskip\sc{Proposition~\ref{#1}.}\begin{itshape}}{\end{itshape}}

\newcommand{\eps}{\epsilon}
\renewcommand{\vec}[1]{\mathbf{#1}}

\newcommand{\coursename}{(67686) Mathematical Foundations of AI}

\newcommand{\handout}[5]{
   \renewcommand{\thepage}{#1-\arabic{page}}
   \noindent
   \begin{center}
   \framebox{
      \vbox{
    \hbox to 5.78in { {\bf \coursename}
         \hfill #2 }
       \vspace{4mm}
       \hbox to 5.78in { {\Large \hfill #5  \hfill} }
       \vspace{2mm}
       \hbox to 5.78in { {\it #3 \hfill #4} }
      }
   }
   \end{center}
   \vspace*{4mm}
}

\def\ol{\overline}

\newcommand{\ceil}[1]{\left\lceil #1 \right\rceil}
\newcommand{\floor}[1]{\left\lfloor #1 \right\rfloor}

\newenvironment{proof-sketch}{\noindent{\bf Sketch of Proof}\hspace*{1em}}{\qed\bigskip}
\newenvironment{proof-idea}{\noindent{\bf Proof Idea}\hspace*{1em}}{\qed\bigskip}
\newenvironment{proof-of-lemma}[1]{\noindent{\bf Proof of Lemma #1}\hspace*{1em}}{\qed\bigskip}
\newenvironment{proof-attempt}{\noindent{\bf Proof Attempt}\hspace*{1em}}{\qed\bigskip}

\makeatletter
\def\fnum@figure{{\bf Figure \thefigure}}
\def\fnum@table{{\bf Table \thetable}}
\long\def\@mycaption#1[#2]#3{\addcontentsline{\csname
  ext@#1\endcsname}{#1}{\protect\numberline{\csname
  the#1\endcsname}{\ignorespaces #2}}\par
  \begingroup
    \@parboxrestore
    \small
    \@makecaption{\csname fnum@#1\endcsname}{\ignorespaces #3}\par
  \endgroup}
\def\mycaption{\refstepcounter\@captype \@dblarg{\@mycaption\@captype}}
\makeatother

\newcommand{\mathify}[1]{\ifmmode{#1}\else\mbox{$#1$}\fi}
\newcommand{\bigO}O

\renewcommand{\vec}[1]{{\mathbf #1}}

\newcommand{\remove}[1]{{}}

\newcommand{\xqed}{\mbox{\raggedright $\Diamond$}}

\newcommand{\step}[1]{\stackrel{{\scriptscriptstyle{#1}}}{\rightarrow}}

\newcommand{\newpar}[1]{
\vspace{-1mm}
\paragraph{#1}}
\newcommand{\newsubsec}[1]{
\vspace{-1.5mm}
\subsection{#1}}

\newcommand{\omittext}[1]{}

\def\shortcite{\citeyearpar}
\def\cite{\citep}

\newcount\Comments % 0 suppresses notes to selves in text
\definecolor{darkgreen}{rgb}{0,0.6,0}
\newcommand{\kibitz}[2]{\ifnum\Comments=1{\color{#1}{#2}}\fi}
\newcommand{\rmr}[1]{\kibitz{blue}{[RESHEF:#1]}}
\newcommand{\olev}[1]{\kibitz{red}{[OMER:#1]}}

\begin{document}

\begin{frontmatter}
\title{A Local-Dominance Theory of Voting Equilibria}
% maybe: From best-response to Local Dominance: A new theory of voting equilibria 

% Title portion
\author[harvard]{Reshef Meir}
\author[huji]{Omer Lev}
\author[huji]{Jeffrey S. Rosenschein}
\address[harvard]{Harvard University}
\address[huji]{Hebrew University of Jerusalem}

\begin{abstract}
It is well known that no reasonable voting rule is strategyproof. Moreover, the common Plurality rule is particularly prone to strategic behavior of the voters and empirical studies show that people often vote strategically in practice. Multiple game-theoretic models have been proposed to better understand and predict such behavior and the outcomes it induces. However, these models often make unrealistic assumptions regarding voters' behavior and the information on which they base their vote. 

We suggest a new model for strategic voting that takes into account voters' bounded rationality, as well as their limited access to reliable information. We introduce a simple behavioral heuristic based on \emph{local dominance}, where each voter considers a set of possible world states without assigning probabilities to them. This set is constructed based on prospective candidates' scores (e.g., available from an inaccurate poll). In a \emph{voting equilibrium}, all voters vote for candidates not dominated within the set of possible states. 

 We prove that these voting equilibria exist in the Plurality rule for a broad class of local dominance relations (that is, different ways to decide which states are possible). Furthermore, we show that in an iterative setting where voters may repeatedly change their vote, local dominance-based dynamics quickly converge to an equilibrium if voters start from the truthful state. Weaker convergence guarantees in more general settings are also provided. 

Using extensive simulations of strategic voting on generated and real preference profiles, we show that convergence is fast and robust, that emerging equilibria are consistent across various starting conditions, and that they replicate widely known patterns of human voting behavior such as Duverger's law. Further, strategic voting generally improves the quality of the winner compared to truthful voting.
\end{abstract}

%
%\category{I.2.11}{Artificial Intelligence}{Distributed Artificial Intelligence}[Multiagent Systems]
%
%\terms{Theory, Algorithms, Economics}

\begin{keyword} 
Voting equilibrium \sep Strict uncertainty \sep Local dominance \sep Strategic voting
\end{keyword}
\end{frontmatter}

\section{Introduction}

It is often argued that people vote ``strategically'', by trying to promote the election of preferable candidates. Game-theoretic considerations have been applied to the study and design of voting systems for centuries, but the question of how people vote, or should vote, is still open. Suppose that we put aside the complications involved in political voting,\footnote{For example, social utilities~\cite{manski1993identification,brock2001discrete}, strategic candidates~\cite{calvert1985robustness,feddersen90}, and other considerations (see e.g. \cite{riker1968theory,edlin2007voting}.} and focus a simple scenario that fits all the ``standard'' assumptions: Each of $n$ voters has complete transitive preferences $\prec_i$ over a fixed set of alternatives $M$, and each voter's only purpose is to bring about the election of her most-favorable alternative. We will further restrict ourselves to discussing the common Plurality rule, where the alternative with the maximal number of votes is the winner. This scenario translates naturally to a \emph{game}, in which the actions of each voter are her possible ballots---voting for one of the alternatives, in case of Plurality. One might expect game theory to give us a definitive answer as to what would be the outcome of such a game. 

However, an attempt to apply the most fundamental solution concept, Nash equilibrium, to the scenario above, reveals a disappointing fact: Almost \emph{any} profile of actions is a pure Nash equilibrium, and in particular every alternative wins in some equilibrium, even if this alternative is least-preferred by all voters.\footnote{Assuming there are at least three voters.} This observation triggered a search for more appropriate solution concepts for voting games. These concepts rely on taking into account various additional factors, such as the information available to the voters, collusion and group behavior, and intrinsic preferences towards certain actions. Some solutions focused on variations of the standard single-shot game, for example when voters vote sequentially rather than simultaneously.
	
\newpar{Strategic voting}
The underlying assumption of game-theoretic analysis is that players are engaged in strategic behavior. But what does it mean to vote ``strategically''?	
Fisher~\shortcite{fisher2004definition} offers the following definition for what he calls ``tactical voting'':\footnote{Some authors distinct between tactical and strategic voting, especially in political settings. For our purpose they are the same.}
\begin{quote} 
\emph{A tactical voter is someone who votes for a party they believe is more likely to win than their preferred party, to best influence who wins in the constituency.}
\end{quote}
	
The two key components of this definition are \emph{belief} and \emph{influence}. Models of strategic voting differ in how they interpret these terms when considering the behavior of a voter.
	
\begin{example}\label{ex1}
As a running example, we consider a profile with 3 candidates $M=\{a,b,c\}$. Suppose that there are 100 voters, and that currently votes are divided as: $45$ for $a$, $40$ for $b$, and $15$ for $c$. Among the supporters of $c$ are voters $v$ and $v'$. Voter $v$ has preference $c \succ b \succ a$, whereas voter $v'$ prefers $c\succ a \succ b$.
\end{example}
While if truthful, both $v,v'$ would stay with $c$, it seems that $c$ has no chance of winning, and thus a wise strategic decision for $v$ would be to \emph{change} her vote to $b$.  Similarly, $v'$ may prefer to vote for $a$. %Thus the response of voter $v$ in such a state would be to \emph{change} her vote to $b$. 
A voter's \emph{response function} would dictate what the voter would do in any given state.
	
Once we define a voter's response function, this immediately induces an \emph{equilibrium model} (or a solution concept): Voting equilibria are simply outcomes where the response of every voter is her current action.
	
By applying the reasoning above to all supporters of $c$ in the example, we would expect to eventually reach an equilibrium where only $a$ and $b$ get votes. The phenomenon that under the Plurality rule almost all votes divide between two candidates is well known in political science, and is called Duverger's Law~\cite{duverger1954political}. 
		%A good model should therefore allow such strabehavior, and predict outcomes that are consistent with it.

\newpar{Our contribution}
	%In this work we put forward a new response model by which a voter decides which candidate ``is more likely to win'', based on simple decision theoretic rules, and without using probabilities. As a result we get a new notion of voting equilibrium, where all voters believe they vote for their best option. 
	%It seems ``reasonable'' that a $c$ voter would change his vote to either $a$ or $b$.
	%In this work we want to argue that all voting models suggested thus far fall short of meeting basic requirements from a solution concept. 
After enumerating the desiderata we believe should guide the search for a proper solution concept, we review some solutions that have been proposed in the literature, and explain where they fall short of meeting these requirements. We then lay out our epistemic framework, which is a non-probabilistic way of capturing uncertainty. Using this framework and simple behavioral assumptions, we present our response function and equilibrium concept. In the remainder of the paper, we will argue, using formal propositions and empirical analysis, that our solution is indeed the appropriate one for Plurality voting. In particular we show that voters who start from the truthful vote will quickly converge to a pure equilibrium, and that convergence is likely to occur even from arbitrary initial states. We show that in various voter distribution models, using the local-dominance framework enables equilibria with desirable properties --- with ``better'' winners and a ``Duverger-like'' stable states. %Moreover, in complex scenarios for which there is no clear theoretical result on voter behavior, applying our framework resulted in stable equilibria with ``desirable'' winners and real-world resembling voter distribution. %\rmr{add key empirical findings} \olev{Added some. YMMV :-) }

\section{Desiderata for Voting Models}
\label{sec:desiderata}
%What are good voting models? We will restrict ourselves to discuss Plurality voting. Plurality is one of the most common and most simple voting rule, yet the behavior of voters under Plurality is poorly understood.
We now present some arguably-desirable criteria for a theory of voting. We will not be too picky about what is considered a voting model, and whether it is described in terms of individual or collective behavior. The key feature of a model is that given a profile of preferences, it can be mapped to a set of outcomes (i.e., of possible or likely voting profiles under the Plurality rule). 
We classify desirable criteria to the following classes: Theoretic (mainly game-theoretic), behavioral, and scientific.
\rmr{This section may be too long}

\newsubsec{Theoretic Criteria}
%By ``theoretic'' we mean mainly game-theoretic.

\newpar{Rationality} The model should assume that voters are behaving in a rational way, in the sense that they are trying to maximize their own utility based on what they know and/or believe.

\newpar{Equilibrium} The model predicts outcomes that are in equilibrium, for example, a refinement of Nash equilibrium, or of another popular solution concept from the game theory literature. It is more appealing if equilibrium can be naturally computed or even reached by some natural dynamic (similar to the convergence of best-response dynamics to a pure Nash in congestion games).

\newpar{Discriminative power} The model predicts a small but non-empty set of possible outcomes (sometimes called predictive power). More specifically, it predicts a small set of possible winners, as there may be multiple voting profiles that have the same winner. 

\newpar{Broad scope} The model applies (or can be easily adapted) to various scenarios such as simultaneous, sequential or iterative voting, and to the use of different voting rules.

\medskip
In addition, we put forward two less formal requirements, that are nevertheless important. First, that predicted outcomes should not include outcomes that are obviously unreasonable or absurd. Second, we would like our model to be grounded in familiar concepts from decision theory, game theory, and voting theory; it will thus be easier understand, and to compare with other models.

\newsubsec{Behavioral criteria}
By behavioral criteria, we mean what implicit or explicit assumptions the model makes on the behavior of voters in the game.

\newpar{Voters' knowledge} Voters' behavior in the model should not be based on information that they are unlikely to have, or that is hard to obtain. 

\newpar{Voters' capabilities} The decision of the voter should not rely on complex computations, non-trivial probabilistic reasoning, etc.

\medskip
In addition, we would like the behavioral assumptions, whether implicit or explicit, to be supported by (or at least not to directly contradict) studies in human decision making. % (this dovetails with the reproduction and experiments scientific criteria below). 
	
\newsubsec{Scientific criteria}

\newpar{Robustness} We expect the model to give similar predictions even if some voters do not exactly follow their prescribed behavior, if we slightly modify the available information, if we change order of players, etc. Except in a few threshold cases, we would not expect every small perturbation to change the identity of the winner.
	
\newpar{Few parameters} If the model has parameters, we would like it to have as few as possible, and we would like each of them to be meaningful (e.g., voters' memory). A related requirement is that the model will be \emph{easy to fit}, given empirical data on voting behavior. 
	
\newpar{Reproduction} When simulating the model on generated or real preferences, we would like it to reproduce common phenomena such as Duverger's law.
	
\newpar{Experiments} The hardest test for a model is to try and predict the behavior of human voters based on their real preferences. By comparing the predicted and real votes (or even just outcomes), we can measure the accuracy of the model. %While experiments are standard evaluation technique for predictive models, it is hard to conduct proper voting experiments, controlling both the preferences of the voters and the information available to them.
	
	%\medskip It should be noted that not all preference profiles are equally important. We mainly want our model to provide reasonable predictions on profiles that are common or likely. Also, we will allow more ambiguity on profiles where where there is a weaker consensus regarding the ``right'' outcome (for measures of consensus in voting profiles, see \cite{meskanen2006distance,elkind2010role}).

\medskip
The behavioral criteria together with the rationality requirement can be thought of as criteria of \emph{bounded rationality}. 

Lastly, Some voting models explain how strategic behavior is \emph{better for society}. For example, equilibrium outcomes in Plurality with a particular voter behavior may have a better Borda score, or coincide more with choosing a Condorcet winner. %Similarly, there are cases where truthful voting results in a \emph{paradox}---a winner that is clearly bad for society~\cite{xia2011strategic}. A model of strategic voting may show how voters overcome, or are still prone to, similar paradoxes. 
\rmr{removed the paradoxes part for now, as we do not deal with them yet}
Although this is not exactly a criterion for a good model (a real strategic behavior may not increase welfare), we are interested in the conditions under which the theory predicts an increased welfare, as these may be useful for design purposes. 

\section{Literature Review}

This section is a critical review of some prominent models from the literature of voting under the Plurality rule. For each model, we point out some criteria by which it excels or fails. This is not an exhaustive list, and our purpose is not to criticize other authors, but rather to identify pitfalls of which to beware, and also to identify the positive properties that we would like any new model to preserve. Outside the scope of this review are theories that take into account social utilities, sense of duty, and other incentives that do not follow directly from the voter's preferences (such as the Calculus of Voting~\cite{riker1968theory}). 

%Some of the models deal with settings different from those that we will consider (for example, where collusion is possible), and will thus receive less attention.

\newpar{The Leader Rule}
Before we start, we would like to highlight a model that fairs nicely in almost all of the above criteria, which is Laslier's \emph{leader rule} for Approval voting~\cite{laslier2009leader}. This is a simple parameter-free behavioral strategy, where a voter only needs to take into account a prospective ranking of the candidates (which can be available from a poll, a prior belief, or a previous voting round).\footnote{According to the leader rule, the voter approves all candidates that are preferred to the prospective leader, and approves the leader iff it is preferred to the runnerup.} The leader rule is behaviorally plausible, has attractive theoretical properties, makes minimal informational assumptions, and seems to explain human voting behavior at least in some contexts. Unfortunately, it does not seem to have a natural extension to Plurality voting. 

%For example, the \emph{Leader rule} of Laslier~\shortcite{laslier2009leader} is a strategic model that fares well in all of the above criteria. Unfortunately, it was defined only for Approval voting, and does not have a natural extension to Plurality.

\newsubsec{Complete information}

The basic notion of a pure Nash equilibrium (PNE) in a normal form voting game is effectively useless, as almost any outcome (even one where all voters vote for their worst candidate) is a PNE. Consider voter $v$ in Example~\ref{ex1}. She is powerless to change the outcome, and therefore has no incentive to change her vote. 
Two refinements that have been suggested rely on plausible behavioral tendencies. 

\olev{These models, like the iterative model, are not complete information models. Why are they here? Shouldn't they move to the partial information section?} \rmr{they are complete info. a votes known exactly how others vote in eq.}\olev{Doesn't full info usually mean knowing everyone’s preferences, so I can tell who will manipulate how? I realize you mean full knowledge of the public knowledge (what's the final tally), but I think that's slightly different than complete knowledge. Maybe ``Certainty'' is a better term?} \rmr{think about you work with Dave. you took a game  defined by voters incentive. Now forget about that being a voting setting. You calculate all Nash equilibria of this normal-form game. There is no partial information, information sets, etc. Hence a public information game.}

\newpar{Truth bias} A truth-biased voter gains some negligible additional utility from reporting his true preferences (i.e., his top candidate)~\cite{meir2010convergence,dutta2012nash}. He will therefore be truthful, unless he can strictly gain by voting for a different candidate. Dutta and Sen focused on implementation rather than on how truth bias affects a particular voting rule. Nash equilibria under Plurality with truth-biased voters have been studied empirically by Thompson et al.~\shortcite{thompson2013empirical}, and analytically by Obraztsova et al.~\shortcite{obraztsova2013plurality}. Indeed, truth-bias significantly reduces the number of pure Nash equilibria (sometimes to zero), and in particular eliminates many unreasonable equilibria such as those where all voters vote for their least-preferred candidate. However if there is gap of at least 2 votes between the winner and the other candidates in the \emph{truthful profile}, then it is a Nash equilibrium with or without truth bias. This is the case in Example~\ref{ex1}, and clearly $v$ would still vote for $c$ under truth bias.

\newpar{Lazy bias} In many voting settings the voting action itself incurs some small cost or inconvenience to the vote. The conclusion that a ``rational'' voter would often rather abstain (as she is rarely pivotal) is typically referred to in the literature as the ``no-vote paradox'', see e.g.~\cite{downs1957economic,owen1984vote}. 
When voting is presented as a normal form game, we can add abstention as an additional allowed action. 
 A ``lazy'' voter would thus choose to abstain if she cannot affect the outcome. Pure Nash equilibria with lazy voters were studied, for example, in~\cite{desmedt2010equilibria}. These are typically highly degenerated voting profiles, where all voters except one abstain.

There are numerous other models that have been suggested for voting behavior with complete information. Solution concepts vary and include collusion~\cite{sertel2004strong,falik2012coalitions}, iterated removal of dominated strategies~\cite{dhillon04}, and specific decision diagrams tailored for three candidates~\cite{niemi1982sophisticated,felsenthal1988tacit}.
Crucially, most of these models assume that (apart from having access to the full preference profile) voters engage in complicated equilibrium computations, perform unlimited steps of iterated reasoning, and so on.
 That said, models that have been crafted based on empirical observations do manage to replicate interesting phenomena such as Duverger's law and the election of Condorcet winners~\cite{felsenthal1988tacit}.

\newsubsec{Voting under uncertainty}

Uncertainty partly solves the the problem of equilibria explosion, since any voter can become pivotal with some probability, and therefore cares about whom to vote for. In Example~\ref{ex1}, suppose that $v$ is unsure about the accurate gap between $a$ and $b$, but knows that it might be zero (whereas the gap between $a$ and $c$ is expected to be much larger). Then the rational thing would be to vote for $b$, since with positive probability the outcome for $v$ will improve. Uncertainty have also been proposed as a partial solution to the no-vote paradox~\cite{owen1984vote}.

One model that introduces uncertainty is trembling-hand perfection~\cite{messner02}, where each vote may be miscounted with some negligible probability. Messner and Polborn also assume some degree of voters' coordination, and prove that in every equilibrium in their model only two candidates get votes. That is, the model predicts a very strong version of Duverger's law.  They provide some additional results for the three-candidate case. 

Myerson and Weber's~\shortcite{myerson93} \emph{theory of voting equilibria} provides a different model of uncertainty, one that is closer to our approach. % yet applies probabilities and expected utility calculation.
 An \emph{outcome} in this model is represented---as in our model---by a prospective vector of candidates' scores. From this vector they derive a distribution over possible scores, and the outcome is an \emph{equilibrium} if there is a voting profile supporting it, s.t. every voter is maximizing her expected utility w.r.t.\ this distribution.
%Utility computations are somewhat simplified by showing that the voter only needs to consider the probability of a tie between any pair of candidates. 
Myerson and Weber prove that an equilibrium always exists for a broad class of voting rules including Plurality. Focusing on a few examples with three candidates under Plurality, they show that their model gives reasonable results, and that some equilibria replicate Duverger's law. 

While the Myerson and Weber model is highly attractive in many respects, it suffers from some severe shortcomings.
One drawback is that while an equilibrium exists, this is proved by a non-constructive (if elegant) fixed-point argument, and it is not clear how to compute such an equilibrium---let alone how the voters are supposed to find it. 

Another restriction  is that the model is only defined for non-atomic voters, whose influence on the outcome is infinitesimally small.

However the main issue we find  problematic is that voters must engage in non-trivial probabilistic reasoning, even if just to verify that they are playing an equilibrium strategy. The assumption that voters maximize some expected utility function is of course not limited to the Myerson and Weber paper, and is prevalent in the political science literature, see e.g.~\cite{silberman1975rational,palfrey1983strategic,alvarez2000new}, as well as in \cite{messner02} which was mentioned above. 

From a behavioral perspective such an assumption weakens the model, as people are notoriously bad at estimating probabilities, and are known to employ various heuristics instead~\cite{tversky1974judgment}.

An additional disadvantage of the expected utility maximization approach, is that voters preferences must be cardinal and cannot be described as a permutation over candidates. 

\newpar{Strict uncertainty} Voting with strict uncertainty (without probabilities) was considered by Ferejohn and Fiorina~\shortcite{ferejohn1974paradox}, who assumed voters each try to minimize their \emph{maximal regret} over all possible outcomes. However their model (like probability-based models) heavily relies on voters having cardinal utilities. Also, they take an extreme approach where voters do not use \emph{any} available information (similarly to the dominance-based approach in \cite{dhillon04}), and thus all states are considered possible (as also pointed out in a critique on the minimax approach by Aldrich~\shortcite{aldrich1993rational}). Another regret-based model was suggested in \cite{merrill1982strategic}.

A different approach to strict uncertainty was taken by Conitzer et al.~\shortcite{conitzer2011dominating} who considered manipulations that are weakly helpful in all possible states and strictly helpful in some. %\olev{This is probably overkill, but Hazon and Elkind \cite{HE10} extend Slinko and White \cite{SW08} on ``safe manipulations'' -- when you don't know how many people would manipulate as you do, what is a safe manipulation (some manipulations can be harmful if everybody uses them. Kant would be angry...)} \rmr{Safe manipulation are more about coordination than about uncertainty. Also if I recall there is no equilibrium analysis there}
Reijngoud and Endriss~\shortcite{reijngoud2012voter} looked at a special case where the set of possible states is derived from a poll. Our notion of local dominance is based on similar reasoning, %, where the set of possible states relies on a distance metric.
 however the mentioned papers restricted their analysis to a single manipulator, whereas we study equilibria. A similar notion of equilibrium under strict uncertainty was defined (without existence or convergence results) by van Ditmarsch et al.~\shortcite{van2013strategic}. We elaborate on the similarities an differences between these models and ours in Section~\ref{sec:epistemic}.

\newsubsec{Iterative and sequential games}

In sequential voting games voters report their preferences one at a time, where every voter can see all of the previous votes (as in Doodle and Facebook polls, and in some internal corporate e-mail surveys). The standard solution concept for sequential games is subgame perfect Nash equilibrium. Subgame perfect voting equilibria, with and without abstention, have been analyzed by multiple researchers~\cite{Farquharson:69,McKelvey:78,dekel2000sequential,desmedt2010equilibria}. However, subgame perfection is a highly sophisticated behavior that requires a voter to know exactly the preferences of all of her peers. It also requires multiple steps of backward induction, at which human players typically fail~\cite{johnson2002detecting}. 

\medskip
Iterative voting sounds like a similar setting, but has generated a very different type of voting models. In an iterative setting, voters start from some given voting profile, and in each turn one or more voters may change their vote~\cite{meir2010convergence,reijngoud2012voter}. For example, voters can be members of a committee sitting in the same room. The game ends when no voter wants to change her vote. Meir et al.~\shortcite{meir2010convergence} proved that if voters play one at a time and adopt a myopic best-response strategy they are guaranteed to converge to a Nash equilibrium of the stage game from any initial state. The main problem with this approach is that it does not solve equilibria explosion. In particular, in Example~\ref{ex1} voter $v$ does not have a response that is better than his current action ($c$). More recent papers on the iterative setting suggested other myopic strategies~\cite{grandi2013restricted}, which suffer from similar problems. 

When some or all voters are allowed to change their votes simultaneously, we essentially have repeated polls~\cite{Chopra:04,reijngoud2012voter}. A particular model based on uncertainty with iterated polls was suggested by Reyhani et al.~\shortcite{reyhani2012coordination}. According to this model, each voter considers some set of \emph{possible winners} based on the poll and on some internal parameter called \emph{inertia}, and votes for her most-preferred candidate in this set. %In Example~\ref{ex1}, a voter $v$ with an appropriate value of inertia would only consider $a$ and $b$ as possible winners, and will thus vote for $b$. 
However the paper provides few results (focused on three candidates), and applies some arbitrary considerations in the construction of the possible winner set.

\section{The Formal Model}

\paragraph{Basic notations}
We denote $[x]=\{1,2,\ldots,x\}$.
The sets of candidates and voters are denoted by $M$ and $N$, respectively, where $m=|M|,n=|N|$. The Plurality voting rule $f$ allows voters to submit their preferences over the candidates by selecting an action from the set $M$. Then, $f$ chooses the candidate with the highest score, breaking ties lexicographically. 

Let $\pi(M)$ be the set of all orders over $M$.
We denote a preference profile by $\vec Q\in (\pi(M))^n$. %, and a voting profile by $\vec a\in M^n$. 
The preferences of voter $i$ are denoted by the total order $Q_i\in \pi(M)$, where $Q_i(a)\in[m]$ is the rank of candidate $a\in M$, and $q_i=Q^{-1}_i(1)$ is the most-preferred candidate. We denote $a \succ_i b$ if $Q_i(a) < Q_i(b)$. Let $\ol Q$ be the lexicographic order over candidates. 
Each voter announces his vote publicly. Thus the action of a voter is $a_{i}\in M$. The profile of all voters' votes is denoted as $\vec a \in M^n$, and the profile of all voters except $i$ is denoted by $\vec a_{-i}$. When abstention is allowed, we have $a_i\in M\cup \{\bot\}$, where $\bot$ denotes abstaining.

Note that every preference profile $\vec Q$ induces a game, where the (ordinal) utility of player $i$ in strategy profile $\vec a$ is $Q_i(f(\vec a))$. We refer to this game as the \emph{base game}.

If $a_i=q_i$ we say that $i$ is \emph{truthful} in $\vec a$, and voter $i$ is called a \emph{core supporter} of $a_i$. Otherwise, we say that $i$ is a \emph{strategic supporter} of $a_i$. 

The scoring profile $\vec s_{\vec a}\in \mathbb N^m$ that corresponds to $\vec a$ assigns a score to every candidate, taking tie-breaking into account. 
Formally, we refer to $s_{\vec a}(c)$ as equal to the number $|\{i\in N:a_i=c\}|$. When comparing two scores, we write $s_{\vec a}(c) >_{\ol Q} s_{\vec a}(c)$ if either $|\{i\in N:a_i=c\}| > |\{i\in N:a_i=c'\}|$, or $c,c'$ have the same number of votes and $c \succ_{\ol Q} c'$. We usually omit the subscript from $>_{\ol Q}$ as it is clear from the context. 

%
 %let $\ol Q\in \pi(M)$ denote the tie-breaking order, and $\eps<\frac{1}{m}$.
%We denote by $\ol s_{\ve\c a}(c) = |\{i\in N:a_i=c\}|$ the score of $c$ before tie-breaking, and by $s_{\vec a}(c) = \ol s_{\vec a}(c) - \ol Q(c)\cdot \eps$ the score including tie-breaking.

We will use $\vec a$ and $\vec s_{\vec a}$ interchangeably where possible, sometimes omitting the subscript $\vec a$ (note that we may only use $\vec s$ in a context where voters' identities are not important). Given a state $\vec s$ and an additional vote $a_i$, in the concatenated state $\vec s'=(\vec s,a_i)$ we have $s'(a_i) = s(a_i)+1$, and $s'(a)=s(a)$ for all $a\neq a_i$. %We use the notation $\vec s_{-i}$ to denote the profile for which $\vec s=(\vec s_{-i},a_{i})$.

\newsubsec{An intuitive description of voter response}
\label{sec:intuition}

While the notation we will introduce momentarily is somewhat elaborate and is intended to enable rigorous analysis, the main idea is very simple and intuitive.
We lean on two key concepts that are featured in previous models: \emph{dominated strategies}, and \emph{best-response}. From the perspective of voter $i$, a candidate $a$ \emph{dominates} candidate $b$ if $f(\vec s, a) \succ_i f(\vec s, b)$ for \emph{all} $\vec s$. In contrast, $a$ is a \emph{better-response} for a voter voting for $b$ \emph{in a particular state} $\vec s^*$, if $f(\vec s^*, a) \succ_i f(\vec s^*, b)$.

In our model, we will relax both concepts in a way that takes into account voters' uncertainty over the actual outcome. %As it turns out, other concepts such as \emph{truth-bias} and \emph{lazy-bias} will also have a natural interpretations in our model.
We assume that voters have a common \emph{estimated}, uncertain, view of the current state $\vec s$. In any given state, a voter considers a set of multiple ``close'' states which might be realized \textbf{without assigning probabilities to them}. We say that $a$ \emph{locally dominates} $b$ in $\vec s$ if $f(\vec s', a) \succ_i f(\vec s', b)$ in all $\vec s'$ that are considered ``possible'' in $\vec s$. Our key behavioral assumption is that a voter will avoid voting for candidates that are locally dominated (a standard assumption in strict uncertainty models, see Section~\ref{sec:epistemic}). In an iterative setting, a voter will vote for her most preferred candidate---among those who locally dominate her current action. Our key epistemic assumption (which is new) is that the possible states are those that are ``close'' to $\vec s$ according to some reasonable metric over vote counts. 

Tying our model back to Fisher's definition of tactical voting, a voter's \emph{belief} is captured by the estimated state $\vec s$, whereas her \emph{influence} is reflected by local dominance.
%We will later see how local dominance generalizes not only global dominance and best-response, but also truth-bias and lazy behavior.
%Note that if the voter only considers the realization of $\vec s$ itself, then he is essentially using best-response. In contrast, if a voter considers any realization as possible, then local dominance coincides with global dominance. We will later see how a slight variation of this idea can also account for truth-bias and lazy behavior.
%As a demonstration of how this approach, consider the example from the beginning of the previous section. Voter~$v$ estimates that the scores are $\vec s = (45,40,15)$, but is uncertain about the real scores. If his uncertainty is large enough to allow a possible outcome where he is pivotal between $a$ and $b$, then voting for $b$ locally-dominates $c$ in state $\vec s$.\footnote{Thfis is unless uncertainty is too large. See the formal model for details.} Note that to determine the optimal response, a voter only considers possible states in which he is pivotal. In that sense our model is similar to the MW model, except that no probabilistic reasoning is required.
%In any case, since a voter considers multiple possible world states, he might be pivotal in some, even when there are no pivotal voters in the current state. This increased perception of self-pivotality allows the voter to take a ``safe'' strategic action, which might turn out to be better and is never worse than the current action. 
The different sets of states that voters consider are part of their type, and can account for diverse voter behavior, yet ones that are bounded-rational. %\footnote{Local dominance is consistent with voters employing heuristics of \emph{availability}~\cite{tversky1973availability}: rather than assessing the probability of, say, a tie between $a$ and $b$, } 
Consider Example~\ref{ex1}, where the estimated counts are $\vec s=(45,40,15)$. If voter $v$ also considers as possible states where scores vary by $\pm 10$, he will be better off voting for $b$. Voting for $b$ might influence the outcome, whereas voting for $c$ is futile (unless $v$ considers an even higher variability in scores). 

One assumption that requires justification is the existence of the publicly known, estimated state $\vec s$. In an iterative and sequential settings the shared common view is easy to explain, as the estimated state is the actual current voting state. Uncertainty still exists since some voters might change their vote in later rounds.
In a simultaneous-vote game, the estimated state might be due to prior acquaintance with the other voters or due to polls, and uncertainty is due to polls' inaccuracy. 
While the assumption that voters know the (estimated) current state is slightly stronger than the one made by Laslier's heuristics for Approval voting (where only the identity of the leader is important), it is much weaker than the assumption that a voter knows the entire preference profile, or the exact distribution over preferences. 
%
%\newpar{Truth- and Lazy-bias}
%In truth-bias models it is assumed that a voter which is non-pivotal, will always vote truthfully. We can similarly weaken this strong assumption (which means that a voter will almost always vote truthfully, as he is almost never pivotal) as follows. A truth-biased voter will resort to truth-bias only if his current action does not locally dominate his truthful candidate. That is, if there is no benefit to vote untruthfully in \emph{any} possible state (according to the voter's view). 
%
%A relaxed notion of lazy-bias can be similarly defined, where a voter will abstain unless some other action locally dominates abstaining. 
%

\newsubsec{Local dominance}
Let $S\subseteq \mathbb N^m$ be a set of states. 
\begin{definition}
\label{def:beat}
We say that action $a_i$ $S$-beats $a'_i$ (w.r.t.\ voter $i$) if there is at least one state $\vec s\in S$ s.t.\ $f(\vec s,a_i) \succ_i f(\vec s, a'_i)$. That is where $i$ strictly prefers $f(\vec s,a_i)$ over $f(s,a'_i)$.
\end{definition}
We can think of $S$ as states that $i$ \emph{believes} to be possible (where these states do not include the action of $i$ himself).
The definition, however, does not depend on this interpretation.

\begin{definition}[Local dominance]
\label{def:local}
We say that action $a_i$ $S$-dominates $a'_i$ (w.r.t.\ voter $i$) if (I) $a_i$ $S$-beats $a'_i$; and (II) $a'_i$ does \emph{not} $S$-beat $a_i$.
\end{definition}
 Note that $S$-dominance is a transitive and antisymmetric relation (but not complete). See more on epistemic interpretation in Section~\ref{sec:epistemic}, where we also compare our definition of local dominance with previous work. In particular our definition coincides with the definition of dominance in \cite{conitzer2011dominating} and with similar definitions in \cite{reijngoud2012voter,van2013strategic}. The novelty  comes from the way $S$ is constructed, which is explained next. 

%
%We say that $a_i$ is [strictly] $S$-dominant if
%\begin{enumerate}
	%\item $a_i$ $S$-dominates at least one [all] $a'_i\in M$.
	%\item $a_i$ is not $S$-dominated by any $a'_i$.
 %\end{enumerate}

\subsection{Distance-based dominance}
%Suppose we have some distance metric for states, denoted by $\delta(\vec s,\vec s')$.
%For voter $i$ and $\vec a\in M^n$, let $S_i(\vec a,x)\subseteq \mathbb N^m$ be a the set of states that are at distance at most $x$ from $a_{-i}$. Formally, $S_i(\vec a,x)=\{\vec s' : \delta(\vec s',\vec s_{\vec a_{-i}}) \leq x\}$. 
%
%The simplest metric we apply is the $\ell_1$ norm. $\delta_{\ell_1}(\vec s',\vec s) \leq x$ means that we can attain $\vec s'$ by adding/removing a total of $x$ voters to $\vec a_{-i}$. We similarly define other $\ell_d$ norms, a multiplicative distance and an Earth Mover distance (EMD). See Appendix~\ref{sec:proofs} for details.
%\subsection{Distance-based dominance}
Recall that every full or partial profile $\vec a \in M^n$ corresponds to state $\vec s_{\vec a}\in \mathbb N^m$. 
Suppose we have some distance metric for states, denoted by $\delta(\vec s,\vec s')$.
For voter $i$ and $\vec a\in M^n$, let $S_i(\vec a,x)\subseteq \mathbb N^m$ be  the set of states that are at distance at most $x$ from $\vec a_{-i}$. Formally, $S_i(\vec a,x)=\{\vec s' : \delta(\vec s',\vec s_{\vec a_{-i}}) \leq x\}$.

The $\delta$ distance may be a $\ell_d$ norm for some $d\geq 1$. Thus $\delta_{\ell_1}(\vec s',\vec s) \leq x$ means that we can attain $\vec s'$ by adding/removing a total of $x$ voters to $\vec a_{-i}$ (think of $x$ an an integer, and note that the total number of votes in $\vec s,\vec s'$ may be different). Similarly, the $\ell_\infty$ norm means that we can add or remove at most $x$ votes for each candidate.

Another distance we can consider is the \emph{multiplicative distance}, where $\delta_{M}(\vec s',\vec s ) \leq x$ if for all $a\in M$, both $s'(a) \leq s(a)(1+x)$ and $s(a) \leq s'(a)(1+x)$. Intuitively, this means that the score of each candidate can change (either increase or decrease) by a factor of at most $(1+x)$.

A third natural distance is the \emph{Earth Mover distance} (EM), where $\delta_{EM}(\vec s',\vec s) \leq x$ if $\vec s'$ can be attained from $\vec s$ by changing the vote of at most $x$ voters (similar to the $\ell_1$ norm, but with the constraint that the number of votes remains the same).
 
We can also consider more elaborate ways of defining possible world states, e.g., where a voter is optimistic or pessimistic regarding the score of his favorite candidates; EM distance where voters only transfer their votes to higher-ranked candidates,\footnote{A similar assumption, together with a variation of the multiplicative metric, was used in \cite{reyhani2012coordination} to determine the set of possible winners.} and so on. %However, we do not consider such extensions in this paper.
For ease of presentation, we will consider the $\ell_1$ norm, unless explicitly stated otherwise.

\newsubsec{Strategic voting and equilibria}

Let $g_i: M^n \rightarrow M$ be a \emph{response function}, i.e. a mapping from voting profiles to actions (which implicitly depends on the preferences of voter $i$). Any set of response function $(g_i)_{i\in N}$ induces a (deterministic) dynamic in the normal form game corresponding to a particular preference profile under Plurality. In particular, it determines all equilibria of this game, which are simply the states $\vec a$ where no voter has a response that changes the state. We refer to the response function of a voter as her \emph{type}. We emphasize that the set of voting equilibria depends only on voters' preferences and response functions, and not on whether they vote iteratively or simultaneously.

% (we will later show that they indeed exist). 
	
\begin{definition}
\label{def:eq} Let $N$ be a set of voters with response functions $(g_i)_{i\in N}$. 
A \emph{voting equilibrium} is a state $\vec a$, where $a_i=g_i(\vec a)$ for all $i\in N$.
\end{definition}
	
We next define the primary response function that strategic voters in our model apply, which is based on local dominance.
%A step of a voter $i$ from candidate $a$ to another candidate $a'$ is called a \emph{strategic response}, and is denoted by $a \step{i} a'$. The following definition describes the basic strategic response in our model.
\begin{definition}
\label{def:strategic}
A \emph{strategic voter of type $r$} (or, in short, an $r$ voter) acts as follows in state $\vec a$.
%A voter of type $(r,k)$ will act as follows in a state $\vec a$.
Let $D\subseteq M$ be the set of candidates that $S_i(\vec a,r)$-dominate $a_i$. If $D$ is non-empty, then $i$ votes for his most preferred candidate in $D$. Formally, $g_i(\vec a) = \argmin_{d\in D} Q_i(d)$ if $D\neq \emptyset$, and $g_i(\vec a)=a_i$ otherwise.
\end{definition}
	
We refer to $r$ (or $r_i$ if types differ) as the \emph{uncertainty parameter} of the voter. We denote such a strategic step by $a_i \step{i} a'_i$, where $a'_i=g_i(\vec a)$.
We observe that:

\begin{itemize}
 \item If $a'_i$ $S_i(\vec a,r)$-dominates $a_i$, then $a'_i$ $S_i(\vec a,r')$-beats $a_i$, for any $r'\geq r$.
 \item For $r=0$, the voter knows the current voting profile exactly, and thus his response function is simple best-response, as in \cite{meir2010convergence}. %$a'_i$ $S$-dominates $a_i$ iff $a'_i$ is a better-response to $a_i$ for player $i$ in state $\vec s$.
	\item For $r=n$, a voter does not know anything about the current voting profile. Thus an action $a'_i$ locally dominates $a_i$ if and only if it weakly \emph{globally dominates} $a_i$. Note that this typically means $a_i$ is $i$'s least preferred candidate.\footnote{This may not hold if $n<m$. E.g. if $i$ is the only voter, then her top preference globally dominates all other candidates.}%, unless $%$a'_i$ $S$-dominates $a_i$ iff $a'_i$ (globally) weakly-dominates $a_i$ for player $i$.
\end{itemize}

Different definitions of strategic responses (distance metrics, value of $r$) may induce different sets of voting equilibria. However, the assumption that $i$ votes for the most-preferred candidate in $D$ is irrelevant to the set of induced equilibria.  The following is an immediate observation.
	
\begin{proposition}
\label{th:eq_dominated}
Let $N$ be a set of voters with preferences $\vec Q$ and following Def.~\ref{def:strategic} (voters may be of heterogeneous types). A voting profile $\vec a$ is a voting equilibrium, if and only if no voter votes for a locally dominated candidate. Formally, if $\forall i\in N, !\exists a'_i\in M,$ such that $a'_i$ $S_i(\vec a,r_i)$-dominates $a_i$.  
\end{proposition}
We assume of course, that the same parameters are used for defining the dominance relation and the strategic response of each voter.
For example, under Definition~\ref{def:strategic} with $r=0$, a strategic move coincides with best-response, and voting equilibria coincide with pure Nash equilibria of the base game $\vec Q$. 

%If a game is guaranteed to converge for \emph{some} scheduler, then it has a weak potential~\cite{?}.
	
\section{Convergence with Strategic Voters}
	
For any $w\in \mathbb N$, let $H_w(\vec s)\subseteq M$ be the set of candidates that need exactly $w$ more votes to become the winner. Thus $H_0(\vec s)=\{f(\vec s)\}$, $H_1(\vec s)=\{c: s(f(\vec s)) > s(c) \geq s(f(\vec s))-1\}$ (either have the same score as the winner and lose by tie-breaking, or $c$ wins in the tie-breaking but have one vote less), etc. Let $\ol H_w(\vec s) = \bigcup_{w'\leq w} H_{w'}(\vec s) = \{c : s(c) \geq s(f(\vec s))-w\}$. 

\newsubsec{Strategic responses and possible winners}

We say that candidate $c$ is a \emph{possible winner for $i$ in state $\vec a$} if there is an accessible state where $c$ wins. Formally, $W_i(\vec a,r) = \{c\in M : \exists \vec s'\in S_i(\vec a, r), f(\vec s',c)=c\}$. 

In contrast with $\ol H_w(\vec s)$, the definition of $W_i(\vec a,r)$ depends on the identity of the voter, and not only on her type.

We first show that in every strategic response, a voter always votes for her favorite possible winner.
\begin{lemma}
\label{lemma:best_winner}
Consider a strategic response $a_i \step{i} a'_i$ s.t.\ $a_i \notin W_i(\vec a,r)$. Then $a'_i = \argmin_{c\in W_i(\vec a,r)}Q_i(c)$. 
\end{lemma}
The lemma holds for any $\ell_d$ norm, $d\geq 1$, the multiplicative distance, and EM distance.
\begin{proof}
Consider the set of candidates $D$ in Def.~\ref{def:strategic}. Since $a'_i\in D$, it is non-empty. We will show that (I) $D \subseteq W_i(\vec a,r)$; and (II) $b = \argmin_{c\in W_i(\vec a,r)}Q_i(c)$ is in $D$. This implies $b=a'_i$. 

For (I), every candidate in $D$ in particular $S_i(\vec a,r)$-beats $a_i$, and therefore must be a possible winner (transferring $i$'s vote from one non-possible winner to another does not change the outcome in any accessible state). 

For (II), first note that since $a'_i\in D$, there is a state $\vec s' \in S_i(\vec a,r)$ where $i$ prefers $f(\vec s',a'_i)=a'_i$ over $f(\vec s', a_i)=a_i$, thus $a'_i \succ_i a_i$. Also, by (I) $a'_i\in W_i(\vec a,r)$ and thus $b \succeq_i a'_i \succ_i a_i$. Thus $a_i$ cannot $S_i(\vec a,r)$-beat $b$. It remains to prove that $b$ $S_i(\vec a,r)$-beats $a_i$ (and thus locally dominates it).

We consider two cases. Suppose first that for all $\vec s'\in S_i(\vec a,r)$, $f(\vec s',a_i)=b$. Then for any $c\in M$, and any $\vec s'$, $f(\vec s',c)\in \{b,c\}$. Moreover, by definition of $b$, if $f(\vec s',c)=c$ then $c\preceq_i b$. Thus no candidate $S_i(\vec a,r)$-beats $a_i$, and $D=\emptyset$. 
Thus suppose there is a state $\vec s'\in S_i(\vec a,r)$, $f(\vec s',a_i)=c\prec_i b$, and there is also a state (since $b$ is a possible winner) $\vec s'' \in S_i(\vec a,r)$, $f(\vec s'',b)= b$. It can be verified for each of the metrics we consider ($\ell_d$ norms, multiplicative, EM) that there must be a state $\vec s^* \in S_i(\vec a,r)$ where $f(\vec s^*,b)= b$ but $f(\vec s^*,a_i)= c$. For the $\ell_1$ norm, we get $\vec s^*$ from $\vec s'$ by adding votes to $b$ until $c,b$ are tied (or until there is a difference of one for $c$, if $b$ beats $c$ in $\ol Q$).\footnote{For other $\ell_d$ norms and multiplicative distance, we can consider any path of states between $\vec s',\vec s''$ that is contained in $S_i(\vec a,r)$, by gradually removing votes from $c$ and from other candidates, and adding votes to $b$. The critical state $\vec s^*$ will be along this path. For EM distance some paths may fail if we transfer votes directly from $c$ to $b$, but we can use a path where we transfer votes first from winners to $a_i$, and then from $a_i$ to $b$.} Thus $b$ $S_i(\vec a,r)$-beats $a_i$.
\end{proof}

Lemma~\ref{lemma:best_winner} does not mean that our dynamics coincides with ``always vote for the most preferred possible winner''. 
It only holds when the current vote of $i$ is \emph{not} a possible winner, and when there are at least two possible outcomes. 
If there is only one possible outcome, $i$ will not move (which makes sense). Also, if $a_i$ is a possible winner, then typically $i$ will not move, and if he does move it may be to a non-possible winner. For example, if $a_i$ is the least-preferred possible winner, than \emph{any} other candidate locally dominates $a_i$. We will later see that while these situations are hard to analyze, they can be avoided in some paths to equilibrium.

\newpar{Threshold for possible winners}
We continue with the following lemma, which shows that for some simple metrics, the set of possible winners is exactly all candidates whose score is above a certain threshold. 
Denote by $f^*_i = f(\vec a_{-i})$ the candidate that would win if $i$ would not vote.
	
\begin{lemma}
	\label{lemma:threshold}
	Each of the metrics $\delta$ from ($\ell_1,\ell_\infty$, multiplicative) induces a function $\beta= \beta_{\delta,r}:\mathbb N \rightarrow \mathbb N$, where
\begin{itemize}
\item For every $\vec a, i\in N$, if $a_i\neq f^*_i$, then $W_i(\vec a,r)=\{c: s_{\vec a_{-i}}(c) \geq_{\ol Q} \beta( s(f^*_i))\}$. That is, possible candidates are all those whose score is above the threshold, which is a function of the score of the winner.\footnote{We break ties with $\beta(s(f^*_i))$ as if we break ties with $f^*_i$.} 
\item $\beta(s)$ is weakly increasing in $s$. 
\end{itemize}
 In particular, for $\delta_{\ell_1}$, $W_i(\vec a,r) = \ol H_{r+1}(\vec s_{\vec a})$ for any $i$ s.t. $a_i\notin \ol H_{r+1}(\vec s_{\vec a})$.\footnote{Similarly, for $\delta_{\ell_\infty}$, $W_i(\vec a,r) = \ol H_{2r+1}(\vec s_{\vec a})$ for any $i$ s.t. $a_i\notin H_{2r+1}(\vec s_{\vec a})$. See appendix for the full proof.}
\end{lemma}
\rmr{We can actually define the threshold based on $\vec a_{-i}$ and remove the requirement that $a_i\neq f^*_i$.}
A similar result can be proved for EM and other $\ell_d$ norms, but the threshold would depend on the score of all candidates and not just the winner. 
\begin{proof}[Proof for the $\ell_1$ metric]
	Consider first the $\ell_1$ norm. We set $\beta(s) =\beta_{\ell_1}(s)= s-r-1$. Clearly, $f^*_i$ is a possible winner and $s(f^*_i)>\beta(\vec a)$. Assume $c\neq f^*_i$, then $s(c)\neq_{\ol Q} \beta(s(f^*_i))$.
	If $s(c) > \beta(s(f^*_i))$, consider the state $\vec s'\in S_i(\vec a,r)$ where $c$ has $r$ additional votes. We have that 
	$s'(c)+1= s(c)+r+1 > \beta(s(f^*_i)) +r+1 = s(f^*_i)$,
	thus $f(\vec s',c)=c$.
	In contrast, if $s(c) <\beta(s^*_i)$, then in any $\vec s' \in S_i(\vec a,r)$, $s'(f^*_i) - s'(c) >1$ and thus $c$ cannot win. 
	Finally, since $a_i\notin W_i(\vec a,r)$, then $s_{\vec a}(c)=s_{\vec a_{-i}}(c)$ for all $c\in W_i(\vec a,r)$. Thus $W_i(\vec a,r) = \{c: s(c)\geq \beta_{\ell_1}(s(f^*_i))\} = \{c : s(c) \geq s(f^*_i)-(r+1)\}= \ol H_{r+1}(\vec s)$.
	\end{proof}
	Note that Lemmas~\ref{lemma:threshold} and \ref{lemma:best_winner} together entail that the strategic decision of the voter is greatly simplified from both a behavioral and a computational perspective. There is no need to consider all possible world states. Only to check which candidates are sufficiently close to the winner in terms of their prospective score, and select the one that is most preferred.

	\rmr{In the Myerson and Weber model, they show that a voter only needs to consider the probabilities that pair of candidates are tied. The above lemma shows something similar in our model: a voter only considers which candidates may be tied.}

	%Lemma~\ref{lemma:threshold} does not hold for the EM distance and for other $\ell_d$ norms, since the score of candidates other than the winner may determine the threshold $\beta(\cdot)$. 

\newsubsec{Existence of equilibrium and convergence from the truthful state}
In what follows, we will only consider the $\ell_1$ norm for simplicity. However most results hold for other metrics as well. We also highlight that as in \cite{meir2010convergence}, we allow any number of non-strategic voters as part of the input. To allow continuous reading, some of the proofs were deferred the appendix.

\newpar{Best-response graphs and schedulers}
Given a game, any dynamic induces a directed graph whose vertices are the states of the game ($M^n$ in the case of the Plurality game). There is an edge from a state $\vec a$ to a state $\vec a'$, if (1) $\vec a,\vec a'$ differ only by the action of one player $i$; and (2) $g_i(\vec a)=a'_i$. We call this graph the \emph{best-response} graph. We can similarly create a \emph{group best-response graph} where edges correspond to the (non-coordinated) actions of subsets of agents. That is, there is an edge $(\vec a,\vec a')$ iff there is a subset $N'\subseteq N$ s.t.\ $a'_i=g_i(\vec a)$ for all $i\in N'$, and $a'_i=a_i$ for $i\notin N'$. 
	
	A \emph{scheduler} selects which voters play at any step of the game.  %is a function $\phi$ from game histories to $2^N$, where $\phi(\calH)$ is the set of players that play in time $t$ after history $\calH$. We require that $g_i(\vec a^t)\neq a^t_i$ for all $i\in \phi(\calH)$, i.e., that all of these players indeed have a non-trivial response. In particular, $\phi(\calH)=\emptyset$ iff $\vec a^t$ is an equilibrium. 
	A scheduler $\phi$ is a \emph{singleton scheduler} if it always selects a single voter, and otherwise it is a \emph{group scheduler}.\footnote{We emphasize that a group of voters moving at the same time \emph{does not} coincide with a coalitional manipulation. Each player acts as if he is the only one moving, and it may well be that the result of a move by a group is worse for all of its members.}
	 We assume that the order of players is determined by an arbitrary singleton scheduler (see~\cite{apt2012classification}). Equivalently, a scheduler can be thought of as a \emph{tie-breaking} mechanism for turns, when more than one voter or set of voters have a strategic move. 
	%, and differs from the one by Milchtaich~\cite{milchtaich2013schedulers}, where a \emph{scheduler} is a subgraph of the original best-response graph. The Milchtaich scheduler only specifies which players are \emph{allowed} to move, and thus does not uniquely determine the outcome. Also, it is history independent.} 
	%Such dynamics is known to converge from any initial state for a variety of games, including congestion games, voting games, and networks games~\cite{rosenthal1973class,meir2010convergence,frongillo2011social}.
	
	We next show that when starting from the truthful state, a singleton scheduler guarantees convergence to an equilibrium. In particular, an equilibrium must exist.
	Proposition~\ref{th:basic} also follows from more general convergence results that we will show later. However, we provide a simple and detailed proof that reveals the natural structure of the equilibrium that is reached. In the appendix, we show how to extend the proof to other distance metrics (Theorem~\ref{th:basic_distances}).
	\begin{theorem}
	\label{th:basic}
Suppose that all voters are of type $r$. Then a voting equilibrium exists. Moreover, in an iterative setting where voters start from the truthful state, for any singleton scheduler, they will converge to an equilibrium in at most $n(m-1)$ steps. 
\end{theorem}
%\documentclass[]{article}
%\usepackage{tikz}
%\usepackage{subcaption}
%%\usepackage{tikz}%,fullpage}
%\usetikzlibrary{arrows,%
                %petri,%
                %topaths}%
%\usepackage{tkz-berge}
%
%\begin{document}
%
%%%%%%%%%%%%%%%%%%%

\begin{figure}
\centering
\begin{subfigure}[]{0.4\textwidth}

\begin{tikzpicture}[scale=0.7,transform shape]

\tikzstyle{vote1}=[draw=black,fill=white,
inner sep=0pt,minimum size=5mm]
\tikzstyle{vote2}=[draw=black,fill=black!10!white,
inner sep=0pt,minimum size=5mm]
\tikzstyle{vote3}=[draw=black,fill=black!30!white,
inner sep=0pt,minimum size=5mm]
\tikzstyle{vote4}=[draw=black,fill=black!50!white,
inner sep=0pt,minimum size=5mm]
\tikzstyle{vote5}=[draw=black,fill=black!70!white,text=white,
inner sep=0pt,minimum size=5mm]

	\draw[thick] (1,0.2) -- (6,0.2);
	\node at (1,0) {a};
	\node at (2,0) {b};
	\node at (3,0) {c};
	\node at (4,0) {d};
	\node at (5,0) {e};
	\draw[dashed] (0,1.5) -- (6,1.5);
	\node at (5,4) {Initial state};
	\node at (1,0.5) [vote1] {};
		\node at (1,1) [vote1] {};
			\node at (1,1.5) [vote1] {};
				\node at (1,2) [vote1] {};
				\node at (1,2.5) [vote1] {};
				%\node at (1,3) [vote1] {};
	
	\node at (2,0.5) [vote2] {$a$};
	\node at (2,1) [vote2] {$a$};
	\node at (2,1.5) [vote2] {$a$};
	\node at (2,2) [vote2] {$a$};

	\node at (3,0.5) [vote3] {};
		\node at (3,1) [vote3] {};
			\node at (3,1.5) [vote3] {};
			\node at (3,2) [vote3] {};
  
	\node at (4,0.5) [vote4] {$a$};
		\node at (4,1) [vote4] {$c$};
	
	\node at (5,0.5) [vote5] {$a$};
		\node at (5,1) [vote5] {$c$};
		\node at (5,1.5) [vote5] {$c$};

\end{tikzpicture}

\end{subfigure}\quad
\begin{subfigure}[]{0.4\textwidth}

\begin{tikzpicture}[scale=0.7,transform shape]

\tikzstyle{vote1}=[draw=black,fill=white,
inner sep=0pt,minimum size=5mm]
\tikzstyle{vote2}=[draw=black,fill=black!10!white,
inner sep=0pt,minimum size=5mm]
\tikzstyle{vote3}=[draw=black,fill=black!30!white,
inner sep=0pt,minimum size=5mm]
\tikzstyle{vote4}=[draw=black,fill=black!50!white,
inner sep=0pt,minimum size=5mm]
\tikzstyle{vote5}=[draw=black,fill=black!70!white,text=white,
inner sep=0pt,minimum size=5mm]

	\draw[thick] (0,0.2) -- (6,0.2);
	\node at (1,0) {a};
	\node at (2,0) {b};
	\node at (3,0) {c};
	\node at (4,0) {d};
	\node at (5,0) {e};
	\draw[dashed] (0,2) -- (3,2);
	\draw[dashed] (3,2.5) -- (6,2.5);
	\node at (5,4) {Mid game};
	\node at (1,0.5) [vote1] {};
		\node at (1,1) [vote1] {};
			\node at (1,1.5) [vote1] {};
				\node at (1,2) [vote1] {};
				\node at (1,2.5) [vote1] {};
				%\node at (1,3) [vote1] {};
	
	\node at (2,0.5) [vote2] {$a$};
	\node at (2,1) [vote2] {$a$};
	\node at (2,1.5) [vote2] {$a$};
	\node at (2,2) [vote2] {$a$};

	\node at (3,0.5) [vote3] {};
		\node at (3,1) [vote3] {};
			\node at (3,1.5) [vote3] {};
			\node at (3,2) [vote3] {};
			\node at (3,2.5) [vote5] {$c$};
		\node at (3,3) [vote5] {$c$};
		\node at (3,3.5) [vote4] {$c$};

	\node at (4,0.5) [vote4] {$a$};

	\node at (5,0.5) [vote5] {$a$};

\end{tikzpicture}
\end{subfigure}\\
\begin{subfigure}[]{0.4\textwidth}

\begin{tikzpicture}[scale=0.7,transform shape]

\tikzstyle{vote1}=[draw=black,fill=white,
inner sep=0pt,minimum size=5mm]
\tikzstyle{vote2}=[draw=black,fill=black!10!white,
inner sep=0pt,minimum size=5mm]
\tikzstyle{vote3}=[draw=black,fill=black!30!white,
inner sep=0pt,minimum size=5mm]
\tikzstyle{vote4}=[draw=black,fill=black!50!white,
inner sep=0pt,minimum size=5mm]
\tikzstyle{vote5}=[draw=black,fill=black!70!white,text=white,
inner sep=0pt,minimum size=5mm]	

\draw[thick] (0,0.2) -- (6,0.2);
	\node at (1,0) {a};
	\node at (2,0) {b};
	\node at (3,0) {c};
	\node at (4,0) {d};
	\node at (5,0) {e};
	\draw[dashed] (1,3) -- (6,3);
	\node at (5,4) {Mid game 2};
	\node at (1,0.5) [vote1] {};
		\node at (1,1) [vote1] {};
			\node at (1,1.5) [vote1] {};
				\node at (1,2) [vote1] {};
				\node at (1,2.5) [vote1] {};
				\node at (1,3) [vote2] {$a$};
	\node at (1,3.5) [vote2] {$a$};
	\node at (1,4) [vote2] {$a$};
				%\node at (1,3) [vote1] {};
					\node at (1,5) {};
	
	\node at (2,0.5) [vote2] {$a$};

	\node at (3,0.5) [vote3] {};
		\node at (3,1) [vote3] {};
			\node at (3,1.5) [vote3] {};
			\node at (3,2) [vote3] {};
			\node at (3,2.5) [vote5] {$c$};
		\node at (3,3) [vote5] {$c$};
		\node at (3,3.5) [vote4] {$c$};

	\node at (4,0.5) [vote4] {$a$};

	\node at (5,0.5) [vote5] {$a$};

\end{tikzpicture}
\end{subfigure}\quad\quad
\begin{subfigure}[]{0.4\textwidth}

\begin{tikzpicture}[scale=0.7,transform shape]

\tikzstyle{vote1}=[draw=black,fill=white,
inner sep=0pt,minimum size=5mm]
\tikzstyle{vote2}=[draw=black,fill=black!10!white,
inner sep=0pt,minimum size=5mm]
\tikzstyle{vote3}=[draw=black,fill=black!30!white,
inner sep=0pt,minimum size=5mm]
\tikzstyle{vote4}=[draw=black,fill=black!50!white,
inner sep=0pt,minimum size=5mm]
\tikzstyle{vote5}=[draw=black,fill=black!70!white,text=white,
inner sep=0pt,minimum size=5mm]	

\draw[thick] (0,0.2) -- (6,0.2);
	\node at (1,0) {a};
	\node at (2,0) {b};
	\node at (3,0) {c};
	\node at (4,0) {d};
	\node at (5,0) {e};
	\draw[dashed] (1,4) -- (6,4);
	\node at (5,4.5) {Equilibrium};
	\node at (1,0.5) [vote1] {};
		\node at (1,1) [vote1] {};
			\node at (1,1.5) [vote1] {};
				\node at (1,2) [vote1] {};
				\node at (1,2.5) [vote1] {};
				\node at (1,3) [vote2] {$a$};
	\node at (1,3.5) [vote2] {$a$};
	\node at (1,4) [vote2] {$a$};
	\node at (1,4.5) [vote4] {$a$};
		\node at (1,5) [vote5] {$a$};
				%\node at (1,3) [vote1] {};
	
	\node at (2,0.5) [vote2] {$a$};

	\node at (3,0.5) [vote3] {};
		\node at (3,1) [vote3] {};
			\node at (3,1.5) [vote3] {};
			\node at (3,2) [vote3] {};
			\node at (3,2.5) [vote5] {$c$};
		\node at (3,3) [vote5] {$c$};
		\node at (3,3.5) [vote4] {$c$};

\end{tikzpicture}

\end{subfigure}

\caption{
\label{fig:converge}\small{An example of 18 voters voting over candidates $\{a,b,c,d,e\}$. The top left figure shows the initial (truthful) state of the game.  The letter inside a voter is his second preference. The dashed line marks the threshold $\beta_{\ell_1}$ of possible winners for voters of type $r=2$.  Thus candidates on or above the threshold are the set $\ol H_3(\vec a)$. Candidates that are on the dashed line (in $H_3(\vec a)$) are considered possible winners only by voters that do not currently vote for them.  Note that due to tie breaking it is not the same for all candidates. For example, since $a$ beats $b$ in tie-breaking, $b$ needs $2$ more votes to win in the initial state.
In the next two figures we can see voters leaving their candidates (who are not possible winners for them) to join one of the leaders.  The last figure shows an equilibrium that was reached. Note that although both of $a,c$ are possible winners for the last supporter of $b$, he has no strategic move. A different equilibrium may have been reached with a different scheduler.} }
\end{figure}
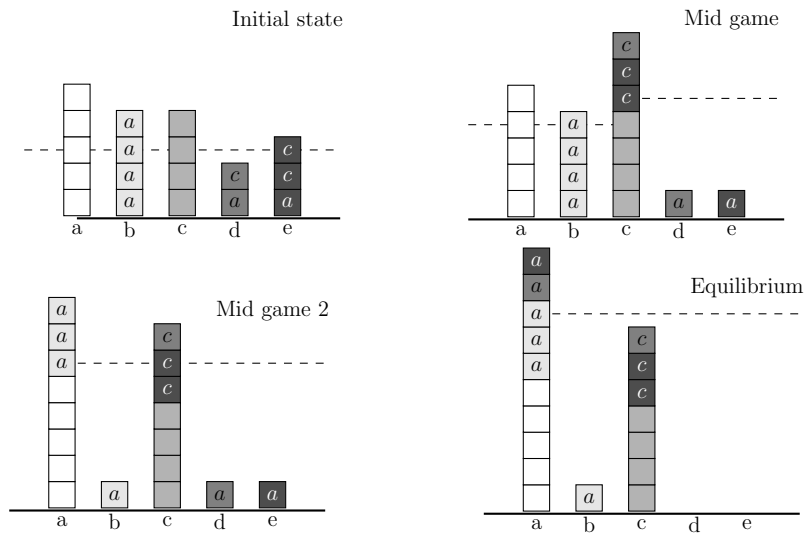
%\end{document}
\begin{proof}
%We prove for the $\ell_1$ norm, and later explain how to extend to other metrics.
If the truthful state $\vec q$ is stable, then we are done. Thus assume it is not.
Let $\vec a^t$ (and $\vec s^t$) be the voting profile after $t$ steps from the initial truthful vote $\vec a^0 = \vec q$.
Let $a_i\step{i} a'_i$ be a move of voter $i$ at state $\vec s=\vec s_{\vec a^t}$ to state $s'=\vec s_{\vec a^{t+1}}$. 
%
%Note that candidates voting for the current winner $f(\vec a)$ do not move. Therefore $i$ does not vote for the winner.
%We observe that for the $\ell_1$ norm, it holds for any candidate not voting for the winner, that
%$$\ol H_{r+1}(\vec s) = \{c : s(c) \geq s(f(\vec a))-r-1\}= \{c : s(c) \geq s(f^*_i)-r-1\} = \{c: s(c)\geq \beta_{\ell_1}(s^*_i)\}.$$
%By Lemma~\ref{lemma:move}, $a'_i$ is always $i$'s most preferred possible winner, and therefore the most preferred candidate in $\ol H_{r+1}(\vec s)$). 
We claim that the following hold throughout the game.
%Also, We make the following claims:
Recall that by Lemma~\ref{lemma:threshold}, $\ol H_{r+1}$ is the set of possible winners for all voters who are voting for other candidates.
\begin{enumerate}
	\item $a_i\notin \ol H_{r+1}(\vec s')$, i.e., once a candidate is deserted, it is no longer a possible winner.
	\item $a'_i \prec_i a_i$, i.e., voters always ``compromise'' by voting for a less-preferred candidate.
	\item $\max_{a\in M} s'(a) \geq \max_{a \in M} s(a)$, i.e., the score of the winner never decreases.
	\item $\ol H_{r+1}(\vec s') \subseteq \ol H_{r+1}(\vec s)$, i.e., the set of possible winners can only shrink.
\end{enumerate}
We prove this by a complete induction. 
%Note that we only need to show the base case for (3), as all other properties vacuously hold for all steps that preceded the first step. Further, if (3) does not hold at $\vec s^0= \vec s_{\vec q}$, then $\vec q$ is already stable.
\begin{enumerate}
	\item If this is the first move of $i$ then $a_i=q_i$. Otherwise, by Lemma~\ref{lemma:best_winner}, $a_i$ is the most-preferred candidate of $i$ in $\ol H_{r+1}(\vec s^{t'})$ where $t'$ is the time when $i$ last moved. By induction on (4), $\ol H_{r+1}(\vec s) \subseteq \ol H_{r+1}(\vec s^{t'})$. So if $a_i \in \ol H_{r+1}(\vec s) $, it must be the most-preferred candidate in the set. Assume, towards a contradiction, that $a_i \in \ol H_{r+1}(\vec s)$; then there is a state $\hat{\vec s} \in S_i(\vec a,r)$ where $i$ is pivotal between $a_i$ and $f(\vec a)\prec_i a_i$. For any $c\neq a_i$, $f(\hat{\vec s},c)=f(\vec a)$, and in particular for $c=a'_i$. Therefore $a_i$ $S_i(\vec a,r)$-beats $a'_i$, which means $a'_i$ does not $S_i(\vec a,r)$-dominate $a_i$. A contradiction. 
	%We observe that the set $D$ from Definition~\ref{def:strategic} is a subset of $\ol H_{r+1}(\vec s)$, since no other candidate $S_i(\vec a,r)$-beats $a_i$. Let $a^*$ be $i$'s most preferred candidate in $\ol H_{r+1}(\vec s)$. If $a^*\neq f(\vec s)$, then $D$ contains all candidates in $\ol H_{r+1}(\vec s)$ that are preferred over the winner, and in particular $a^*$.
	%If $a=f(\vec s)$, then either $D$ is empty (in which case there is no move), or $D=\{a^*\}$. In either case, $a'=a^*$.
	
	\item If this is the first move of $i$ then this is immediate. Otherwise, by induction on Lemma~\ref{lemma:best_winner} and (4), if $a'_i \succ_i a_i$, then $i$ would prefer to vote for $a'_i$ in his previous move, rather than for $a_i$.
	 
	\item As in (1), if $i$ votes for $a_i=f(\vec a)\in \ol H_{r+1}(\vec a)$, then $a_i$ is $i$'s most-preferred possible winner. Thus it cannot be locally dominated by any other candidate.
	%since $a_i \notin \ol H_{r+1}(\vec s')% since in any step the deserted candidate $a_i$ cannot even be a possible winner after one vote is removed. Thus $a_i$ cannot be the winner before the vote is removed.
	 \item Since by (3) the score of the winner never decreases, the only way to expand $\ol H_{r+1}$ is to add a vote to a candidate not in $\ol H_{r+1}$. By Lemma~\ref{lemma:best_winner} this never occurs. %, so we are done. 
\end{enumerate}
Finally, by property (2), each voter moves at most $m-1$ times before the game converges. %This concludes the proof for $\ell_1$.
\end{proof}

The proof not only shows that an equilibrium exists, it also describes exactly the way in which such equilibria are reached from the truthful state. There is always a set of ``leaders'' ($\ol H_{r+1}$ in the case of the $\ell_1$ norm). Strategic voters vote for their favorite candidate in this set, if their current candidate is not a possible winner. At some point candidates may ``drop out'' of the race as their gap from the winner increases, and the set $\ol H_{r+1}$ shrinks. This continues; in the reached equilibrium, all strategic voters vote for their best possible winners (which is in $\ol H_{r+1}$). As we will see next, the only case where there are voters $N'\subseteq N$ not voting for a possible winner, is when the gap between the winner and the runner-up is exactly $r+1$, and all of $N'$ prefer the current winner. In particular, the gap can never grow larger than $r+1$, and thus there are always at least two leaders whenever there is some strategic interaction (see Figure~\ref{fig:converge}). 

%We also show that this competition cannot result in a winner with a very large gap from the runner-up (unless this is already the case in the initial state).
\begin{lemma}
\label{lemma:two_winners}
Under the conditions of Theorem~\ref{th:basic}, either $\vec q$ is stable, or in every state $\vec s^t$ we have $|\ol H_{r+1}(\vec s^t)|>1$. Also, in the last state $\vec a$ either $|\ol H_{r}(\vec s_{\vec a})|=1$, or all voters vote for possible winners. Any voter voting for $c\notin \ol H_{r+1}(\vec s^t)$ prefers $f(\vec a)$ over any other candidate in $\ol H_{r+1}(\vec s^t)$.
\end{lemma}

\newsubsec{Convergence under broader conditions}

To show robust convergence results, there seem to be two main extensions. First, we would like the game to converge from any initial state, and not just from the truthful one. Second, we would like convergence to occur even if more than one voter moves between states. In other words, we would like to see convergence under any group scheduler. 

Unfortunately, under arbitrary group schedulers, convergence is not guaranteed even from the truthful state. A simple example for $r=0$ appears in \cite{meir2010convergence}, where there are two runner-ups that are preferred to the winner by all of their supporters, but the supporters fail to coordinate on a runner-up to promote. 
We conjecture that as in the case of best-responses ($r=0$), a \emph{singleton scheduler} would guarantee convergence from any initial state. 

We show that if we make two mild restrictions on the scheduler, then convergence from any state is guaranteed even for groups. We say that a step $a\step{i} a'$ is of type~1 if $a' \prec_i a$, and type~2 if $a' \succ_i a$. We call type~1 steps \emph{compromise steps}, and type~2 \emph{opportunity steps}.

\begin{proposition}
\label{th:group}
Suppose that all voters are of type $r$. 
Consider any group scheduler such that (1) any voter has some chance of playing as a singleton (i.e., this will occur eventually); (2) the scheduler always selects (an arbitrary subset of) voters with type~2 moves, if such exist. 
Then convergence is guaranteed from any initial state after at most $O(nm)$ singleton steps occur.
\end{proposition}

The proof shows that in particular, for singleton schedulers there is a path of best-responses from any state to an equilibrium (any singleton path where type~2 steps precede type~1 steps). % This means that the voting game has a \emph{weak potential}~\cite{marden2009payoff}.
Starting from the truthful state is a special case, where there are no type~2 moves before the first type~1 move. 
The assumption that type~2 moves are played first can be justified to some extent, since type~1 moves are ``compromises'' and thus voters may be more reluctant to carry them out. 

%
%One corollary is that if the game is played with a singleton scheduler in a random order, then it converges from any initial state, with probability that goes to $1$ with the number of steps. This is since in every state $\vec a^t$ there is some probability $p$ that the next $2nm$ steps will be played in the described order, i.e. that type~2 moves precede type~1 moves. Thus after $t+2nm$ steps the probability that the game is still running is at most $(1-p)^t \stackrel{t\rightarrow \infty}{\rightarrow} 0$. Alternatively, we can say that the voting game has a \emph{weak potential}~\cite{marden2009payoff}: the potential increases along every path where type~2 moves precede type~1 moves. %[I think we can also say the game has a weak potential]

\newsubsec{Truth-bias and Lazy-bias}

The basic strategic behavior from Definition~\ref{def:strategic} still allows for some counter intuitive actions. For example if there is only one possible winner, then a voter will not change her vote regardless of how much she dislikes her current action. In particular we can still construct an equilibrium where all voters vote for their least preferred candidate!

However, the notion of local dominance is very flexible, and allows us to define more subtle and plausible behaviors. Specifically, by adding a negligible utility $\eps$ to a favorite action, such as truth-telling or abstaining, we get that this action locally dominates any other action where the outcome is the same (that is, the same in all accessible world states). We can thus define \emph{truth-biased} or \emph{lazy} voters, who prefer to tell the truth or to abstain whenever they do not see themselves as pivotal. We highlight that the local neighborhood considered by the voter when deciding whether to cast a strategic vote and when applying truth-bias/lazyness, is not necessarily the same neighborhood.
	
	\begin{definition}
	\label{def:truth_bias}
	A {strategic truth-biased voter of type $(r,k)$} (or, in short, a $T(r,k)$ voter) acts as follows in state $\vec a$.
\begin{enumerate}
	\item (strategic move) Let $D\subseteq M$ be the set of candidates that $S_i(\vec a,r)$-dominate $a_i$. If $D$ is non-empty, then $i$ votes for his most preferred candidate in $D$. 
	\item if $a_i$ $S_i(\vec a,k)$-beats $q_i$, then $i$ keeps current vote $a_i$.
	\item (truth bias move) otherwise, $i$ moves to $q_i$.
 \end{enumerate}
\end{definition}
A \emph{strategic lazy voter of type $(r,k)$} (an $L(r,k)$ voter) can be similarly defined, replacing $q_i$ with the action $\bot$ (abstain).
%(maybe denote them as r-type, k-type, and t-type moves?)

%\rmr{remove if need space:}
%\begin{remark}
%An equivalent way to define the truth-bias (or lazy) move, is to first consider the strategic move as specified, and then say that $i$ moves to $q_i$ (or $\bot$) if $q_i$ $S_i(\vec a,k)$-dominates $a_i$, under the assumption that voting for $q_i$ yields additional $\eps$ utility.
%However we do not add the $\eps$ when considering the initial strategic move.
%
%We cannot simply replace ``$S$-beats'' with ``$S$-dominates'' in the second bullet, as then we can have a voter who is stuck is a loop of strategic and truth-bias moves (when some $a$ $S_i(\vec a,r)$-dominates $q_i$, but does not $S_i(\vec a,k)$-dominate $q_i$).
%\end{remark}

We make the following observations.

\begin{itemize}
	%\item if $i$ is completely non-pivotal in $\vec a$ (according to $k$), then he will resort to truth-telling.
	%\item if $r$ is sufficiently large (in particular $r=m$), then a type $r$ voter will move in profile $\vec a$ if and only if $a_i$ is her least preferred candidate (note that the least preferred candidate is a weakly dominated strategy globally, and locally dominated for any $S$). In particular, type $m$ voters who start as truthful, remain truthful.
% \item A type $0$ voter always plays best-response, as in the basic model [Meir et al.'2010].
 \item A $T(r,n)$ or $L(r,n)$ voter is just an $r$ voter. This is since truth/lazy-bias will never be applied. %Unless $a_i$ is $i$'s least preferred candidate, there is always some state where it is better to vote $a_i$ than $q_i$ or $\bot$.
%\item A $T(r,m)$ voter is always truthful, since no action globally dominates $q_i$.
%\item Similarly, a $L(t,m)$ voter always abstains.
 \item $T(0,0)$ voters follow the truth-bias model in \cite{thompson2013empirical,obraztsova2013plurality}.
\item $L(0,0)$ voters follow the (simultaneous) ``lazy'' model of Desmedt and Elkind~\shortcite{desmedt2010equilibria}.
\end{itemize}
Intuitively, the parameters $r$ and $k$ reflect two uncertainty thresholds that determine how much a voter is inclined to make a strategic \underline{r}esponse, and how much to \underline{k}eep his current vote. 
 We argue that it is natural to assume $k>r$, which entails that a voter requires a lower uncertainty level in order to make a new strategic step, than to merely keep his current strategic vote. From a behavioral perspective, such an assumption accounts for \emph{default-bias}: decision makers have a higher tendency to stay with their current decision, than to adopt a new one~\cite{kahneman1991anomalies}. Lower values of $k$ (i.e., closer to $r$) correspond to a stronger truth-bias, whereas higher values correspond to a stronger default bias.
\begin{proposition}
\label{th:truth_bias}
Suppose that each voter $i$ is of type $L(r,k_i)$ or $T(r,k_i)$, where $k_i>r$. Then a voting equilibrium exists. %Assume further that $|\ol H_{r+1}(\vec q)|>1$, 
Moreover, in an iterative setting where voters start from the truthful state, they will always converge to an equilibrium in at most $3nm$ steps.
\end{proposition}

%\end{proof}
\newsubsec{Other considerations}

\paragraph{General voting equilibria}
What is the structure of voting equilibria that are not obtained by starting from the truthful state? They can be more diverse, but no voter votes for a locally dominated candidate (Prop.~\ref{th:eq_dominated}).

\begin{corollary} \label{th:general_eq} In any voting equilibrium $\vec a$, for any set of heterogeneous voters (i.e. voter $i$ is of type $T(r_i,k_i)$ or $L(r_i,k_i)$, it holds that:
\begin{enumerate}
	\item Every voter is either truthful, or votes for a candidate in $\ol H_{k_i}(\vec a)$.
	\item No voter votes for his least-preferred candidate in $W_i(\vec a,r_i)$.
\end{enumerate}
\end{corollary}
In particular, when $k_i$ is not too high, this rules out ``crazy'' equilibria where voters vote for their least preferred candidate.
It is possible, though, that all voters vote for their second-least-preferred possible winner, or their second-least-preferred candidate (consider a profile where candidates $a,b$ are ranked last by all voters, where roughly half rank $a$ above $b$). 

\newpar{Unique best-response}
Meir et al.~\shortcite{meir2010convergence} argued for a ``unique best-response''. They enforced a requirement that voters only move to candidates that, as a result, become winners. We show that when $r>0$, the notion of local dominance provides some natural justification for this requirement.\footnote{We thank Greg Stoddard for a discussion leading to this observation.}
\begin{observation}
Let $r>0$, $S=S_i(\vec a,r)$. Suppose that $c$ is the most preferred possible winner for $i$, and $b$ is not a possible winner for $i$. Then $c$ $S_i(\vec a,r)$-dominates $b$. %Thus $i$ has no reason to vote for $b$.
\end{observation}
\rmr{However, $i$ might still vote for $b$ if $b$ dominates $a_i$. This could be a starting point to a finer definition of the strategic response, where $i$ only votes for candidates that are maximal in the dominance relation. That is, if $b$ dominates $a_i$, but $c$ dominates $b$, then $i$ would not vote for $b$ even if $b\succ_i c$. In our current model what may happen is that $i$ will vote for $b$, and then (if no other agent interrupts) will change his vote to $c$. }
\olev{In your above example in which $b$ would be voted for, I think it also requires that only $a_{i}$ and some $c\succ_{i} a_{i}$ are the only elements in $S_{i}(\vec a, r)$} \rmr{changed "`will happen'' to "`may happen"'}
Intuitively, while voting for $b$ may improve the outcome for $i$ (over $a_i$) in some world states, voting for $c$ will be at least as good in those states, and strictly better in the states where $i$ is pivotal between $c$ and another possible winner. 
\section{Simulations of strategic voting}\label{simulationSection}
%Approaching the model presented in this paper empirically, 
We explore via extensive simulations how employing local-dominance affects the result of the voting process. These simulations have two primary goals. First, we want to understand better the effect of different parameters on the technical level (for example, how long does it take to reach an equilibrium if we vary the uncertainty level?). More importantly, we use simulations to test the properties of our strategic model with respect to the desiderata listed on Section~\ref{sec:desiderata}. For example, what is its discriminative power, is it robust to small changes, and whether it replicates common phenomena. See Section~\ref{sec:des_eval} for a summary of our findings in light of the desiderata.

\newpar{Preferences} As the outcome depends not just on our strategic model but also, and perhaps mainly, on the preference profile, we need to specify appropriate distributions on preferences for our simulations. 
%However, even examining these aspects require us not only to analyze our model, but also to choose voting distributions, as, naturally, the properties of our model are significantly affected by the particular characteristics of the voter distributions chosen. As there is a general lack of good, public-domain, complete preference orders,
 We generate preference profiles from a set of distributions which have been examined in the research literature, with a focus on distributions that are claimed to resemble preferences of human societies: The Uniform (or impartial culture) distribution; a uniform Single-peaked distribution; Polya-Eggenberger Urn model (with 2 urns and with 3 urns); a Riffle distribution; and Placket-Luce distribution. Urn models were particularly designed to resemble preference structures in human societies, whereas in Placket-Luce distributions each voter is assumed to have a noisy signal of some ground truth. See \ref{sec:distributions} for details.

\newpar{Methods} 
 We generated profiles from all distribution types for various numbers of voters and candidates, which resulted in 108 distinct distributions. From each distribution we sampled 200 instances.\footnote{We also used three datasets from German pre-election polls, with 100 voters, 3 candidates, and no sampling, as well as all 225 currently available full preferences from PrefLib (\url{http://preflib.org}).} Then, we simulated strategic voting on each instance varying the distance metric ($\ell_1$, multiplicative), the voters' types (basic, truth-biased, lazy) and the uncertainty parameters $r$ and $k$.

We simulated voting in an iterative setting, where voters start from a an initial  state, and then iteratively make strategic moves until convergence. Simulations all started from the truthful outcome, except for one batch on which we will elaborate later.
We repeated each simulation 100 times (as the scheduler may pick a different path each time), and collected multiple statistics on the equilibrium outcomes. See \ref{sec:methods} for details. All of the collected data can be downloaded from \url{http://tinyurl.com/k2b775e}.

\newsubsec{Results}
\paragraph{Meaningful parameters and robustness}
We observed that results with the multiplicative metric were generally very similar to those with the additive ($\ell_1$) metric. Also, while voting with lazy and truth-biased voters resulted in a somewhat different dispersion of the votes in equilibrium, all of the properties that we measured remained largely the same, regardless of the value of $k$. We conclude that these parameters have little effect on the model when the initial state is truthful, and focus on results for voters with additive distance metric and without truth- or lazy-bias.\footnote{%We note the number of reachable states was larger with the multiplicative distance function, though this effect shrunk with more voters, practically reversing itself slightly for large $r$. 
With lazy-bias simulations typically took more steps to converge.}

The choice of scheduler type also turned out to be an immaterial one. While a group scheduler typically converged much faster, there was only a negligible difference in the equilibrium outcomes for the vast majority of preference profiles. 

%We will focus on the results of additive $r$-voters, as $k$ turned to have almost no effect on most data points (it had some minuscule effect increasing the number of steps to reach an equilibrium and the number of winners found by the runs, see more information in the appendix), and the multiplicative metric turned out to be generally equivalent to the additive one (for both $r$ and $(r,k)$ voters) in most respects

The most meaningful parameter in the simulations was the uncertainty level $r$. As we vary the value of $r$ from $0$ to $15$, there is an increase and then a decrease in the amount of strategic behavior, with a ``peak value'' for $r$. We can see the effect of more strategic behavior by looking at the number of steps to convergence (Figure~\ref{fig:step_num}), the higher dispersion of equilibrium states, and the (lower) agreement with the Plurality winner (Figure~\ref{fig:riffle_plurality} in appendix). This pattern makes sense, as with low $r$ the voter knows the current state exactly, and often realizes he is not pivotal. As uncertainty grows the voter considers himself pivotal more often, but beyond the peak $r$ uncertainty is sufficiently large for all voters to believe that their truthful vote is also a possible winner (and then the initial state is stable). 

This pattern repeats in all 108 distributions. The effect of $r$ and in particular its peak value are determined mainly by the type of the distribution and the number of voters, where the peak $r$ increases with $n$ (for distributions with $n=10$, $r$ typically peaks at $0$). The number of candidates may affect the strength of the strategic effect, but not the peak $r$. 

%Studying the effect of changing $r$ values, it became clear, throughout the various scenarios, that $r$'s influence in many variables was mostly shaped by the distribution and the number of voters, with the number of candidates accenting a certain behavior, but not changing it. In all distributions, as the number of voters grew, the convergence to a single stable state happened later, so the voter number determined, in a sense, the support of variable --- what is the $r$ value it converges at. Furthermore, one can observe the ``peak $r$'', i.e, the point where $r$ causes the most significant change, which is generally the same for all simulations of particular distribution and voter number. This is, in a sense, the point where voters have enough unstable states to move from (unlike with large $r$ values, where truthful states are generally stable), but not too little as all processes are directed to the few existing ones. We can see that at $r$ ranging from 2 to 5, players reach more stable states, select slightly more different winners, at times entailing a large number of steps to reach en equilibrium (obviously, as the voter number grows, the peak point grows --- for the Placket-Luce distribution, the average number of steps to reach an equilibrium is maximized at $r>15$).
%An example for this clustering according to voter number can be seen in Figure~\ref{peakedStateNum}, showing the single-peaked distribution average number of equilibria states reached in the iterative process.

\begin{figure}[ht]
\begin{center}
\includegraphics[scale=0.26]{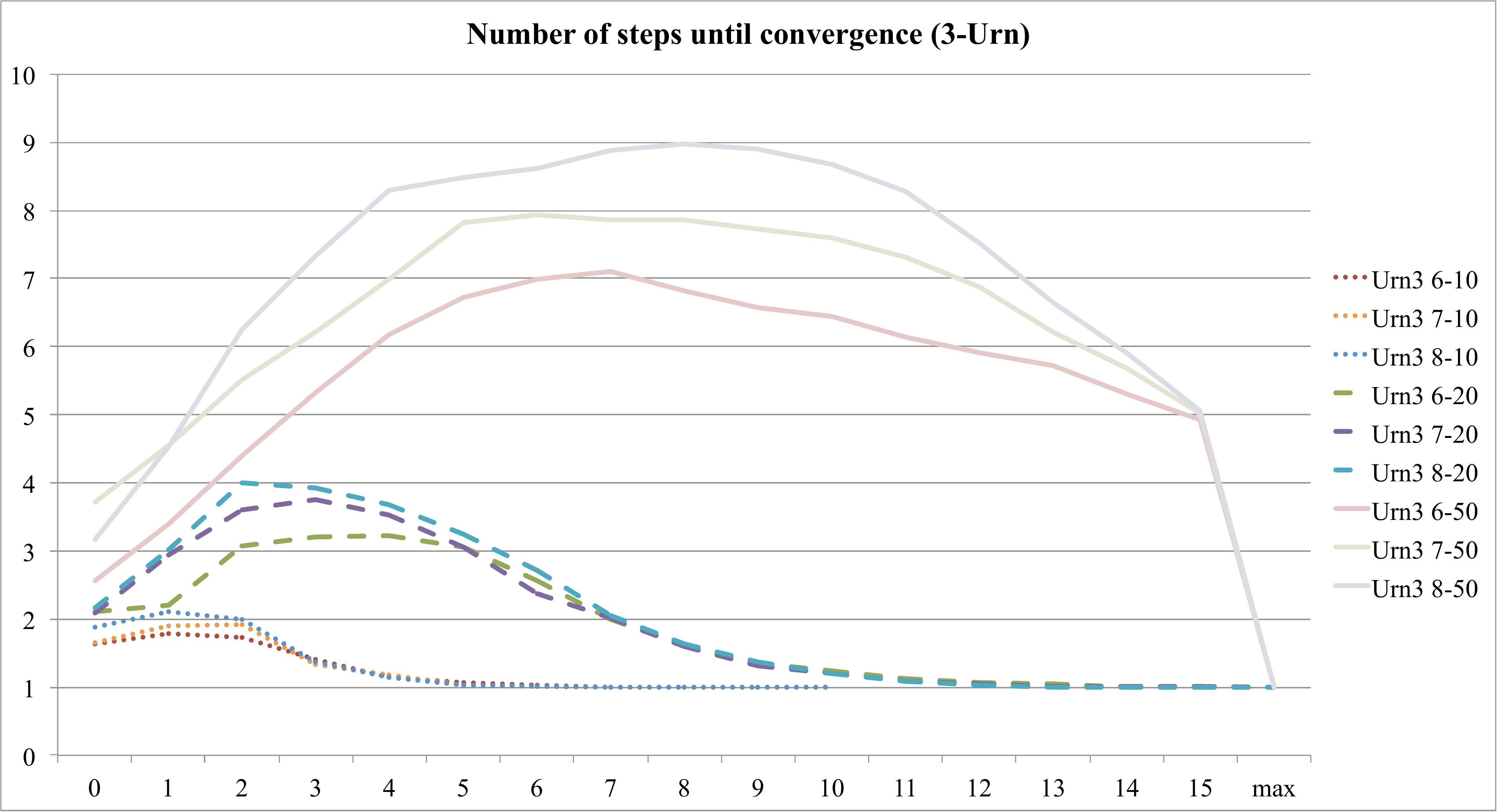}
 \end{center}
\caption {Average number of steps to reach equilibrium for the 2-urn distribution (\emph{NumStep}), as a function of the uncertainty parameter $r$. Each group of lines marks a different number of voters, and peaks at a different value of $r$. Recall that for $r=n$ (denoted by `max'), we have no strategic steps. This serves as our baseline.\label{fig:step_num}}
\end{figure}

\newpar{Quality of winner}
In the Placket-Luce distribution, the quality of a winner can be determined according to its rank in the ground truth used to generate the profile. In the other distributions there is no notion of ground truth, and hence we measured how often the equilibrium winner agreed with the (truthful) winner of another common voting system, which takes the entire preference profile into account (Borda, Copland, Maximin). We also measured how often the winner was a Condorcet winner (out of cases where one exists), and the social welfare of voters (assuming Borda utilities). 
 
According to Borda, Copland, Condorcet consistency, social welfare (see Figure~\ref{fig:peaked_welfare}) and the ground truth, a clear pattern was observed almost invariably across all distributions. As strategic activity increases, so does the winner quality.\footnote{There was typically no higher agreement with the Maximin winner. Also, in the Urn models strategic behavior did not in general improve consistency with Borda.} Best winner quality is attained at peak $r$ or very close to it (see Figures~\ref{fig:Luce_ground},\ref{fig:peaked_median} in the appendix).

In particular, these results are interesting for the Single-Peaked profiles. In such profiles there is always a Condorcet winner, which is the median candidate. As voters strategize more under Plurality, they in fact get closer to the outcome of the strategy-proof median mechanism.

\begin{figure}[ht]
\begin{center}
\includegraphics[scale=0.28]{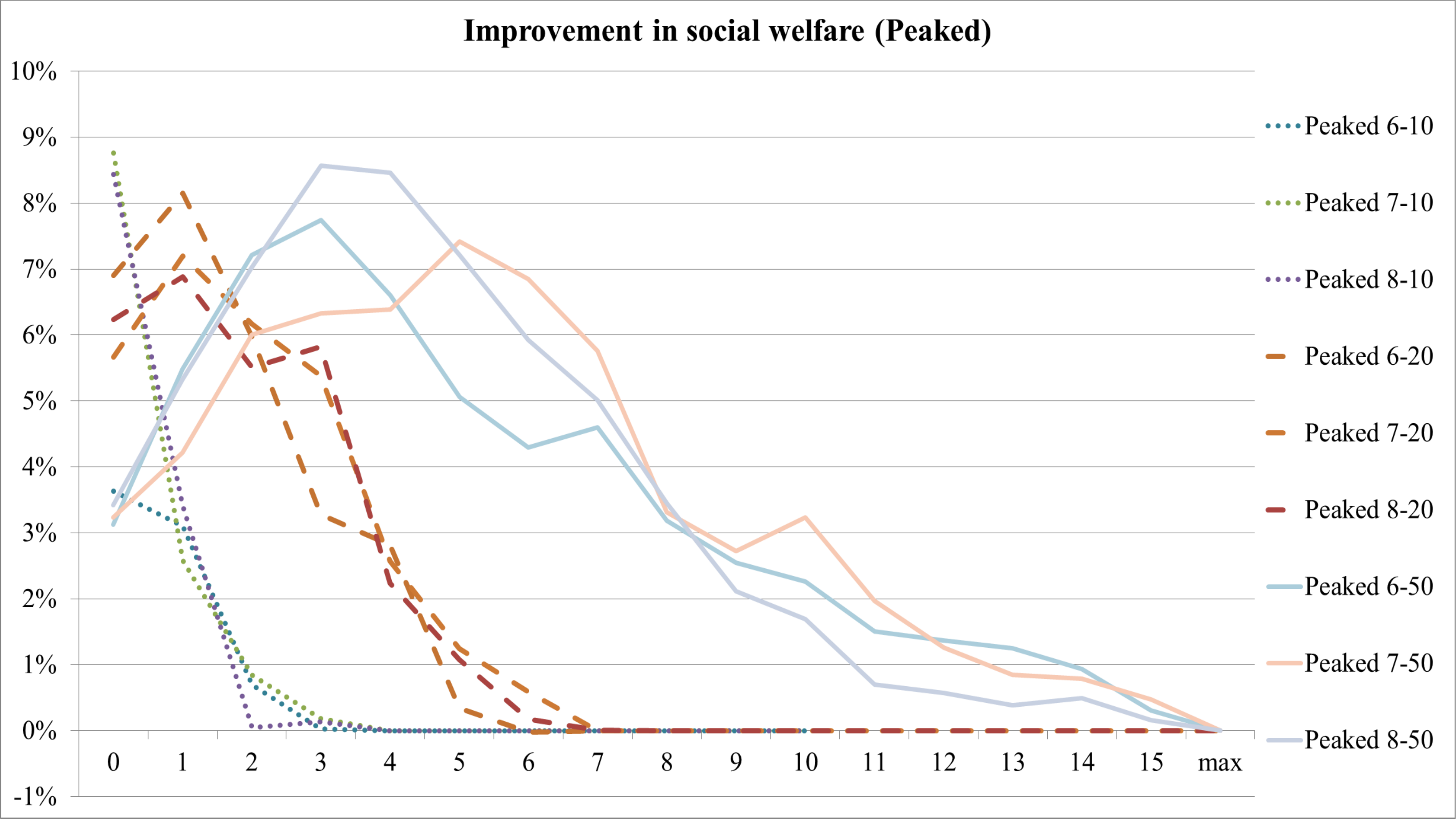}
 \end{center}
\caption {The increase in voters' social welfare, compared to the truthful Plurality outcome.} %For the Placket-Luce distribution, distance from ground truth
\label{fig:peaked_welfare}
\end{figure}

 %(average winner rank; probability of selecting Condorcet winner; distance of single-peak from median; and distance of Placket-Luce from its ground truth) we continue to see the similar property of a ``peak $r$'' shared among a distribution and voter number, but the candidate number amplification of the values is much more pronounced. Most quality measures seem unchanged for most large $r$ values (as there are, it seems, enough equilibria state for most potential winners). However, as $r$ shrinks, this situation changes, and again, there is a $r$ value for which, in average, better candidates are selected. The peculiarities of various distributions and their interaction with voting system emerges in these cases, as, for example, the states for which the Placket-Luce reaches winners which are closest to its ground truth (and hence, more likely to select a Condorcet winner) are those where the proportion of its winners which is the same as plurality would have selected is smallest.

%\begin{figure}[ht]
%\begin{center}
%\includegraphics[scale=0.3]{luceOptimal.pdf}
 %\end{center}
%\caption {The Placket-Luce distribution average distance of winner from the ground truth winner. Simulations with the same voter number have the same overall shape and reach better winners approximately at the same $r$.\label{luceQuality}}
%\end{figure}

\newpar{Duverger Law}
The positive effect of more strategic behavior on the concentration of votes was remarkably clear across all distributions. As $r$ gets closer to the peak value, over 75\% of the voters (all voters in some distributions) end up voting for only two candidates. This holds both in distributions like Uniform and Placket-Luce where no candidate initially has a strong advantage (see Figure~\ref{fig:luce_duv} in appendix), and in distributions like 2-Urn where there are two leading candidates to begin with (Figure~\ref{fig:duv_all}).
%Another particular property which was most evident was how ``medium sized'' $r$ values encourage equilibria exhibiting Duverger law. This could, or course, be seen very prominently in distributions which do not have a ``built-in'' bias for large homogenous groups of voters, but it could even be seen in such distributions where one would believe already have a Duverger law bias, like Polya-Eggenberger and Placket-Luce.

%\begin{figure}[ht]
%\begin{center}
%\includegraphics[scale=0.3]{urnDuv.pdf}
 %\end{center}
%\caption {The Polya-Eggenberger 2-urn distribution average equilibria where all voters voted only to top two candidates..}\label{urnDuverger}
%\end{figure}

\begin{figure}[ht]
\begin{center}
\includegraphics[scale=0.3]{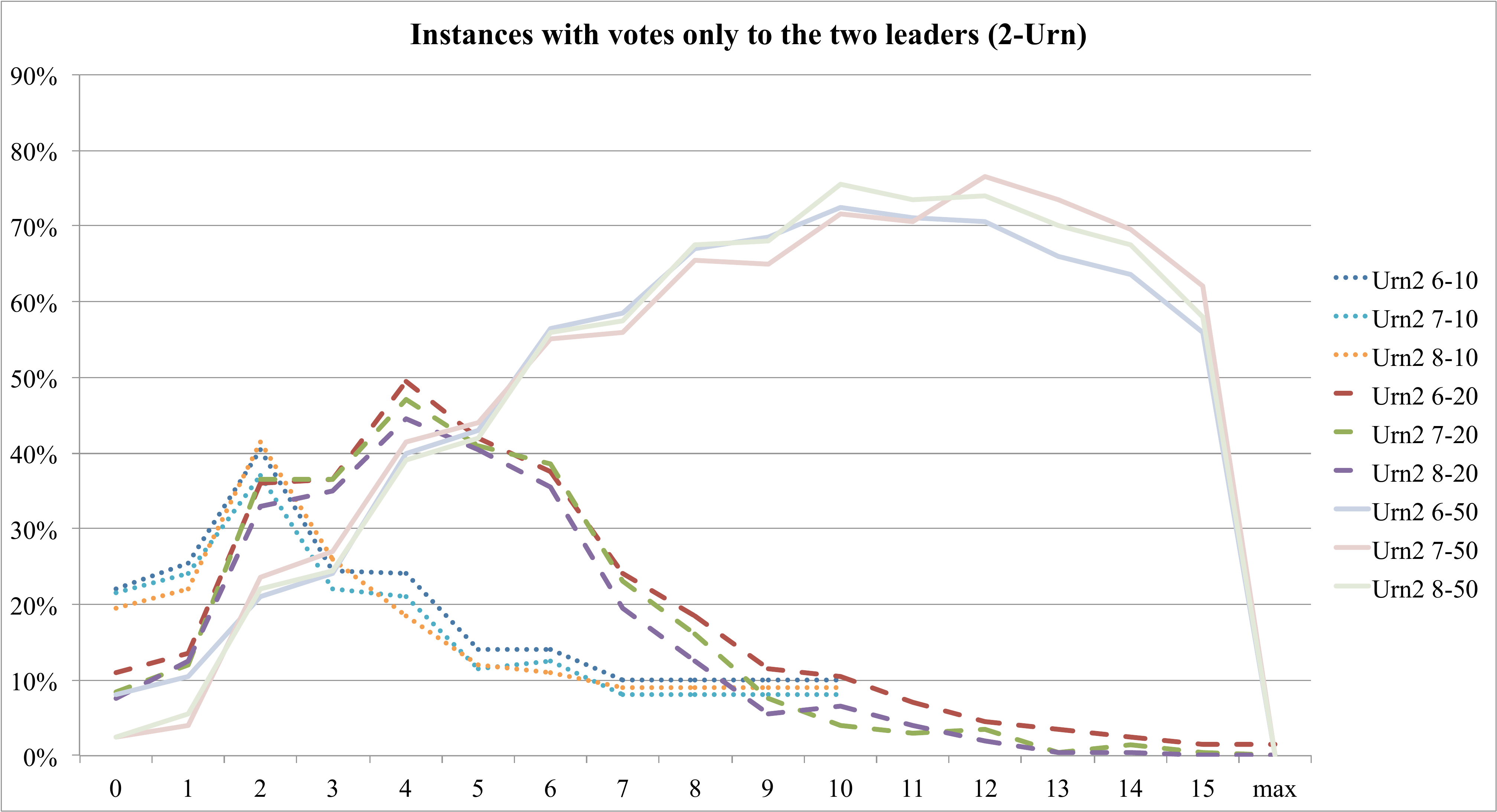}
 \end{center}
\caption {The fraction of simulations in 2-Urn distribution in which \emph{all} voters ended up voting for only two candidates (\emph{TotalDuverger}).\label{fig:duv_all}}
\end{figure}

\newpar{Real preference datasets}
In general, all of our empirical findings were  replicated on the real preference data, but since each instance has a different number of voters and candidates, results are more qualitative.
 
In all three German election datasets we observed similar patterns as above: nearly all supporters of the third candidate deserted it to join one of the leaders. In two datasets this did not change the identity of the Plurality winner. In the third dataset the strategic behavior (for any $r$ between $1$ and $30$) replaced the Plurality winner with the Condorcet winner. 

Similarly, in most of the PrefLib datasets there was a clear winner, and thus there were none-to-few  strategic moves. Votes were typically already quite concentrated for the two leaders in the initial truthful profile, but this concentration increased with strategic activity (Duverger's law).  In the few instances where the identity of the winner changed, it usually replaced the Plurality winner with the Condorcet winner.

%\olev{returned with change}\rmr{not so clear. just say how it supports robustness. move simulations details to the appendix}
\newpar{Non-truthful starting profile and Discriminative power}
 We ran a batch of simulations starting from a voting profile chosen uniformly at random. While we observed voters sometimes switching back and forth between candidates, and despite having no formal guarantee of convergence, all simulations eventually converged to an equilibrium.

%Although in the initial profile every candidate gets roughly a fraction of $1/m$ of the votes, most simulations (around peak $r$) resulted in the same winner, demonstrating both robustness and high discriminative power (see Figure~\ref{fig:riffle_consistency}).  
%In all instances the there was a high concentration, with mostly $40\%-60\%$ (and with small $r$ values --- even much higher percentage) of the simulations converging to the same winner. However, the difference within the distributions were, unsurprisingly, dominant in this case. While for the uniform distribution as $r$ grew the probability of a each candidate to win was almost equal (as few voters deviate), for distributions where many voters have a shared basis (e.g., Placket-Luce), this trend was less pronounced, as all those who deviate (those who voted for their least favorite candidate) changed to the same option. Of course,
We generally observe that the dependency of various attributes in $r$ is more complex (there is no clear ``peak $r$''), but $r$ and $n$ are still the most significant factors. For example, we see much more strategic activity with $r=0$ as in a random state there are more likely to be many pivotal voters than in the truthful state.
 
Other properties observed above such as  Duverger's law, and an increase in winner quality are replicated when starting from a random profile (see e.g. Figure~\ref{fig:Luce_ground_random} in appendix).

More importantly, simulations with random initial states enable us to test the discriminative power of the model.\footnote{Looking only on simulations from the truthful profile does not mean much in that respect. Note for example that with $r=n$, we would get ``perfect'' discriminative power as the truthful Plurality winner is always selected.} Without strategic behavior, we would not expect any candidate to win in a large fraction of the instances (e.g., in the uniform distribution every candidate should win in $\sim 1/m$ of the instances). However when voters are strategic we get that most of the simulations select the same winner regardless of the initial state, which indicates high discriminative power (Figure~\ref{fig:riffle_discriminate}).      

\begin{figure}[ht]
\begin{center}
\includegraphics[scale=0.35]{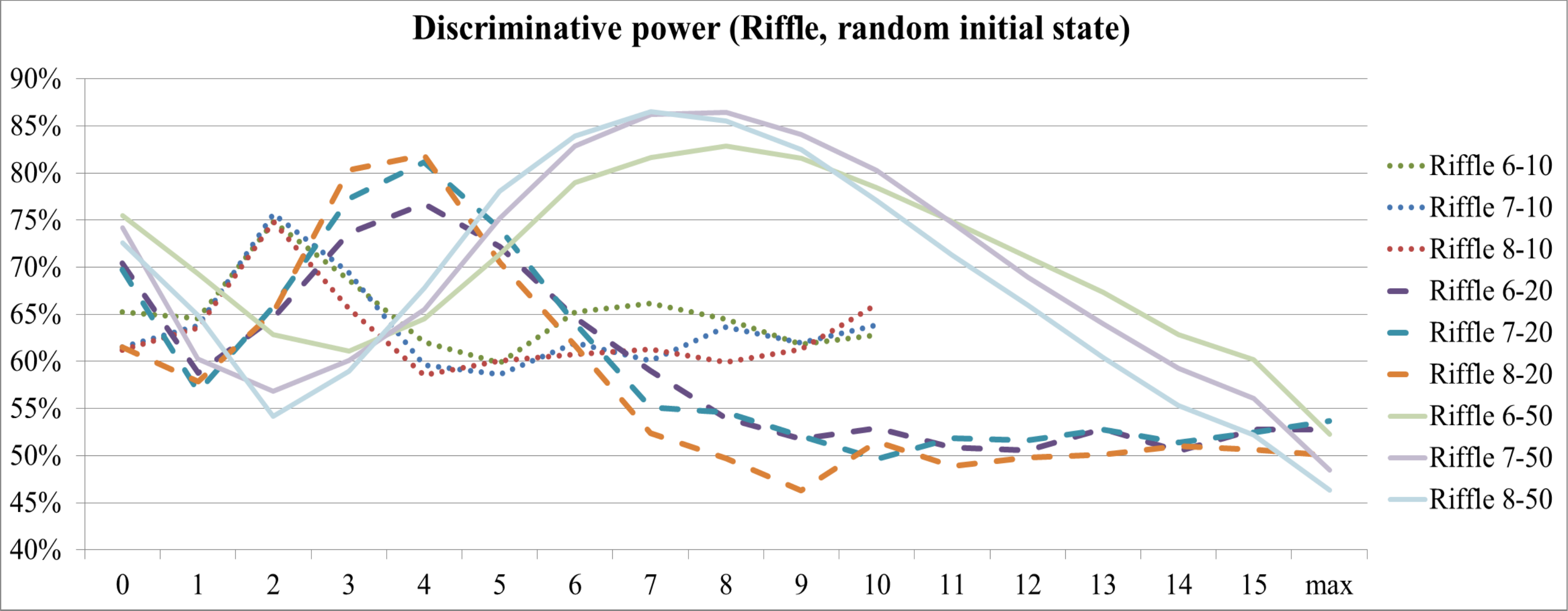}
 \end{center}
\caption {The fraction of simulations (out of 100 random initial states for each preference profile) that ended with the same winner (\emph{WinnerConsistency}). Note that around peak $r$ we have $75\%-85\%$ of the starting points leading to the same winner, regardless of the number of candidates. Very similar patterns were observed for the Urn models and Placket-Luce.} %For the Placket-Luce distribution, distance from ground truth
\label{fig:riffle_discriminate}
\end{figure}

\newpar{Diverse population}Finally, we repeated some of our simulations with heterogeneous voters, where $r_i$ are sampled uniformly i.i.d. from $\{0,1,\ldots,n/m\}$.
Despite the fact that our convergence proofs do not cover heterogeneous populations, convergence was just as robust. Not only that all simulations converged, typically all invariants that we proved for the homogenous case (e.g., that a voter always compromises for less desirable candidates) also hold in the diverse case.\footnote{There were as few as 4 violations of the invariant out of $\sim 120,000$ strategic steps. \rmr{Omer: please revise this footnote}} 

The diverse simulations replicated nearly all the patterns of strategic voting across all distributions. Notably, although we used the same simple distribution of $r_i$ values in all simulations, effects of strategic behavior were always approximately as strong as in the peak $r$ value of every profile distribution and across most measured properties. %---and only in convergence time was there a consistent longer time for simulations with diverse population. \rmr{this is not true according to the bar charts}

Winner quality was also generally comparable to peak $r$.\footnote{Interestingly, the Placket-Luce distribution is an exception, where diverse population led to degradation in the winner quality according to the ground truth, but not according to the other measures.}
 
Regarding Duverger's law, while the number of votes to the top 2 candidates was generally quite similar to the one in peak $r$ (sometimes even higher), with diverse population there were much \emph{fewer} instances where \emph{only} two candidates received votes. See Figure~\ref{fig:divS3_all}. Looking at a typical equilibrium profile reveals that it has a much more ``natural'' dispersion, with many voters voting for the two leaders, but also some voters (with either very high or very low uncertainty values) voting for other candidates. 

\begin{figure}[ht]
\begin{center}
\includegraphics[scale=0.38]{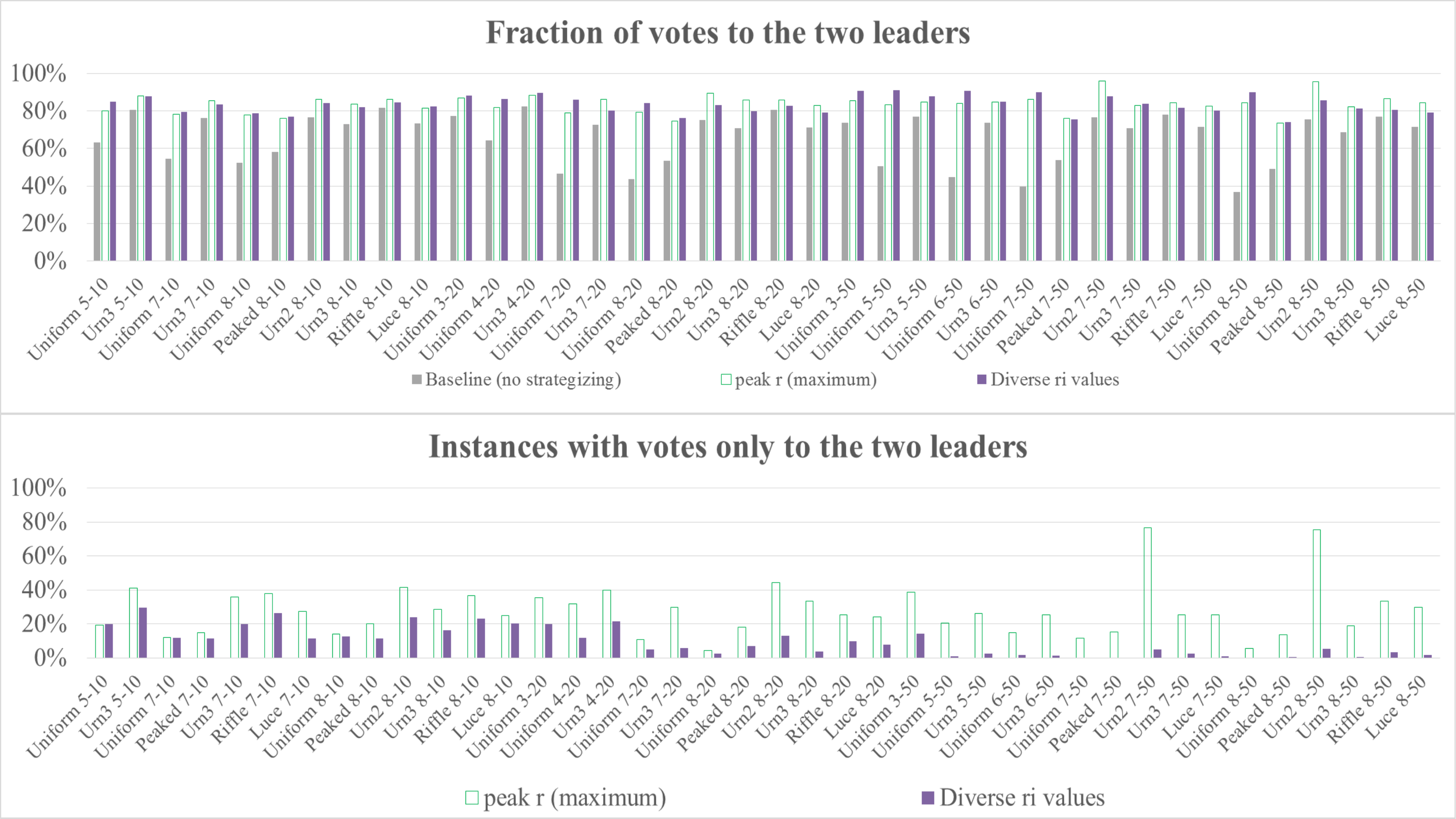}
 \end{center}
\caption {In the top figure we can see that with diverse population, votes were just as concentrated as with fixed population with peak $r$, across all distributions. The bottom figure shows that with fixed $r$, this concentration is due to many instances where only two candidates get votes, while this is not the case with diverse population.
\label{fig:divS3_all}}
\end{figure}

\section{Discussion}
In \cite{abramson1992sophisticated}, sophisticated (strategic) voting based on expected utility maximization is defended on the grounds that it ``...is a simplification of reality that 
seeks to capture the most salient features of actual 
situations. Many voters may see some candidates as 
having real chances of winning and others as likely 
losers, and they may weigh these perceptions against 
the relative attractiveness of the candidates.''

Our theory is also a simplification of reality, and applies similar logic to explain and justify strategic voting. However, the local-dominance approach allows voters to take into account both ``chances of winning'' and ``relative attractiveness'', without regressing to probabilistic calculations and expected utility maximization. 

\subsection{The model and the desiderata}
\label{sec:des_eval}
We summarize by showing how model of local dominance answers to the desiderata we presented in Section~\ref{sec:desiderata}.

\begin{itemize}
	\item Looking at \emph{theoretical criteria}, our model is grounded in traditional game-theoretic concepts: voters are trying to maximize their utility, and results are in equilibrium. Further links to decision theory and classical notions of rationality are detailed in Section~\ref{sec:epistemic}. 
	
When all voters are of the same type,  an equilibrium always exists, and convergence of local-dominance dynamics is guaranteed under rather week conditions. 
	Our simulations show existence and convergence even without these conditions, and demonstrate high discriminative power.
	The model is broad enough to encompass different scenarios such as simultaneous, sequential and iterative voting, and to account for behaviors such as truth-bias and lazy-bias. 
	Our definitions could be easily extended to other positional scoring rules, although it is an open question whether our results would still hold.\footnote{Extensions to other voting rules are simple with metrics like EM distance, but may be ill-defined with other metrics.}
%
%An additive (e.g., $\ell_{1}$) or EM one, which rely on changing voters are easily adaptable to any anonymous voting rule. However, multiplicative distance is less well defined when dealing with non-scoring-rule based voting mechanisms.
Furthermore, as Proposition~\ref{th:general_eq} shows, preposterous equilibria are unlikely. %, and our model combines well-established concepts from game-theory, decision theory, and logic (see next section).  %Barring it being the opening state, voters will never reach an equilibrium where all of them vote for the same, least favorite, candidate, a situation which is completely eliminated even if this was the starting positions in the truth/lazy-biased variants of the model.

%Examining the behavioral criteria of our model, voters' knowledge is quite limited, as they only share the current state (or estimation of state) of each candidate --- they do not know the votes or preferences of particular voters, unlike various manipulation models which rely on full knowledge of the world. Moreover, unlike analysis that assumes mathematical optimization from voters, our model requires relatively little mathematical complexity from each voter, and as in regular manipulation research, in voting rules for which even a simple change may be computationally complex, simplified manipulations could be analyzed (e.g. the model suggested in \cite{BFLR12}).\olev{Unsure if the last sentence here should be included -- it may seem odd with too little explanation}

\item As argued above, voters in our model fit the \emph{behavioral criteria} we posed, as they avoid complex complex computations. Moreover, as Lemma~\ref{lemma:threshold} shows, voters do not even need to consider the entire space of possible states, but merely to check which candidates have sufficient score to become possible winners. Our informational assumptions are rather weak and plausible, as we argue in the end of Section~\ref{sec:intuition}.

%In order to examine the relationship between the model and human interaction, as well as to see if we can reproduce common voting phenomena, we checked empirically (Section~\ref{simulationSection}) some common models on human voting, and the model seemed rather consistent with them.

\item Regarding the \emph{scientific criteria}, once we set the distance metric, every voter can be described by a single parameter (two in the case of lazy or truth-biased voters), which has a clear interpretation as her certainty level. 
Our extensive simulations demonstrate robustness to the order in which voters play (including whether they act simultaneously or not), and that changing the parameters results in a rather smooth transition.
Simulations also show that the model replicates patterns that are common in the real world such as the Duverger Law, and resulting equilibria, especially with diverse population, seem reasonable.
Experimental validation as outside the scope of this work.
\end{itemize}

Finally, it is shown that strategic behavior yields a better winner for the society according to various measures of quality (compared to the truthful Plurality winner).

\subsection{Epistemic foundations and rationality}
\label{sec:epistemic}
We can phrase dominance relations in terms of modal logic. Consider a Kripke structure over states where $S$ are the states accessible from $\vec s=\vec s_{\vec a}$. Then ``$b_i$ $S$-beats $a_i$ in state $\vec s$'' can be written as $\vec s \models\diamond (f(b_i,\vec s) \succ_i f(a_i,\vec s))$. Similarly, ``$b_i$ $S$-dominates $a_i$'' means $\vec s \models\square (f(b_i,\vec s) \succeq_i f(a_i,\vec s)) \wedge \diamond (f(b_i,\vec s) \succ_i f(a_i,\vec s))$ ($b_i$ is necessarily at least as good, and possible a better action than $a_i$). We note that $S_i(\vec a,k)$ defines a Kripke structure that is reflexive and symmetric but non-transitive.

A common semantic interpretation of the modal operator $\square P$ is ``$P$ is \emph{known}''.\footnote{An alternative notation $K_i P$ is sometimes used for the statement ``$P$ is known to agent~$i$''. See, for example, \cite{aumann1999interactive}.} According to this, we can naturally interpret ``$b_i$ $S$-beats $a_i$ in state $\vec s$'' as ``in state $s$, $i$ does not know that voting for $a_i$ is at least as good as voting for $b_i$''. Thus ``$b_i$ $S$-dominates $a_i$ in state $\vec s$'' means that $i$ knows that $b_i$ is at least as good as $a_i$, but does not know that $a_i$ is at least as good as $b_i$.

%[compare with Reingould thesis. there a voter has a partial information on the real profile (e.g., identity of the winner), and infers all consistent profiles as possible worlds]

In our model, $S(\vec a)$ is not uniquely defined, and in fact even the same voter uses both $S_i(\vec a,r), S_i(\vec a,k)$, where $r<k$. Since $S_i(\vec a,r) \subseteq S_i(\vec a,k)$, a straight-forward extension of the epistemic interpretation is to add \emph{certainty levels}, where a larger (in terms of containment) set of accessible states indicates lower certainty. % Let us use the terms ``\underline{k}nows'' for the larger set of worlds $S(\vec a,k)$, and ``\underline{r}eally-sure'' for the smaller set $S(\vec a,r)$. Then a voter can keep his non-truthful vote if he merely knows it is better than voting truthfully, but he has to be really sure that a strategic move will not hurt him.

\newpar{Local dominance and rationality}
According to the standard non-Bayesian incomplete information model (due to Aumann~\shortcite{aumann1995backward,aumann1999interactive}), a player $i$ playing strategy $a_i$ at some world state $\vec \pi$ is \emph{rational}, if there is no other strategy $a'_i$ that yields a same or better outcome in all states accessible from $\pi$ (and in some states strictly better).%\footnote{Aumann focuses on extensive form games, whereas our game is always interpreted by players as a single-shot game. Since single-shots games are a special case of extensive form, Aumann's definitions apply.\rmr{remove comment?}\olev{IMHO, yes}}

In other words, rationality under strict uncertainty according to Aumann simply means that players avoid locally dominated strategies. Voting equilibria in our model are therefore rational (Prop.~\ref{th:eq_dominated}). 
Our model is more specific in that it specifies a particular dynamic of how voters act when their current strategy is dominated. 
%However, observe that the set of pure equilibria depends only on the accessibility relations $(S_i)_{i\in N}$ (and on the fact that players are rational); and not on the details of the dynamics (e.g., selection of the most preferred candidate in the set). 

Another difference is that in Aumann's models the accessibility relation is a partition, and in particular transitive. Other papers such as \cite{Bicchieri95} do not make any assumptions on the accessibility relation other than consistency. In our model the relation is based on distance, and in particular it is non-transitive (if $\pi$ is close to $\pi'$, and $\pi'$ is close to $\pi''$, then it may not hold that $\pi$ is close to $\pi''$).
%Then i would strategically move to a'_i if she presumes that a'_i is at least as good as a_i (and does not presume that a_i is at least as good as a'_i). Similarly, i will resort to truth if she knows that t_i is at least as good as a_i. 

\newpar{Local dominance and voting}
Dominance within a restricted set of states was considered by several recent papers.  
In \cite{reijngoud2012voter,van2013strategic} the assumption is that voters information sets can be described as a partition $\Pi$, as in the Aumann model. Reijngoud and Endriss say that a voter \emph{has an incentive to $\Pi$-manipulate using ballot $P'_i$} (under profile $\vec P$), if she weakly gains by voting $P'_i$ in every state that is ``equivalent'' to $\vec P$ according to $\Pi$.\footnote{The definition is for arbitrary voting rules.} In the special case of Plurality, the definition coincides local dominance: Consider Def.~\ref{def:local}, where we set $S$ to be all states equivalent to $\vec a$ under $\Pi$. Then $a'_i$ $S$-dominates $a_i$ iff $i$ has an incentive to $\Pi$-manipulate using ballot $a'_i$. In the terminology of \cite{van2013strategic}, voter $i$ \emph{knows `de re' that she can weakly successfully manipulate}.
Our definition of local dominance also coincides with the definition of dominance in \cite{conitzer2011dominating}, which do not make any assumption on the ``information set'' $S$. 

 In our work the accessibility relation is defined by a distance metric and is \emph{not} a partition. Still, many of the definitions in \cite{reijngoud2012voter,van2013strategic} can be applied just the same in our case. In particular, a combination of these works can be used to extend the notion of local dominance to other voting rules. 

\newsubsec{Conclusion and future work}

We see a unifying theory as the one we present as a productive step in the quest to understand voting. We hope that future researchers will find our theoretical framework useful for formulating new, more specific, voting behaviors. Furthermore, our particular distance-based model can serve as a strong baseline for competing theories. Experiments with human voters will be important to settle how close each of these theories comes in adequately describing human voting behavior. 

On the technical level, we conjecture that stronger convergence properties can be proved; in particular, that there are no cycles in voting games with voters of the same type, and that a voting equilibrium exists even in games with heterogeneous voters.

We also believe that distance-based local dominance, with the necessary adaptations, can provide a useful non-probabilistic framework for uncertainty in other classes of games where there are natural distance metrics over states, such as congestion games. 

Finally, insights based on our theory, for example on how voters' uncertainty level affects quality of the outcome, can be useful in designing better voting mechanisms. 

\section*{Acknowledgments}
We thank Joe Halpern, David Parkes, Yaron Singer, Greg Stoddard, and several anonymous reviewers for their valuable feedback.

\appendix

\section{Proofs}
\label{sec:proofs}

\begin{rlemma}{lemma:threshold}
	Each of the metrics $\delta$ from ($\ell_1,\ell_\infty$, multiplicative) induces a function $\beta= \beta_{\delta,r}:\mathbb N \rightarrow \mathbb N$, where
	
	\begin{itemize}
	\item For every $\vec a, i\in N$, if $a_i\neq f^*_i$, then $W_i(\vec a,r)=\{c: s_{\vec a_{-i}}(c) \geq_{\ol Q} \beta( s(f^*_i))\}$. That is, possible candidates are all those whose score is above the threshold, which is a function of the score of the winner.\footnote{We break ties with $\beta(s(f^*_i))$ as if we break ties with $f^*_i$.} 
\item $\beta(s)$ is weakly increasing in $s$. 
\end{itemize}
 In particular, for $\delta_{\ell_1}$, $W_i(\vec a,r) = \ol H_{r+1}(\vec s_{\vec a})$ for any $i$ s.t. $a_i\notin \ol H_{r+1}(\vec s_{\vec a})$. Similarly, for $\delta_{\ell_\infty}$, $W_i(\vec a,r) = \ol H_{2r+1}(\vec s_{\vec a})$ for any $i$ s.t. $a_i\notin H_{2r+1}(\vec s_{\vec a})$.
	\end{rlemma}
	A similar result can be proved for EMD and other $\ell_d$ norms, but the threshold would depend on the score of all candidates and not just the winner. 
	\begin{proof}
		Consider first the $\ell_1$ norm. We set $\beta(s) =\beta_{\ell_1}(s)= s-r-1$. Clearly, $f^*_i$ is a possible winner and $s(f^*_i)>\beta(\vec a)$. Assume $c\neq f^*_i$, then $s(c)\neq_{\ol Q} \beta(s(f^*_i))$.
	If $s(c) > \beta(s(f^*_i))$, consider the state $\vec s'\in S_i(\vec a,r)$ where $c$ has $r$ additional votes. We have that 
	$$s'(c)+1= s(c)+r+1 > \beta(s(f^*_i)) +r+1 = s(f^*_i),$$
	thus $f(\vec s',c)=c$.
	In contrast, if $s(c) <\beta(s^*_i)$, then in any $\vec s' \in S_i(\vec a,r)$, $s'(f^*_i) - s'(c) >1$ and thus $c$ cannot win. 
	Finally, 
	$$W_i(\vec a,r) = \{c: s(c)\geq \beta_{\ell_1}(s(f^*_i))\} = \{c : s(c) \geq s(f^*_i)-(r+1)\}= \ol H_{r+1}(\vec s).$$
	
	Similarly, for $\ell_\infty$ we set $\beta(s) = \beta_{\ell_\infty}(s) = s-2r -1$. The critical state $\vec s'$ is where $c$ gets $r$ additional votes, and we subtract $r$ votes from all other candidates (including $f^*_i$).
	
	For the multiplicative distance, we set $\beta(s) = \ceil{ \ceil{\frac{s}{1+r}} / (1+r) } -1 $. In the critical state we multiply the score of $c$ by $(1+r)$, getting $s'(c)= \floor{s(c)(1+r)}$, and divide the score of all other candidates by $(1+r)$, so e.g. for $f^*_i$, we get $s'(f^*_i) = \ceil{\frac{s(f^*_i)}{1+r}}$. Thus 
	\begin{align*}
	&f(\vec s',c)=c &\iff\\
	&s'(c) +1 > s'(f^*_i) &\iff \\
	&\floor{s(c)(1+r)} +1 > \ceil{\frac{s(f^*_i)}{1+r}} &\iff \\
	&s(c) > \ceil{ \ceil{\frac{s(f^*_i)}{1+r}} / (1+r) } -1 = \beta(s(f^*_i)).
	\end{align*}
	\end{proof}

\begin{theorem}
	\label{th:basic_distances}
Suppose that all voters are of type $r$ and using any norm for which Lemmas~\ref{lemma:best_winner},\ref{lemma:threshold} apply. Then a voting equilibrium exists. Moreover, in an iterative setting where voters start from the truthful state, for any singleton scheduler, they will always converge to an equilibrium in at most $n(m-1)$ steps. 
\end{theorem}
\begin{proof}
%We prove for the $\ell_1$ norm, and later explain how to extend to other metrics.
If the truthful state $\vec q$ is stable, then we are done. Thus assume it is not.
Let $\vec a^t$ (and $\vec s^t$) be the voting profile after $t$ steps from the initial truthful vote $\vec a^0 = \vec q$.
Let $a_i\step{i} a'_i$ be a move of voter $i$ at state $\vec s=\vec s_{\vec a^t}$ to state $s'=\vec s_{\vec a^{t+1}}$. 

%By Lemma~\ref{lemma:best_winner}, $a'_i$ is always $i$'s most preferred possible winner, and therefore the most preferred candidate in $\ol H_{r+1}(\vec s)$). 
We claim that the following hold throughout the game.
%Also, We make the following claims:
\begin{enumerate}
	\item $a_i\notin W_{i}(\vec a,r)$. Voters only leave non-possible winners.
	\item After a step $a_i \step{i} a'_i$ at time $t$, $a_i \notin W_j(\vec a^{t'},r)$ at any later time $t'$, for any voter $j$.
	\item $a'_i \prec_i a_i$. I.e., voters always ``compromise'' by voting for a less preferred candidate.
		\item $\max_{a\in A} s'(a) \geq \max_{a \in A} s(a)$. I.e., the score of the winner never decreases.
	\item For all $j$, all $t'>t$ after the first step of $j$, $W_j(\vec a^{t'},r) \subseteq W_j(\vec a^{t},r)$. I.e., the set of possible winners can only shrink (after the first move).
\end{enumerate}
We prove this by a complete induction. 
%Note that we only need to show the base case for (3), as all other properties vacuously hold for all steps that preceded the first step. Further, if (3) does not hold at $\vec s^0= \vec s_{\vec q}$, then $\vec q$ is already stable.
\begin{enumerate}
	\item If this is the first move of $i$ then $a_i=q_i$. Otherwise, by Lemma~\ref{lemma:best_winner}, $a_i$ is the most preferred candidate of $i$ in $W_i(\vec a^{t'},r)$ where $t'$ is the time when $i$ last moved. By induction on (5), $W_i(\vec a,r) \subseteq W_i(\vec a^{t'},r)$. So either $a_i$ is not a possible winner in $\vec a$ (and the we are done), or it must be the most preferred candidate in $W_i(\vec a,r)$. Assume, toward a contradiction, that $a_i \in W_i(\vec a,r)$, then there is a state in $S_i(\vec a,r)$ where $i$ is pivotal between $a_i$ and $f(\vec a)$. With any other action, $f(\vec a)=c \prec_i a_i$ would win. Therefore $a_i$ $S_i(\vec a,r)$-beats any other candidate including $a'_i$. In particular, $a'_i$ does not $S_i(\vec a,r)$-dominate $a_i$, which is a contradiction. 
	%We observe that the set $D$ from Definition~\ref{def:strategic} is a subset of $\ol H_{r+1}(\vec s)$, since no other candidate $S_i(\vec a,r)$-beats $a_i$. Let $a^*$ be $i$'s most preferred candidate in $\ol H_{r+1}(\vec s)$. If $a^*\neq f(\vec s)$, then $D$ contains all candidates in $\ol H_{r+1}(\vec s)$ that are preferred over the winner, and in particular $a^*$.
	%If $a=f(\vec s)$, then either $D$ is empty (in which case there is no move), or $D=\{a^*\}$. In either case, $a'=a^*$.
	\item The scores by which different voters determine possible winners are almost identical, and the score of $a_i$ may differ by at most $1$ vote between $\vec a_{-i}$ and $\vec a_{-j}$. When $i$ moves then by (1) and Lemma~\ref{lemma:threshold} $s_{\vec a_{-i}}(a_i)< \beta(s(f(\vec a)))$, and thus $s_{\vec a_{-j}}(a_i)< \beta(s(f(\vec a)))+1$. Thus while $j$ may still consider $a_i$ as a possible winner \emph{before} $i$ moves, $a_i$ is no longer a possible winner for $j$ after $i$ moves, as 
	$$s'_{\vec a_{-j}}(a_i) = s_{\vec a_{-j}}(a_i)-1 < \beta(s(f(\vec a))).$$
	
	\item If this is the first move of $i$ then this is immediate. Otherwise, by induction on Lemma~\ref{lemma:best_winner} and (5), if $a'_i \succ_i a_i$, then $i$ would prefer to vote for $a'_i$ in his previous move, rather than to $a_i$.
	 
\item As in (1), if $i$ votes for $a_i=f(\vec a)\in \ol H_{r+1}(\vec a)$, then $a_i$ is $i$'s most preferred possible winner. Thus it cannot be locally dominated by any other candidate.
	
	 \item Since by (3) the score of the winner never decreases, the only way to expand $W_j(\vec a,r)$ is if some voter $i$ added a vote to a candidate not in $W_j(\vec a,r)$. Recall that by Lemma~\ref{lemma:best_winner}, $i$ only votes to candidates in $W_i(\vec a,r)$. Since $W_j(\vec a,r),W_i(\vec a,r)$ differ only by the votes of $i$ and $j$, the only candidate in $W_i(\vec a,r) \setminus W_j(\vec a,r)$ can be $a_j$ (the current vote of $j$). 
	
	However if $j$ already moved once then $a_j$ was a possible winner for $j$. Since at time $t$ $a_j\notin W_j(\vec a,r)$, then some voter must have deserted $a_j$ before time $t$. Then by (2) no voter would consider $a_j$ a possible winner after time $t$, and by Lemma~\ref{lemma:best_winner}, no voter would vote for it before $t'$. 
	
	Note that if $j$ have never moved, then it is possible that $q_j$ is not a possible winner for $j$, but still gets a vote later from $i$. 
\end{enumerate}
Finally, by property (2), each voter moves at most $m-1$ times before the game converges. 
 %The proof holds for any other distance measure $\delta$ for which Lemma~\ref{lemma:move} applies. We simply need to replace $\ol H_{r+1}(\vec s)$ with the set of candidates with score at least $\beta_{\delta}(s^*_i)$. A slight variation of the proof also holds for any $\ell_d$ norm, $d\geq 1$ and for EMD. It still holds that every strategic move weakly increases the threshold, and thus the set of possible winners only shrinks. \rmr{maybe extend the lemma to the other metrics}
\end{proof}

\begin{rlemma}{lemma:two_winners}
Under the conditions of Theorem~\ref{th:basic}, either $\vec q$ is stable, or in every state $\vec s^t$ we have $|\ol H_{r+1}(\vec s^t)|>1$. Also, in the last state $\vec a$ either $|\ol H_{r}(\vec s_{\vec a})|=1$, or all voters vote for possible winners. Any voter voting for $c\notin \ol H_{r+1}(\vec s^t)$ prefers $f(\vec a)$ over any other candidate in $\ol H_{r+1}(\vec s^t)$.
\end{rlemma}
\begin{proof}
Note first that once $|\ol H_{r+1}(\vec s)|=1$, there are no strategic moves, as no candidate can challenge the winner. Thus a violation can occur only in the last step. Assume, toward a contradiction, that in the last step $a'_i=f(\vec a)$, and $|\ol H_{r}(\vec s)|=1$ (but $|\ol H_{r+1}(\vec s)|>1$). However, since by (1) $a_i\notin \ol H_{r+1}(\vec s)$, each candidate in $\ol H_{r+1}(\vec s)\setminus \{f(\vec a)\}$ has the same score with and without $i$, and this score is at most $s(f(\vec a))-(r+1)$. Therefore $f(\vec a)$ also wins in any state in $S_i(\vec a,r)$ (no candidate is a threat to the winner). This means that $f(\vec a)$ does not $S_i(\vec a,r)$-beat $a_i$, in contradiction to a strategic move where $a'_i=f(\vec a)$. 

Finally, suppose that there is more than one candidate in $\ol H_{r}(\vec s_{\vec a})$. Then any voter $i$ not voting for a possible winner sees himself potentially pivotal between the winner and the runner-up (there is a possible state where the runner-up wins if $i$ keeps voting for $a_i$). Since $i$ always strictly prefer on of them, this candidate will locally dominate $a_i$. 

We can further see that if $|\ol H_{r}(\vec s_{\vec a})|=1, |\ol H_{r+1}(\vec s_{\vec a})|>1$, then any voter not voting for $\ol H_{r+1}(\vec s)$ in equilibrium must favor the winner. Otherwise he would be potentially pivotal between the current winner and a better possible winner, and thus his most preferred possible winner would locally dominate his current vote.
\end{proof}

\begin{rproposition}{th:group}
Suppose that all voters are of type $r$. 
Consider any group scheduler such that (1) any voter has some chance of playing as a singleton (i.e. this will occur eventually); (2) The scheduler always selects (an arbitrary subset of) voters with type~2 moves, if such exist. 
Then convergence is guaranteed from any initial state after at most $O(nm)$ singleton steps occur.
\end{rproposition}
\begin{proof}
%For any $w\in \mathbb N$, let $H_w(\vec s)\subseteq M$ be the set of candidates that need $w$ more votes to become the winner. Thus $H_0(\vec s)=\{f(\vec s)\}$, $H_1(\vec s)$ are all runner ups (either have the same score as the winner and lose by tie-breaking, or win tie breaking but have one vote less), etc. Let $\ol H_w(\vec s) = \bigcup_{w'\leq w} H_{w'}(\vec s)$. 
Denote $N_w(\vec a) = \bigcup_{a\in \ol H_{w}(\vec a)}\{i:a_i=a\}$ i.e. all voters voting for the top $w$ candidates.
 
%Let $\vec a^t$ (and $\vec s^t$) be the voting profile after $t$ steps from the initial truthful vote $\vec a^0$.
%Let $a_i\step{i} a'_i$ be a move of voter $i$ at state $\vec s=\vec s_{\vec a^t}$ to state $s'=\vec s_{\vec a^{t+1}}$. 
 Note that there can be at most $n(m-1)$ consequent moves of type~2, regardless of the scheduler. Let $\vec a^{t^*}$ be the state first reached when no voter has a type~2 step.
If the state $\vec a^{t^*}$ is stable, then we are done. Thus assume it is not.

Consider a group $N'$ that moves at time $t$. By property (2) of the scheduler, either all of $N'$ has type~1 moves, or all of it has type~2 moves. Thus we can classify all (group) steps to type~1 and type~2.

We define a \emph{chunk potential function} $\alpha$, where
$\alpha(\vec a) = -n\times |\ol H_{r+1}(\vec a)| + |N_{r+1}(\vec a)|$.
That is, the potential increases as the set of possible winners shrinks, but for a fixed size it increases with the total mass of voters for such candidates.

Consider a sequence of moves after time $t^*$. A \emph{chunk} is composed of a type~1 step and all type-2 steps that follow until the next type~1 step (i.e. a type~1 step and then zero or more type~2 steps). We claim that (a) after every chunk $\alpha$ (as well as the score of the winner) does not decrease; (b) after a finite number of chunks $\alpha$ strictly increases.

We observe that in any type~1 move $\vec a \step{N'} \vec a'$, $a'_i = \argmin\{Q_i(a) : a \in\ol H_{r+1}(\vec s)\}$ for all $i\in N'$.
%
%Assume toward a contradiction that a cycle exists, and let $t^{**}$ be the time after all voters that are active in the cycle played at least one. 
We will prove by induction that starting from $t^{*}$, after every chunk $\vec a \step{N'} \vec a' \step{N''} \cdots \step{\hat{N}} \hat{\vec a}$:
\begin{enumerate}
%\item for any voter $i\in N$, $a_i \succeq_i f(\vec a)$ (or there are no moves). %in $N_{r+1}(\hat{\vec a})$ vote for a candidate in $H_{r+1}(\hat{\vec a})$ that is at least as good as the winner. 
	\item $\max_{a\in M} \hat s(a) \geq \max_{a \in M} s(a)$. I.e., the score of the winner never decreases.
	\item $\alpha(\hat{\vec a}) \geq \alpha(\vec a)$.
\end{enumerate}
We prove this by a complete induction. 
%We first need to show the base case for (1). This is simply since otherwise there is a voter $i$ s.t. voting for the winner dominates his current vote, and thus has a type~2 move. This also proves that (1) holds after every chunk. 

Consider the type~1 move of the chunk. All voters vote for less preferred candidates, thus for all $i\in N'$, $s(a'_i) \geq s(a_i)$. Moreover, a voter only moves if $a_i\notin \ol H_r(\vec a)$, since otherwise by he is potentially pivotal (as we show in the proof of Th.~\ref{th:basic}). Case I: for any $i\in N'$, there is no $j$ s.t. $a'_j=a_i$. In this case all moves are essentially independent, and every $i$ s.t. $a_i\notin \ol H_{r+1}(\vec a)$ increases is added to $N_{r+1}(\vec a')$,thereby increasing $\alpha$ by $1$. Every voter $i$ s.t. $a_i\in H_{r+1}(\vec a)$ increases $\alpha$ by $n$, since $a_i \notin \ol H_{r+1}(\vec a')$. Thus in case~I the chunk potential strictly increases, and there are no followup type~2 steps (the chunk ends).

The problematic case is where there is $I\subseteq N'$, s.t. for all $i\in I$ there is $j\in N'$ with $a'_j = a_i$ (see Figure~\ref{fig:group}). If it also holds that $a_i\in H_{r+1}(\vec a)$, then after the move we have $a_i\in \ol H_{r+1}(\vec a')$. Then we have that while $a_i$ was not a possible winner for $i$ before the step, it is a possible winner for $i$ after the step. Note however that only voters who moved can have a new possible winner, and it can only be the candidate they deserted. %On the other hand, $a'_i$ to which $i$ moved may now \emph{no longer} be a possible winner for $i$, if it was deserted by another candidate. 
%If the two evens occur in conjunction, 
It is then possible that voter $i$ now has a type~2 move. Thus some possibly empty subset $I'\subseteq I$ have type~2 moves, which is to return to their original candidates $a_i$ (deserting other candidates in $\ol H_{r+1}(\vec a')$). After the first type~2 move there may other voters $I''\subseteq I\setminus I'$ that want to return and so on. However no other voter has a new type~2 move since all of $N'\setminus I$ vote for their most preferred possible winner in $\vec a'$; and any $i'\in N\setminus N'$ does not have a type~2 move since they did not have one in $\vec a$ and there are no new possible winners. 

So every further step~2 in the chunk rolls back some of the first type~1 steps. At the end of the chunk we are left with a set $N''=N' \setminus (I'\cup I'' \cup \cdots)$, where $N''$ performed a type~1 step and all other voters vote as in $\vec a$. If $N''\neq \emptyset$ then by the previous case $\alpha$ strictly increases (and the score of the winner does not decrease). Clearly if $N''=\emptyset$ then $\hat{\vec a} = \vec a$ and thus $\alpha$ does not change. However note that the type~1 step is a set of disjoint cycles. For $N''$ to be empty, each of these cycles must be contained completely in $I'$ or $I''$, etc.: a voter that does not roll back his action $a_j \step{j} a'_j$ at the same time with the voter who joined $a_j$, will not be able to roll back at a later type~2 step, since $a_j$ will no longer be a possible winner for $j$. Thus a singleton type~2 move (which cannot be a cycle) means that $N''\neq \emptyset$.

Since eventually there will be singleton move (either type~1 or type~2), the same cycle repeat forever, and $\alpha$ must increase. Clearly it cannot increase more than $nm$ times. 
%\documentclass[]{article}
%\usepackage{tikz}
%\usepackage{subcaption}
%%\usepackage{tikz}%,fullpage}
%\usetikzlibrary{arrows,%
                %petri,%
                %topaths}%
%\usepackage{tkz-berge}
%
%\begin{document}
%
%%%%%%%%%%%%%%%%%%%
%
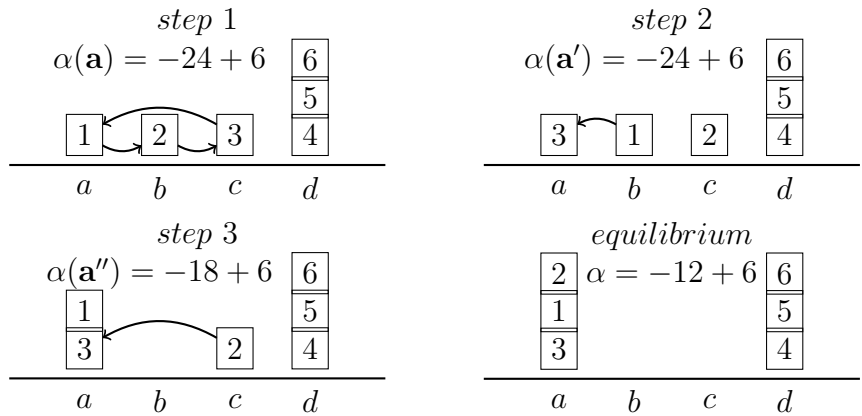
\begin{figure}
\centering
\begin{subfigure}[]{0.4\textwidth}

\begin{tikzpicture}[scale=1,transform shape]
	
	\draw[thick] (0,0.1) -- (5,0.1);
	
  \tikzstyle{VertexStyle}=[draw,black]
	\Vertex[x=1,y=0.5,L=$1$]{v1}
	\Vertex[x=2,y=0.5,L=$2$]{v2}
	\Vertex[x=3,y=0.5,L=$3$]{v3}
	
	\Vertex[x=4,y=0.5,L=$4$]{v4}
	\Vertex[x=4,y=1,L=$5$]{v5}
	\Vertex[x=4,y=1.5,L=$6$]{v6}
	
	 \tikzstyle{LabelStyle}=[fill=white,sloped]
  %\tikzstyle{EdgeStyle}=[dashed,thin,bend right,->]
	 \tikzstyle{EdgeStyle}=[black,bend right,->]
	\Edge[](v1)(v2)
	\Edge[](v2)(v3)
	\Edge[](v3)(v1)

  \tikzstyle{VertexStyle}=[]
	\Vertex[x=1,y=-0.2,L=$a$]{a}
	\Vertex[x=2,y=-0.2,L=$b$]{b}
	\Vertex[x=3,y=-0.2,L=$c$]{c}
	\Vertex[x=4,y=-0.2,L=$d$]{d}
  \Vertex[x=2.5,y=2,L=$step~1$]{step}
	\node at (2,1.5) {$\alpha(\vec a)=-24+6$};
%\
\end{tikzpicture}

\end{subfigure}\quad\quad
\begin{subfigure}[]{0.4\textwidth}
\begin{tikzpicture}[scale=1,transform shape]
	
	\draw[thick] (0,0.1) -- (5,0.1);
	
  \tikzstyle{VertexStyle}=[draw,black]
	\Vertex[x=1,y=0.5,L=$3$]{v3}
	\Vertex[x=2,y=0.5,L=$1$]{v1}
	\Vertex[x=3,y=0.5,L=$2$]{v2}
	
	\Vertex[x=4,y=0.5,L=$4$]{v4}
	\Vertex[x=4,y=1,L=$5$]{v5}
	\Vertex[x=4,y=1.5,L=$6$]{v6}
	
	 \tikzstyle{LabelStyle}=[fill=white,sloped]
  %\tikzstyle{EdgeStyle}=[dashed,thin,bend right,->]
	 \tikzstyle{EdgeStyle}=[black,bend right,->]
	\Edge[](v1)(v3)
	%\Edge[](v2)(v3)
	%\Edge[](v3)(v1)

  \tikzstyle{VertexStyle}=[]
	\Vertex[x=1,y=-0.2,L=$a$]{a}
	\Vertex[x=2,y=-0.2,L=$b$]{b}
	\Vertex[x=3,y=-0.2,L=$c$]{c}
	\Vertex[x=4,y=-0.2,L=$d$]{d}
  \Vertex[x=2.5,y=2,L=$step~2$]{step}
	\node at (2,1.5) {$\alpha(\vec a')=-24+6$};
%\
%\draw[dashed,->] (2,1) -- (2-0.5,1);
%\draw[dashed,->] (3,2) -- (3,2+0.5);

%\caption{\label{sfig:s0}The value for every bidding profile $(b_1,b_2)$ at state $s_0$. White wins in white areas, black wins in grey areas. The dashed arrows mark the winner in case of a bidding profile that is on the boundary between segments.}
\end{tikzpicture}

\end{subfigure}\\
\begin{subfigure}[]{0.4\textwidth}
\begin{tikzpicture}[scale=1,transform shape]
	
	\draw[thick] (0,0.1) -- (5,0.1);
	
  \tikzstyle{VertexStyle}=[draw,black]
	\Vertex[x=1,y=0.5,L=$3$]{v3}
	\Vertex[x=1,y=1.0,L=$1$]{v1}
	\Vertex[x=3,y=0.5,L=$2$]{v2}
	
	\Vertex[x=4,y=0.5,L=$4$]{v4}
	\Vertex[x=4,y=1,L=$5$]{v5}
	\Vertex[x=4,y=1.5,L=$6$]{v6}

	 \tikzstyle{LabelStyle}=[fill=white,sloped]
  %\tikzstyle{EdgeStyle}=[dashed,thin,bend right,->]
	 \tikzstyle{EdgeStyle}=[black,bend right,->]
	\Edge[](v2)(v3)
	
  \tikzstyle{VertexStyle}=[]
	\Vertex[x=1,y=-0.2,L=$a$]{a}
	\Vertex[x=2,y=-0.2,L=$b$]{b}
	\Vertex[x=3,y=-0.2,L=$c$]{c}
	\Vertex[x=4,y=-0.2,L=$d$]{d}
  \Vertex[x=2.5,y=2,L=$step~3$]{step}
	\node at (2,1.5) {$\alpha(\vec a'')=-18+6$};
	
%\draw[->] (2,1) -- (2-0.5,1);
%\draw[dashed,->] (3,2) -- (3,2+0.5);

%\caption{\label{sfig:s0}The value for every bidding profile $(b_1,b_2)$ at state $s_0$. White wins in white areas, black wins in grey areas. The dashed arrows mark the winner in case of a bidding profile that is on the boundary between segments.}
\end{tikzpicture}

\end{subfigure}\quad\quad
\begin{subfigure}[]{0.4\textwidth}
\begin{tikzpicture}[scale=1,transform shape]
	
	\draw[thick] (0,0.1) -- (5,0.1);
	
  \tikzstyle{VertexStyle}=[draw,black]
	\Vertex[x=1,y=0.5,L=$3$]{v3}
	\Vertex[x=1,y=1.0,L=$1$]{v1}
	\Vertex[x=1,y=1.5,L=$2$]{v2}
	
	\Vertex[x=4,y=0.5,L=$4$]{v4}
	\Vertex[x=4,y=1,L=$5$]{v5}
	\Vertex[x=4,y=1.5,L=$6$]{v6}

  \tikzstyle{VertexStyle}=[]
	\Vertex[x=1,y=-0.2,L=$a$]{a}
	\Vertex[x=2,y=-0.2,L=$b$]{b}
	\Vertex[x=3,y=-0.2,L=$c$]{c}
	\Vertex[x=4,y=-0.2,L=$d$]{d}
  \Vertex[x=2.5,y=2,L=$equilibrium$]{step}
		\node at (2.5,1.5) {$\alpha=-12+6$};
%\draw[dashed,->] (2,1) -- (2-0.5,1);
%\draw[dashed,->] (3,2) -- (3,2+0.5);

%\caption{\label{sfig:s0}The value for every bidding profile $(b_1,b_2)$ at state $s_0$. White wins in white areas, black wins in grey areas. The dashed arrows mark the winner in case of a bidding profile that is on the boundary between segments.}
\end{tikzpicture}

\end{subfigure}

\caption{
\label{fig:group}An example of a group scheduler, with $r=1$. Voters~1,2,3 rank $d$ last, and rank their current vote first. In the initial state $\vec s = (1,1,1,3)$, we have $H_{r+1}(\vec s) = \{a,b,c\}$. Step~1 is a compromise move, where $N'=\{1,2,3\}$. After this move, $a,b$ and $c$ are all still possible winners. Now all three voters have an opportunity move, which is to go back to their previous vote. If the scheduler keeps selecting $\{1,2,3\}$ then we would have a cycle going forever. In this example the scheduler picks $I'=\{1\}$ (step~2), which ends the chunk. Now there is only one voter with a strategic move (voter~2), so the next chunk has just one step. After step~3 no voter has a strategic move, so the game is in equilibrium.}
\end{figure}
%\end{document}
\end{proof}

\begin{rproposition}{th:truth_bias}
Suppose that each voter $i$ is of type $L(r,k_i)$ or $T(r,k_i)$, where $k_i>r$. Then a voting equilibrium exists. %Assume further that $|\ol H_{r+1}(\vec q)|>1$, 
Moreover, in an iterative setting where voters start from the truthful state, they will always converge to an equilibrium in at most $3nm$ steps.
\end{rproposition}
\begin{proof}
We first prove for truth-biased voters.
Consider a $T(r,k_i)$ voter $i$. A truth-bias move can only occur when $i$ has no strategic moves.
By Lemma~\ref{lemma:two_winners}, $|\ol H_{r+1}(\vec s)|>1$, and by Lemma~\ref{lemma:threshold} this means there are at least two possible winners for $i$. If some of them locally dominate $a_i$, then $i$ would have a strategic move. Thus if $i$ makes a truth-bias move he is in one of two situations:
(a) when $i$ is already voting for his most preferred candidate in $\ol H_{r+1}(\vec s)$; (b) $a_i \notin \ol H_{r+1}(\vec s)$, but none of the possible winners $S_i(\vec a,r)$-dominates $a_i$. Denote such moves by type-a and type-b, respectively. 

We first argue that there are no type-a truth-bias moves. Indeed, we have $a_i\in \ol H_{r+1}(\vec s) \subseteq \ol H_{k_i}(\vec s)$. If $a_i\neq f(\vec a)$, there is a state $\vec s^*\in S_i(\vec a,k_i)$ where $a_i$ wins if getting the vote of $i$, and otherwise $f(\vec a)$ wins. If $a_i=f(\vec a)$, then consider some other $b\in \ol H_{r+1}(\vec s)$ (such a $b$ exists according to Lemma~\ref{lemma:two_winners}). Then there is a state $\vec s^*\in S_i(\vec a,k_i)$ where $a_i=f(\vec a)$ wins if getting the vote of $i$, and otherwise $b$ wins.
In either case, $a_i$ $S_i(\vec a,k_i)$-beats $q_i$, and thus there is no type-a move. Note that this is where we apply the assumption that $k_i>r$.

Type-b moves are possible, and thus we need to show that invariants (3) and (4) in the proof of Theorem~\ref{th:basic} still hold. I.e., that a truth-bias move cannot cause the winner to lose votes, and cannot expand the possible winners set. Clearly the winner cannot lose score, since $a_i\notin \ol H_{r+1}(\vec a)$. 
It is left to show that $a'_i=q_i$ cannot become a possible winner. Let $N'_{q_i}\subseteq N$ be all core supporters of $q_i$ that vote strategically in $\vec a$. For all $j\in N'_{q_i}$ (including $i$), when $j$ deserted $q_i$ then $q_i$ was no longer a possible winner for $j$ (i.e., its score was at least $r+2$ below the winner). Since then, the score of the winner does not decrease. Thus even if all of $N'_{q_i}$ return to $q_i$, the gap between $q_i$ and $f(\vec a)$ would be at least $r+2$. 
%\end{proof}

For the bound on the number of steps, denote by $R_i,K_i\in \mathbb N$ the number of strategic moves and Truth-bias moves of voter $i$, respectively. Observe that a strategic move can only occur after the set of possible winners shrinks---unless it is the first move or it follows a truth-bias move. Thus $R_i \leq (K_i+1) + (m-1) \leq K_i+m$. A truth-bias can only come after the set of possible winner shrinks as well, thus $K_i\leq m-1$. In total, $R_i+K_i\leq 2K_i +m \leq 3m$. So all voters together can do at most $3nm$ moves.

\medskip
We now turn to prove that convergence still holds when adding lazy-biased voters.
%\begin{proof}
Assume first that $|\ol H_{r+1}(\vec q)|>1$.
Then the proof is essentially the same, replacing $q_i$ with $\bot$ for lazy-biased voters. 

However the proof breaks when $|\ol H_{r+1}(\vec q)|=1$, as ``lazy'' supporters of the winner may abstain. Thus this case requires a special treatment. We will show that every lazy voter plays at most once, and then the game ends. 
%We can still show that a PNE exists, but not that there is a quick convergence in any order of players. Let $a^*=f(\vec q)$ be the truthful winner. %We fix the order of voters so that core supporters of $a^*$ play first, ordered by non-increasing values of $k_i$. Then we argue that after a single decision made by each voter of $a^*$, the game ends. 

Let $k^t$ be the gap between $a^*=f(\vec q)$ and its closest runner-up (i.e., $k^t = \min\{k>0 : H_{k}(\vec s^t)\neq \emptyset\}$). By our assumption $k^0> r+1$. At time $t$, voter $i=i_t$ will choose to abstain if and only if $k_{i}\leq k^t-2$. To see why, if $k_{i} \geq k^t-1$, then $a^*$ $S_i(\vec a^t,k_i)$-beats $\bot$ (by removing $k_i$ votes from $a^*$, $i$ becomes pivotal between $a^*$ and the runner-up).
However, a voter who already abstains, will choose to enter only if $k^t \leq r+1$ (or $k^t \leq r$, for supporters of $a^*$). 
 Thus $k^t$ may either increase or decreases by $1$ in every step. However in any case $k^t$ cannot go below $r+2$. We can show this by induction: the base case follows since $|\ol H_{r+1}(\vec q)|=1$. In any later step, $k^{t+1} < k^t$ if and only if a supporter $i$ of $a^*$ abstains, which occurs iff $r<k_{i}\leq k^t-2$. Thus $k^t\geq r+3$, and $k^{t+1}\geq r+2$. 

Finally, after each voter made his decision whether to abstain or not, there are no strategic moves, as $k^t\geq r+2$ entails $|\ol H_{r+1}(\vec a^t)|=1$. In this case, the game will converge after at most $n$ moves.
\end{proof}

\section{Simulations}
\subsection{Distributions of preference profiles}
\label{sec:distributions}
\begin{itemize}
\item {\bf Uniform:} Also known as the \emph{impartial culture} distribution, this is the simplest distribution to study. While people's votes are rarely distributed at random, the uniform distribution allows more confidence that our results are not particular and specific to the distributions analyzed, and is thus often used in simulations of voting~\cite{nurmi1992assessment}.

\item {\bf Single-peaked:} This distribution assigns each candidate a point on the interval $[0,1]$, and each voter is randomly assigned a point on the interval, which defines its preferences --- it prefers candidates closer to its point. This distribution has been long used in sociological and political research (as resembling the common right-left political axis) \cite{Ked14}, but has also been widely examined in game theoretic scenarios. A particular interesting property is that for single-peaked preferences can be aggregated using strategy-proof mechanisms. The most prominent such mechanism is the \emph{median vote}~\cite{Spr91}.\footnote{We also tried simulations with \emph{single-dipped} preferences, where the voter's most preferred candidates are at the extreme. However, in such profiles the truthful vote has only two candidates with positive support (the extremes), and no voter ever has an incentive to move.}

%\item{\bf Single-dipped:} This distribution is similar to the single-peaked one%, as it also assigns a location on the interval $[0,1]$ to all candidates, and a voter is assigned a random point on the interval as well. 
%However, in this distribution, the voter prefers candidates which are located as far as possible from its assigned point. As for single-peaked preferences, strategy-proof mechanisms have been designed for single-dipped preferences~\cite{KPS97,BBM09}.

\item {\bf Polya-Eggenberger urn:} This model was developed and used to model the grouping of much of society to major homogenous groups~\cite{Ber85,Wal10,RS12}. In a $k$-urn model, $k$ preference orders are chosen, and an urn is built to let voters choose preference orders from it. Each of the $k$ chosen preferences gets $\frac{1}{k+1}$ of the preference orders in the urn, with remaining $\frac{1}{k+1}$ of the urn filled by all non-selected preference orders. Preferences chosen from this urn have significant likelihood to be of the $k$ main groups. In this work, we used the $2$-urn and $3$-urn model.

\item{\bf Riffle:} In a riffle model we get preferences of each voter by interleaving two separate preference orders on subsets of candidates in an independent manner. Huang and Guestrin~\shortcite{HG09} showed real-world elections which resembled this distribution.

\item{\bf Placket-Luce:} In the Placket-Luce model each candidate has an intrinsic cardinal value in the interval $[0,1]$ (the ``ground truth''). Each vote is then sampled from a particular distribution which adds noise to the true ranking~\cite{SPX12}.

%This is a Random Utility Model which views votes, in a sense, as ``noise'' disguising a basic ground true ranking. The distribution assigns each candidate a point in the interval $[0,1]$ (and hence, implicitly, a relation), and the preference order is determined by drawing a value for each candidate according to the Placket-Luce distribution, in which the drawn value for each candidate is in expectation (or mode) it's own value.

\end{itemize}

\paragraph{German election data}
We used data from three polls based on the German National Election Study, from 1969, 1972, and 1976.\footnote{Since these were election to the parliament and not a single-winner elections, it is impossible to compare the simulation results to the actual outcome.} 
The German election datasets had three candidates: The Christian Democratic Union (CDU), the Social Democratic Party (SDP), and the Free Democratic Party (FDP); we used $n=100$ voters whose preferences (from the six possible orders) were taken from~\cite{regenwetter2006behavioral}.

\paragraph{PrefLib} We used all 225 complete preference profiles available from PrefLib.org. Most instances have 3 or 4 candidates (over 10 in some), and several hundred voters. 

\subsection{Methods}
\label{sec:methods}
We simulated voting in an iterative setting, where voters start from a particular state, and then iteratively make strategic moves until convergence. 
We constructed a simulator that enables us to control the following features of the simulation. First, determine the parameters of the preference profile:
\begin{itemize}
	\item Number of voters. We used $n\in\{10,20,50\}$ (also $n=100$ in a few simulations).
	\item Number of candidates. We used $m\in\{3,\ldots,8\}$.
	\item Distribution of preferences. We used all the $6$ distributions described above.
\end{itemize}

Ignoring the $n=100$ simulations, this defines $3\times 6 \times 6=108$ distinct distributions. From each such distribution we generated $200$ preference profiles, to a total of $21,600$ profiles.
Then, we determine the parameters of the strategic model:
\begin{itemize}
	\item Distance metric used for accessibility relation. We used the $\ell_1$ norm and the multiplicative distance.
	\item The uncertainty parameter $r$ which determines the radius of local dominance. For the $\ell_1$ distance, we varied $r$ in $\{0,1,\ldots,15\}$; for multiplicative we varied $r$ in $\{\frac{1}{n},\ldots,\frac{4}{n},0.3,0.5,0.7\}$. We also used $r=m$ as a baseline value, where there is no strategic behavior (the outcome is the truthful Plurality outcome).
	\item For voters with truth-bias or lazy-bias, the parameter $k$ was varied as well from $r+1$ (or $r+\frac{1}{n}$ for multiplicative) to $2r$.
	\item The simulator allows the creation of voters of different types. However in most of our simulations all voters had the same type (but different preferences).
\end{itemize}
Finally, there are setting of the simulation itself:
\begin{itemize}
	\item The initial voting profile.  In most simulations this was the voters' truthful profile. % We also had a condition of where the initial voting profile was selected uniformly at random at each simulation. % (thus for every preference profile we had 100 different initial states). % apart a set simulation to see how the model handles non-truthful positions, which where done with 20 voters and 4 or 8 candidates. %\rmr{edit after wave 3}
	\item The scheduler. In most simulations we used a singleton scheduler, which selects the next voter uniformly at random. We also used a group scheduler, which randomly selects a subset of at most $n/2$ voters and lets all of them make independent strategic moves (if they have such). 
\end{itemize}
%un not only a variety of voter and candidate numbers, but also allows us to determine what distance function to use (additive, i.e., $\ell_{1}$ norm; or multiplicative), allows us to determine the $r$ and $k$ values for each voter, and allows us to run 3 types of voters: $r$-voters, truthful voters and lazy voters (both of which have both $r$ and $k$ values). We also can produce preference orders of whatever distribution we require.

\paragraph{Basic simulations}
For each profile, we ran simulations with strategic voters using the $\ell_1$ metric, ranging the value of $r$. For each combination of preference profiles and voters' type, we repeated the simulation 100 times with a random singleton scheduler, recording all the equilibria that were attained from this profile. In this setting, we know by Theorem~\ref{th:basic} that every simulation must converge.
For profiles with $m\leq 5$, we conducted the same procedure with the multiplicative metric as well.

%Each profile was used for a set of runs: the profile was assigned to voters of a certain type and a certain distance function (the same type of voter and distance function for all, but all combinations were done), and assigned values of $r$ and $k$ (for additive $r$, values started with 1 and extended with jumps to $0.7n$; with $k$ taking multiple values from $r+1$ to $2r$; for multiplicative $r$ started at $\frac{1}{n}$ and continued to be sampled until it was $0.7n$ \rmr{unclear}, with $k$ taking values from $r+\frac{1}{n}$ to $2r$). Since iterative voting is a non-deterministic order of voters, once a scenario is fully defined (voter type, distance function, $r$, $k$), it is run multiple times in order to verify we cover most of its run cases. 

%Following these scenarios, we decided to focus on $r$-voters with additive distance function (see reasoning below). As we no longer had $k$ values to compute, we could increase the range of voters to 6,7 and 8 in addition to previous work. We also sampled every $r\in\{0,1,\ldots,15\}$.

\paragraph{Concurrent voting}
We repeated a subset of the simulations above with a group scheduler. We used all the above distributions, with $n=50,m=5$, and all values of $r$. We then compared the results to the corresponding results under a singleton scheduler.
%To conclude we ran more complex settings: First, we added a non-singular schedular, which randomly decides on a number (up to $\frac{n}{2}$) for how many players will play concurrently in a single stage, i.e., each playing without knowing what the others are doing. We ran this non-deterministic scheduler in a setting of 5 candidates and 50 voters on the same profiles, and subsequently compared their output to our previous simulations. Finally, instead of giving all our voters in each scenario an identical parameter $r$, we allowed each to randomize and chose a different one. We ran the same profile with different $r$ values in a setting of 5 candidates and 50 voters, and once again, compared the output to the prior simulations.

\paragraph{Random initial state}
We repeated our simulations where the initial state is sampled uniformly at random from all voting profiles (regardless of the preference profile). Thus for each preference profile we ran 100 simulations, each with a different initial state and a different scheduler. 

\paragraph{Diverse types} Since actual societies are likely to contain voters of different types, we repeated the same simulations with a diverse population. In these simulations the value $r_i$ for each voter was sampled uniformly from $\{0,1,\ldots,n/m\}$.

Note that in the three batches of simulations above, we had no formal guarantee that the game converges. However, in practice all simulations converged to an equilibrium.

\paragraph{Truth-biased and lazy-biased societies}
For $m\leq 5$ we repeated the simulations while varying the values of $k$ in addition to $r$. In that process we simulated truth-biased and lazy-biased societies.

\paragraph{Non-truthful starting profiles}
For $n=20, m=4$, we examined the possibility of changing the the starting position of players to a random one. A different starting profile was assigned every run of a profile (to enable examination how many different winners are reached using our dynamic).

%\rmr{Please order results as follows: 1) effect of metric and of $r$ on ``amount'' of strategic behavior (number of steps, dispersion of states and winners, clustering by n?); 2) Robustness (low dispersion of winners in general, same results with group schedulers); 3) Duverger's Law ; 4) Quality of winner; 5) Diversity simulations (see my .docx).
%We will list all of the key finding in the main text, plus 2-3 nice images. the rest will go to the appendix.} \olev{We must mention that k turned out immaterial, and that mult is like additive. We can't mention this just in the appendix.}

\paragraph{Observed variables}
For every generated profile, we measured the following variables (all averaged over 100 simulations with a random singleton scheduler). We then averaged again over all 200 generated profiles of a given distribution.
\rmr{Omer, can we make the headers in the xlsx consistent with these?}
\begin{description}
	\item[NumStep] The number of steps from the initial (truthful) profile to convergence.
	\item[NumStates] The number of distinct equilibrium outcomes (for the same preference profile), in terms of voting profiles.
	\item[NumWinners] The number of distinct equilibrium outcomes (for the same preference profile), in terms of winner's identity.
	\item[WinnerConsistency] The maximal fraction of simulations (for the same preference profile) that ended with the same winner.  
	\item[PluralityAgreement] The fraction of simulations where the winner was the original Plurality winner.
	\item[BordaAgreement] The fraction of simulations where the winner was the Borda winner.
 \item[CoplandAgreement] The fraction of simulations where the winner was the Copland winner.
\item[MaximinAgreement] The fraction of simulations where the winner was the Maximin winner.
\item[CondorcetAgreement] The fraction of simulations where the winner was the Condorcet winner, when one exists (not counted otherwise).
\item[SocialWelfare] The relative rank of the winner, according to its borda score (lower is better). Equivalently: the complement of the average social welfare of voters, assuming Borda utilities.
\item[Gap1-2] The ratio between the score of the winner and the score of the runner-up ($s(c_1)/s(c_2)$, where $c_i = \argmax_{c\neq c_j, j<i} s(c)$.) 
\item[Gap2-3] The ratio between the score of the second and the third candidates ($s(c_2)/s(c_3)$).
\item[TotalDuverger] Fraction of simulations where at most two candidates received votes ($1$ iff $s(c_3)=0$).
\item[RelativeDuverger] The fraction of votes for the two leading candidates ($(s(c_1)+s(c_2))/n$).
\end{description}
For the Placket-Luce distribution, we also measured \emph{WinnerGroundRank}, which is the rank of the winner according to the ground truth used to generate the profile (between $0$ and $m-1$). 

%
%\rmr{use larger font inside figure. names for axes. remove m=3,4 (too cluttered). smaller figure.} 
%\begin{figure}[ht]
%\begin{center}
%\includegraphics[scale=0.25]{peakedStateNum.pdf}
 %\end{center}
%\caption {The single-peaked distribution average number of states. With 50 voters $r$ peaks at $3$, and strategic behavior wanes off completely by $r=7$.\label{fig:peakedStateNum}}
%\end{figure}

\begin{figure}[p]
\begin{center}
\includegraphics[scale=0.3]{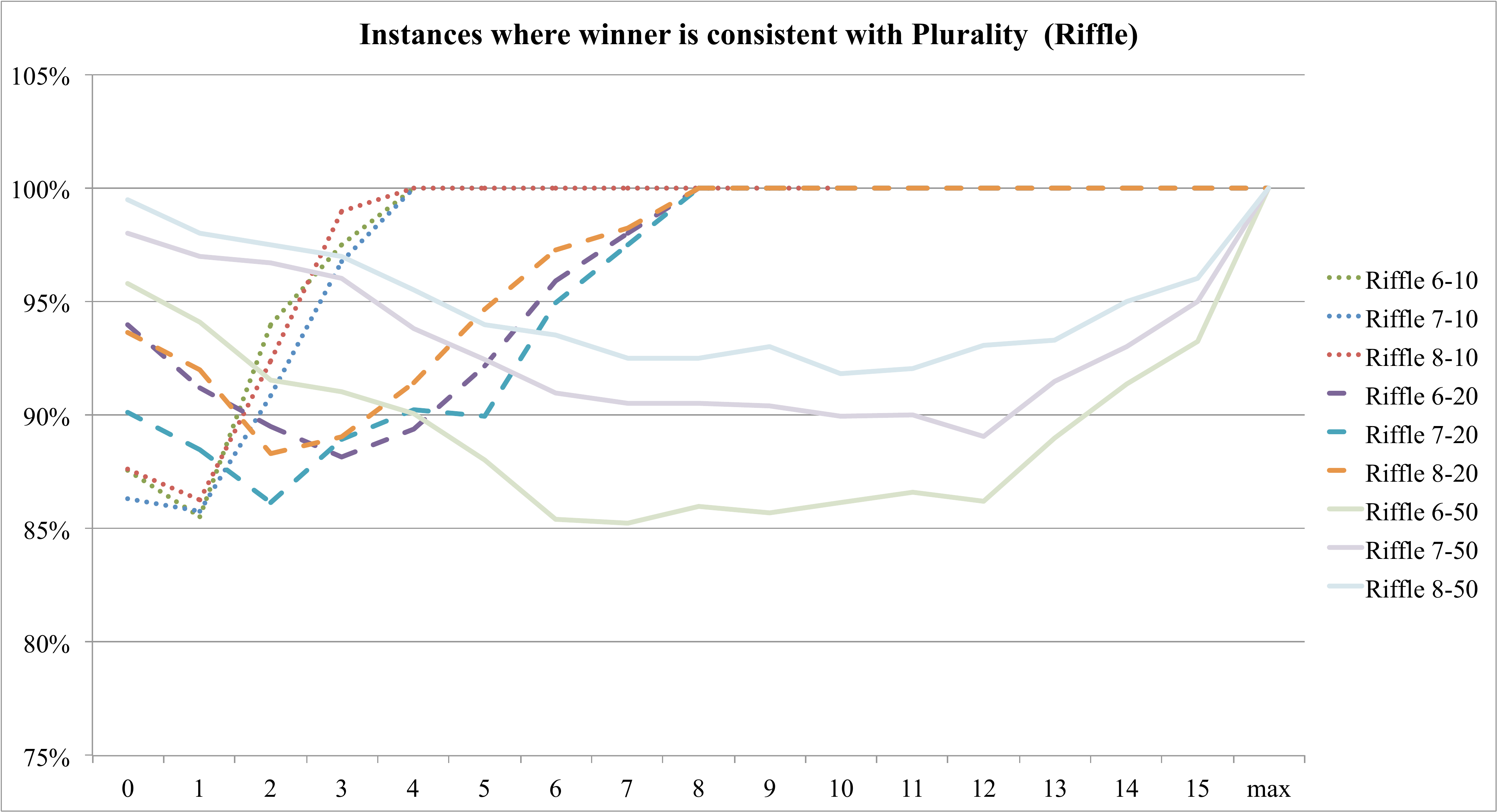}
 \end{center}
\caption {Ratio of games where the winner was the same winner as plurality for Riffle distribution simulation (\emph{PluralityAgreement})\label{fig:riffle_plurality}}
\end{figure}

\begin{figure}[p]
\begin{center}
\includegraphics[scale=0.4]{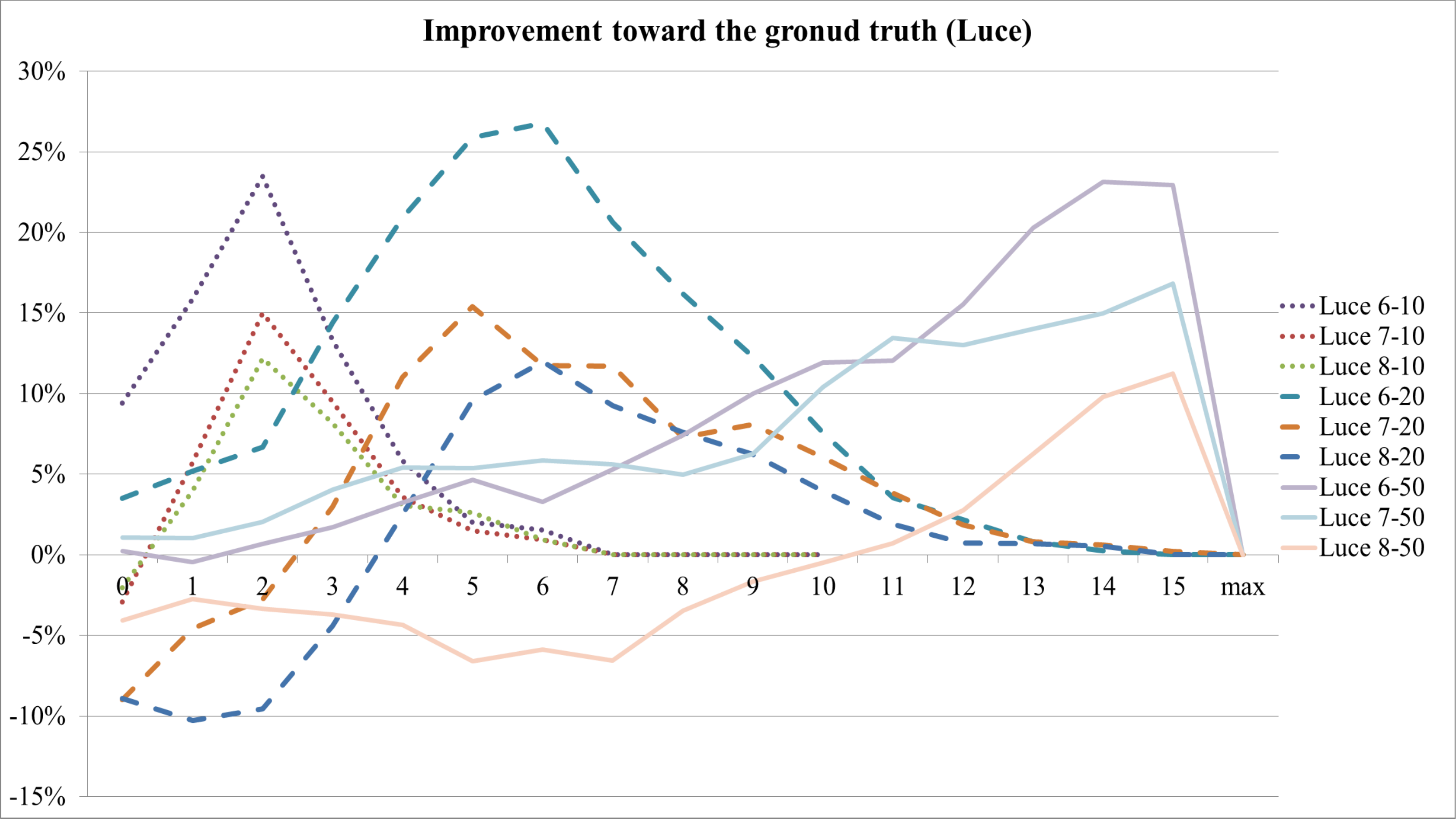}
 \end{center}
\caption{Improvement in the objective quality of the winner (\emph{WinnerGroundRank}), according to the ground truth. Note that when there is some strategic interaction but not as much as in peak $r$, winner quality might \emph{decrease}.\label{fig:Luce_ground}}
\end{figure}

\begin{figure}[p]
\begin{center}
\includegraphics[scale=0.4]{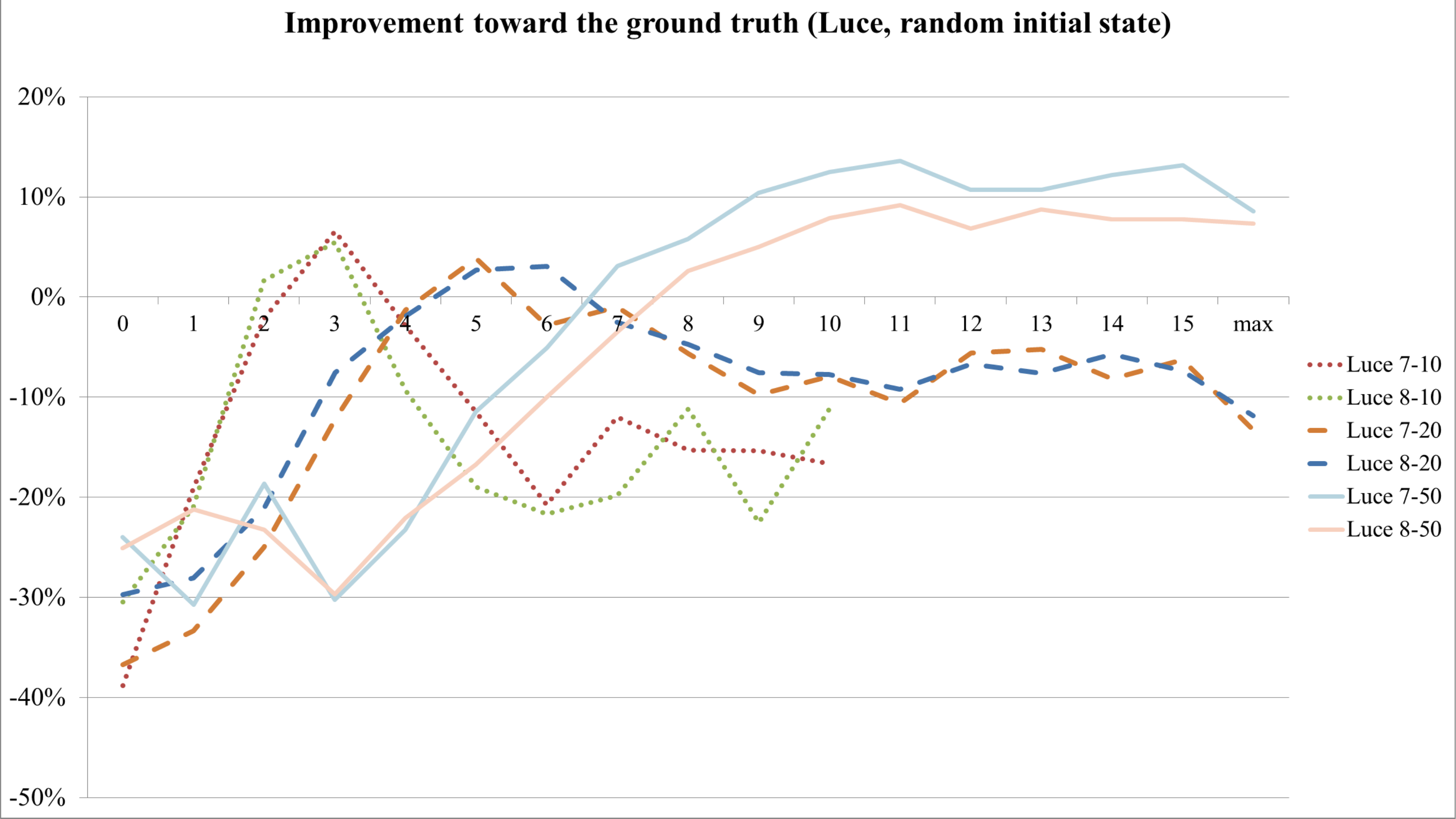}
 \end{center}
\caption{Same as Figure~\ref{fig:Luce_ground}, but with random initial state. We can see strategic behavior (around peak $r$) yields better outcome on average than truthful voting, even when starting from a random state.\label{fig:Luce_ground_random}}
\end{figure}

\begin{figure}[p]
\begin{center}
\includegraphics[scale=0.3]{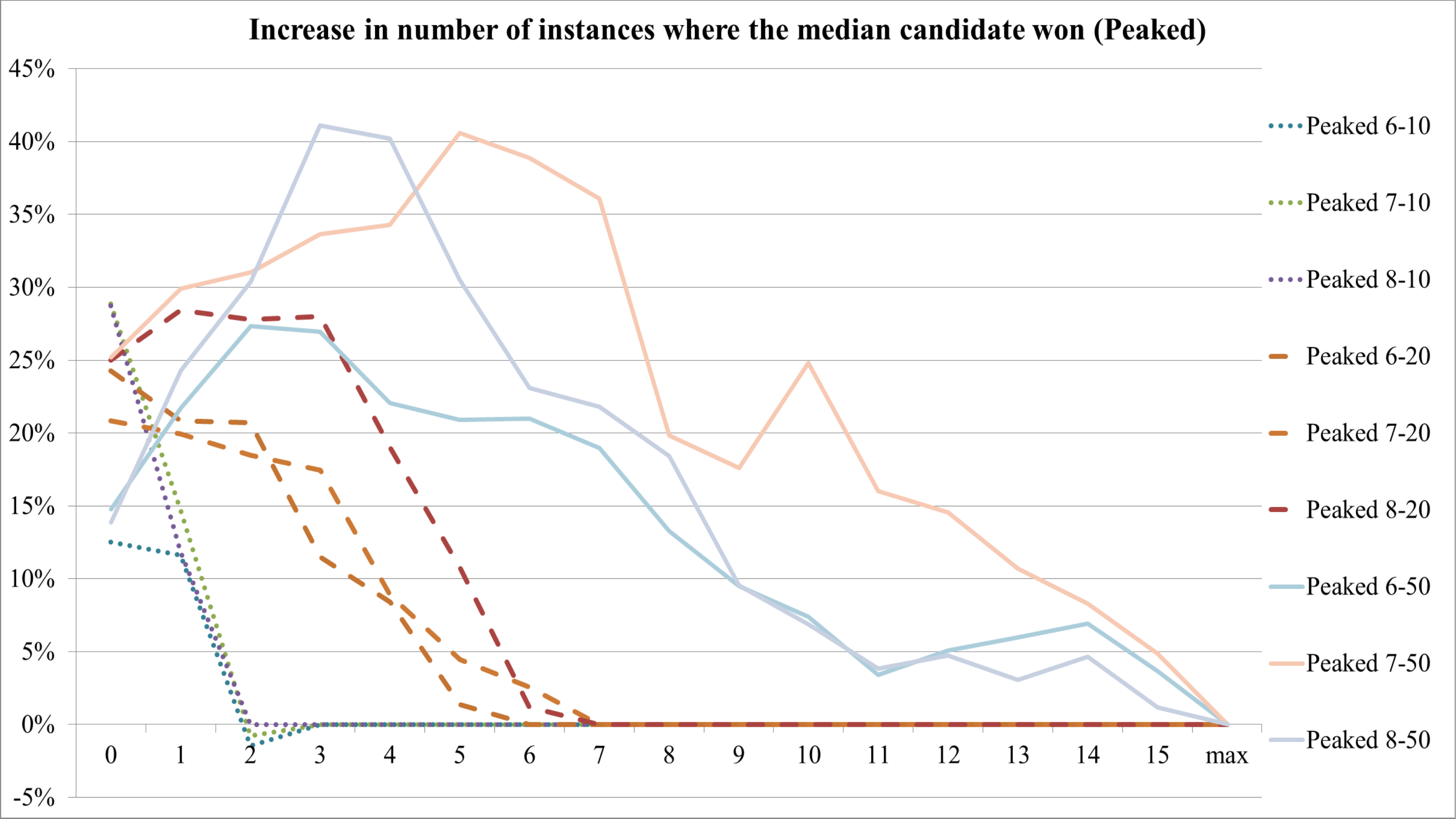}
 \end{center}
\caption {The fraction of simulations in the single-peaked distribution, where the winner is the median candidate (\emph{CondorcetAgreement}), relatively to the truthful baseline.\label{fig:peaked_median}}
\end{figure}

\begin{figure}[p]
\begin{center}
\includegraphics[scale=0.3]{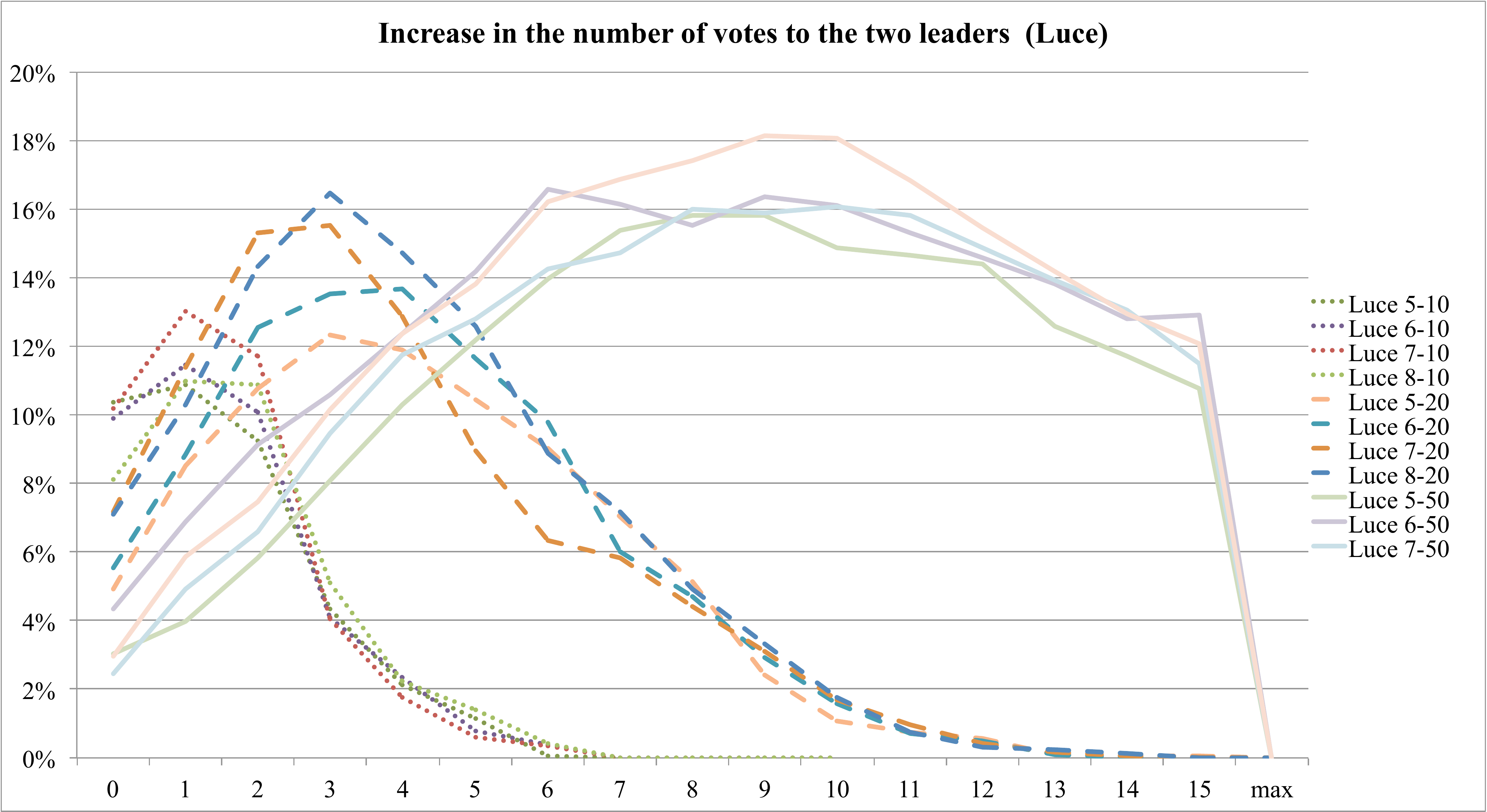}
 \end{center}
\caption {The Difference between \emph{RelativeDuverger} in the equilibrium outcome (the fraction of voters voting for the two leaders), and the truthful baseline.\label{fig:luce_duv}}
\end{figure}

\begin{figure}[t]
\begin{center}
\includegraphics[scale=0.35]{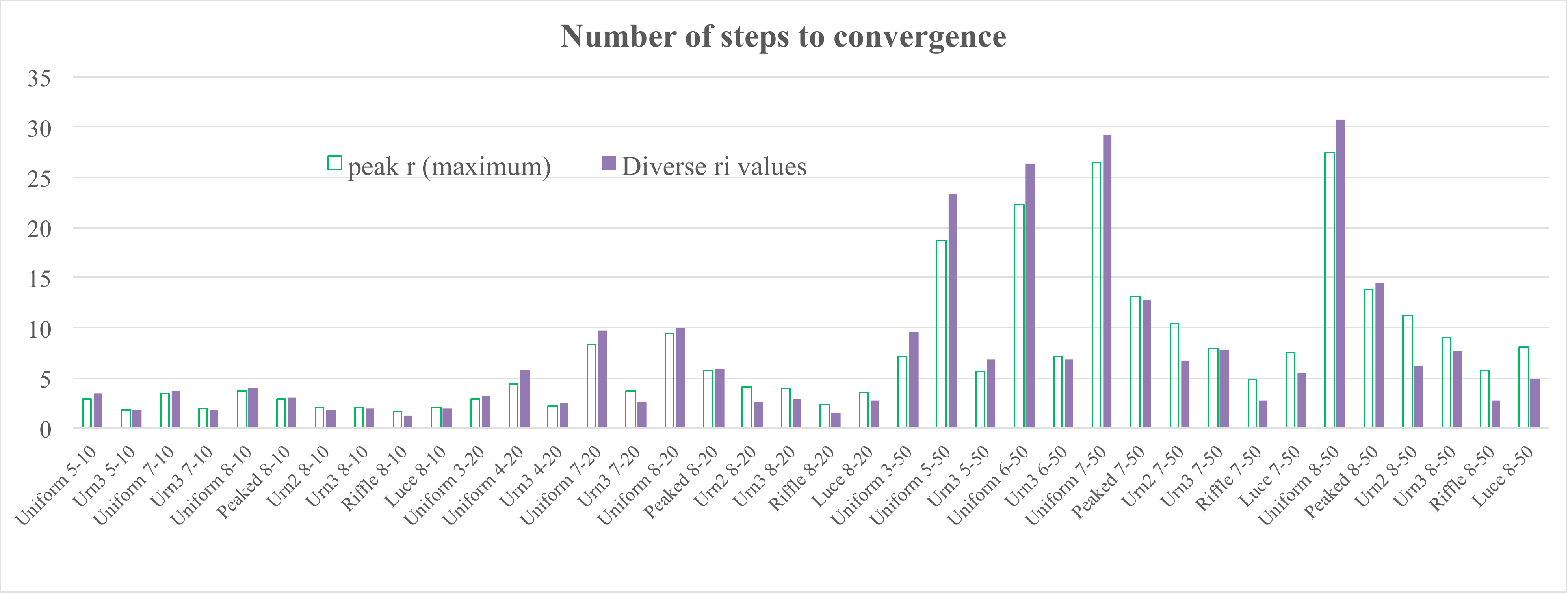}
 \end{center}
\caption {Number of steps to convergence (\emph{NumStep}).\label{fig:diverseStep}}
\end{figure}

\newpage
\begin{small}
\bibliography{plurality.scw}

\begin{thebibliography}{55}
\expandafter\ifx\csname natexlab\endcsname\relax\def\natexlab#1{#1}\fi
\expandafter\ifx\csname url\endcsname\relax
  \def\url#1{\texttt{#1}}\fi
\expandafter\ifx\csname urlprefix\endcsname\relax\def\urlprefix{URL }\fi

\bibitem[{Abramson et~al.(1992)Abramson, Aldrich, Paolino, and
  Rohde}]{abramson1992sophisticated}
Abramson, P.~R., Aldrich, J.~H., Paolino, P., Rohde, D.~W., 1992.
  ``sophisticated'' voting in the 1988 presidential primaries. American
  Political Science Review 86~(01), 55--69.

\bibitem[{Aldrich(1993)}]{aldrich1993rational}
Aldrich, J.~H., 1993. Rational choice and turnout. American Journal of
  Political Science, 246--278.

\bibitem[{Alvarez and Nagler(2000)}]{alvarez2000new}
Alvarez, R.~M., Nagler, J., 2000. A new approach for modelling strategic voting
  in multiparty elections. British Journal of Political Science 30~(1), 57--75.

\bibitem[{Apt and Simon(2012)}]{apt2012classification}
Apt, K.~R., Simon, S., 2012. A classification of weakly acyclic games. In:
  SAGT'12. pp. 1--12.

\bibitem[{Aumann(1995)}]{aumann1995backward}
Aumann, R.~J., 1995. Backward induction and common knowledge of rationality.
  Games and Economic Behavior 8~(1), 6--19.

\bibitem[{Aumann(1999)}]{aumann1999interactive}
Aumann, R.~J., 1999. Interactive epistemology i: knowledge. International
  Journal of Game Theory 28~(3), 263--300.

\bibitem[{Berg(1985)}]{Ber85}
Berg, S., 1985. Paradox of voting under an urn model: The effect of
  homogeneity. Public Choice 47~(2), 377--387.

\bibitem[{Bicchieri and Antonelli(1995)}]{Bicchieri95}
Bicchieri, C., Antonelli, G.~A., 1995. Game-theoretic axioms for local
  rationality and bounded knowledge. Journal of Logic, Language, and
  Information 4~(2), pp. 145--167.

\bibitem[{Brock and Durlauf(2001)}]{brock2001discrete}
Brock, W.~A., Durlauf, S.~N., 2001. Discrete choice with social interactions.
  The Review of Economic Studies 68~(2), 235--260.

\bibitem[{Calvert(1985)}]{calvert1985robustness}
Calvert, R.~L., 1985. Robustness of the multidimensional voting model:
  Candidate motivations, uncertainty, and convergence. American Journal of
  Political Science, 69--95.

\bibitem[{Chopra et~al.(2004)Chopra, Pacuit, and Parikh}]{Chopra:04}
Chopra, S., Pacuit, E., Parikh, R., 2004. Knowledge-theoretic properties of
  strategic voting. Presented in JELIA-04, Lisbon, Portugal.

\bibitem[{Conitzer et~al.(2011)Conitzer, Walsh, and
  Xia}]{conitzer2011dominating}
Conitzer, V., Walsh, T., Xia, L., 2011. Dominating manipulations in voting with
  partial information. In: AAAI. Vol.~11. pp. 638--643.

\bibitem[{Dekel and Piccione(2000)}]{dekel2000sequential}
Dekel, E., Piccione, M., 2000. Sequential voting procedures in symmetric binary
  elections. Journal of Political Economy 108~(1), 34--55.

\bibitem[{Desmedt and Elkind(2010)}]{desmedt2010equilibria}
Desmedt, Y., Elkind, E., 2010. Equilibria of plurality voting with abstentions.
  In: ACM-EC'10. pp. 347--356.

\bibitem[{Dhillon and Lockwood(2004)}]{dhillon04}
Dhillon, A., Lockwood, B., 2004. When are plurality rule voting games
  dominance-solvable? Games and Economic Behavior 46, 55--75.

\bibitem[{Downs(1957)}]{downs1957economic}
Downs, A., 1957. An economic theory of democracy.

\bibitem[{Dutta and Sen(2012)}]{dutta2012nash}
Dutta, B., Sen, A., 2012. Nash implementation with partially honest
  individuals. Games and Economic Behavior 74~(1), 154--169.

\bibitem[{Duverger(1954)}]{duverger1954political}
Duverger, M., 1954. Political Parties: Their Organization and Activity in the
  Modern State. New York: John Wiley. Y.

\bibitem[{Edlin et~al.(2007)Edlin, Gelman, and Kaplan}]{edlin2007voting}
Edlin, A.~S., Gelman, A., Kaplan, N., 2007. Voting as a rational choice: Why
  and how people vote to improve the well-being of others. Rationality and
  society 1.

\bibitem[{Falik et~al.(2012)Falik, Meir, and Tennenholtz}]{falik2012coalitions}
Falik, D., Meir, R., Tennenholtz, M., 2012. On coalitions and stable winners in
  plurality. In: WINE'12. pp. 256--269.

\bibitem[{Farquharson(1969)}]{Farquharson:69}
Farquharson, R., 1969. Theory of Voting. Yale Uni. Press.

\bibitem[{Feddersen et~al.(1990)Feddersen, Sened, and Wright}]{feddersen90}
Feddersen, T.~J., Sened, I., Wright, S.~G., 1990. Rational voting and candidate
  entry under plurality rule. American Journal of Political Science 34~(4),
  1005--1016.

\bibitem[{Felsenthal et~al.(1988)Felsenthal, Rapoport, and
  Maoz}]{felsenthal1988tacit}
Felsenthal, D.~S., Rapoport, A., Maoz, Z., 1988. Tacit co-operation in
  three-alternative non-cooperative voting games: a new model of sophisticated
  behaviour under the plurality procedure. Electoral Studies 7~(2), 143--161.

\bibitem[{Ferejohn and Fiorina(1974)}]{ferejohn1974paradox}
Ferejohn, J.~A., Fiorina, M.~P., 1974. The paradox of not voting: A decision
  theoretic analysis. The American political science review, 525--536.

\bibitem[{Fisher(2004)}]{fisher2004definition}
Fisher, S.~D., 2004. Definition and measurement of tactical voting: the role of
  rational choice. British Journal of Political Science 34~(1), 152--166.

\bibitem[{Grandi et~al.(2013)Grandi, Loreggia, Rossi, Venable, and
  Walsh}]{grandi2013restricted}
Grandi, U., Loreggia, A., Rossi, F., Venable, K.~B., Walsh, T., 2013.
  Restricted manipulation in iterative voting: Condorcet efficiency and borda
  score. In: ADT'13. pp. 181--192.

\bibitem[{Huang and Guestrin(2009)}]{HG09}
Huang, J., Guestrin, C., December 2009. Riffled independence for ranked data.
  In: NIPS'09.

\bibitem[{Johnson et~al.(2002)Johnson, Camerer, Sen, and
  Rymon}]{johnson2002detecting}
Johnson, E.~J., Camerer, C., Sen, S., Rymon, T., 2002. Detecting failures of
  backward induction: Monitoring information search in sequential bargaining.
  Journal of Economic Theory 104~(1), 16--47.

\bibitem[{Kahneman et~al.(1991)Kahneman, Knetsch, and
  Thaler}]{kahneman1991anomalies}
Kahneman, D., Knetsch, J.~L., Thaler, R.~H., 1991. Anomalies: The endowment
  effect, loss aversion, and status quo bias. The journal of economic
  perspectives 5~(1), 193--206.

\bibitem[{Kedar(2014)}]{Ked14}
Kedar, O., January 2014. Voting for Policy, Not Parties. Cambridge Studies in
  Comparative Politics. Cambridge University Press.

\bibitem[{Laslier(2009)}]{laslier2009leader}
Laslier, J.-F., 2009. The leader rule: A model of strategic approval voting in
  a large electorate. Journal of Theoretical Politics 21~(1), 113--136.

\bibitem[{Manski(1993)}]{manski1993identification}
Manski, C.~F., 1993. Identification of endogenous social effects: The
  reflection problem. The review of economic studies 60~(3), 531--542.

\bibitem[{McKelvey and Niemi(1978)}]{McKelvey:78}
McKelvey, R.~D., Niemi, R., 1978. A multistage representation of sophisticated
  voting for binary procedures. Journal of Economic Theory 18, 1--22.

\bibitem[{Meir et~al.(2010)Meir, Polukarov, Rosenschein, and
  Jennings}]{meir2010convergence}
Meir, R., Polukarov, M., Rosenschein, J.~S., Jennings, N.~R., 2010. Convergence
  to equilibria in plurality voting. In: AAAI'10.

\bibitem[{Merrill(1982)}]{merrill1982strategic}
Merrill, S., 1982. Strategic voting in multicandidate elections under
  uncertainty and under risk. In: Power, voting, and voting power. Springer,
  pp. 179--187.

\bibitem[{Messner and Polborn(2002)}]{messner02}
Messner, M., Polborn, M.~K., 2002. Robust political equilibria under plurality
  and runoff rule. Mimeo, Bocconi University.

\bibitem[{Myerson and Weber(1993)}]{myerson93}
Myerson, R.~B., Weber, R.~J., 1993. A theory of voting equilibria. The American
  Political Science Review 87~(1), 102--114.

\bibitem[{Niemi and Frank(1982)}]{niemi1982sophisticated}
Niemi, R.~G., Frank, A.~Q., 1982. Sophisticated voting under the plurality
  procedure. In: Political Equilibrium. Springer, pp. 151--172.

\bibitem[{Nurmi(1992)}]{nurmi1992assessment}
Nurmi, H., 1992. An assessment of voting system simulations. Public Choice
  73~(4), 459--487.

\bibitem[{Obraztsova et~al.(2013)Obraztsova, Markakis, and
  Thompson}]{obraztsova2013plurality}
Obraztsova, S., Markakis, E., Thompson, D.~R., 2013. Plurality voting with
  truth-biased agents. In: SAGT'13. pp. 26--37.

\bibitem[{Owen and Grofman(1984)}]{owen1984vote}
Owen, G., Grofman, B., 1984. To vote or not to vote: The paradox of nonvoting.
  Public Choice 42~(3), 311--325.

\bibitem[{Palfrey and Rosenthal(1983)}]{palfrey1983strategic}
Palfrey, T.~R., Rosenthal, H., 1983. A strategic calculus of voting. Public
  Choice 41~(1), 7--53.

\bibitem[{Regenwetter et~al.(2006)Regenwetter, Grofman, Marley, and
  Tsetlin}]{regenwetter2006behavioral}
Regenwetter, M., Grofman, B., Marley, A., Tsetlin, I., 2006. Behavioral social
  choice. Cambridge: Cambridge University Press.

\bibitem[{Reijngoud and Endriss(2012)}]{reijngoud2012voter}
Reijngoud, A., Endriss, U., 2012. Voter response to iterated poll information.
  In: AAMAS'12. pp. 635--644.

\bibitem[{Reyhani et~al.(2012)Reyhani, Wilson, and
  Khazaei}]{reyhani2012coordination}
Reyhani, R., Wilson, M.~C., Khazaei, J., 2012. Coordination via polling in
  plurality voting games under inertia. In: COMSOC'12.

\bibitem[{Riker and Ordeshook(1968)}]{riker1968theory}
Riker, W.~H., Ordeshook, P.~C., 1968. A theory of the calculus of voting. The
  American Political Science Review, 25--42.

\bibitem[{Rothe and Schend(2012)}]{RS12}
Rothe, J., Schend, L., 2012. Control complexity in bucklin, fallback, and
  plurality voting: An experimental approach. In: Experimental Algorithms.
  Springer Verlag, pp. 356--368.

\bibitem[{Sertel and Sanver(2004)}]{sertel2004strong}
Sertel, M.~R., Sanver, M.~R., 2004. Strong equilibrium outcomes of voting games
  are the generalized condorcet winners. Social Choice and Welfare 22~(2),
  331--347.

\bibitem[{Silberman and Durden(1975)}]{silberman1975rational}
Silberman, J., Durden, G., 1975. The rational behavior theory of voter
  participation. Public Choice 23~(1), 101--108.

\bibitem[{Soufiani et~al.(2012)Soufiani, Parkes, and Xia}]{SPX12}
Soufiani, H.~A., Parkes, D.~C., Xia, L., 2012. Preference elicitation for
  general random utility models. In: NIPS'12. pp. 126--134.

\bibitem[{Sprumont(1991)}]{Spr91}
Sprumont, Y., March 1991. The division problem with single-peaked preferences:
  A characterization of the uniform allocation rule. Econometrica 59~(2),
  509--519.

\bibitem[{Thompson et~al.(2013)Thompson, Lev, Leyton-Brown, and
  Rosenschein}]{thompson2013empirical}
Thompson, D.~R., Lev, O., Leyton-Brown, K., Rosenschein, J., 2013. Empirical
  analysis of plurality election equilibria. In: AAMAS'13. pp. 391--398.

\bibitem[{Tversky and Kahneman(1974)}]{tversky1974judgment}
Tversky, A., Kahneman, D., 1974. Judgment under uncertainty: Heuristics and
  biases. science 185~(4157), 1124--1131.

\bibitem[{van Ditmarsch et~al.(2013)van Ditmarsch, Lang, and
  Saffidine}]{van2013strategic}
van Ditmarsch, H., Lang, J., Saffidine, A., 2013. Strategic voting and the
  logic of knowledge. In: TARK'13. To appear.

\bibitem[{Walsh(2010)}]{Wal10}
Walsh, T., 2010. An empirical study of the manipulability of single
  transferable voting. In: ECAI'10. pp. 257--262.

\end{thebibliography}
\end{small}

\end{document}